\documentclass{article} % For LaTeX2e
\usepackage{iclr2026_conference,times}

\usepackage{amsthm}
\newtheorem{theorem}{Theorem}

\usepackage{amsmath, amsthm,amsfonts,bm}

\def\eqref#1{equation~\ref{#1}}
\def\1{\bm{1}}

\def\vr{{\bm{r}}}

\def\vu{{\bm{u}}}

\DeclareMathAlphabet{\mathsfit}{\encodingdefault}{\sfdefault}{m}{sl}
\SetMathAlphabet{\mathsfit}{bold}{\encodingdefault}{\sfdefault}{bx}{n}

\def\gG{{\mathcal{G}}}

\def\sU{{\mathbb{U}}}

\newcommand{\R}{\mathbb{R}}

\newtheorem{definition}{Definition}
\newtheorem{lemma}{Lemma} 
\newtheorem{remark}{Remark}

\newtheorem{corollary}{Corollary}
\usepackage{booktabs}
\usepackage{graphicx}
\usepackage{subcaption}
\usepackage{multirow}
\usepackage{hyperref}
\usepackage{cleveref}
\usepackage{algorithm}
\usepackage{algpseudocode}
\usepackage{tikz}

\usepackage{url}
\usepackage{enumitem}

\usepackage{xcolor}
\definecolor{green}{HTML}{43A047}
\newcommand{\green}[1]{\textcolor{green}{#1}}
\definecolor{blue}{HTML}{00BCD4}
\newcommand{\blue}[1]{\textcolor{blue}{#1}}
\definecolor{pink}{HTML}{E91E63}
\newcommand{\pink}[1]{\textcolor{pink}{#1}}
\definecolor{purple}{HTML}{3F51B5}
\newcommand{\purple}[1]{\textcolor{purple}{#1}}
\definecolor{orange}{HTML}{FF9800}
\newcommand{\orange}[1]{\textcolor{orange}{#1}}

\title{\centering Socially-Aware Recommender Systems Mitigate Opinion Clusterization}

\author{
\textbf{Lukas Schüepp} \\
ETH Zurich \\
\texttt{lukaschu@ethz.ch}
\And 
\textbf{Carmen Amo Alonso}\\
Stanford University\\
\texttt{camoalon@stanford.edu}
\AND
\textbf{Florian Dörfler}\\
ETH Zurich \\
\texttt{dorfler@ethz.ch}
\And
\textbf{Giulia De Pasquale}\\
TU Eindhoven \\
\texttt{g.de.pasquale@tue.nl}
}

\iclrfinalcopy % Uncomment for camera-ready version, but NOT for submission.
\begin{document}

\maketitle

\begin{abstract}
Recommender systems shape online interactions by matching users with creators’ content to maximize engagement. Creators, in turn, adapt their content to align with users’ preferences and enhance their popularity. At the same time, users’ preferences evolve under the influence of both suggested content from the recommender system and content shared within their social circles. This feedback loop generates a complex interplay between users, creators, and recommender algorithms, which is the key cause of filter bubbles and opinion polarization. We develop a social network-aware recommender system that explicitly accounts for this user-creators feedback interaction and strategically exploits the topology of the user's own social network to promote diversification. Our approach highlights how accounting for and exploiting user's social network in the recommender system design is crucial to mediate filter bubble effects while balancing content diversity with personalization. Provably, opinion clusterization is positively correlated with the influence of recommended content on user opinions.  %By incorporating social network awareness in the recommender system, we showcase strategies to mitigate polarization both locally, within users’ social circles, and globally, across the overall users' population. %The proposed approach declusterizes opinions of agents for which the recommender system has large influence.
Ultimately, the proposed approach shows the power of socially-aware recommender systems in combating opinion polarization and clusterization phenomena.
\end{abstract}

\section{introduction}
The proliferation of streaming services along with e-commerce platforms has created a need for efficient content Recommender Systems (RS) \cite{Amazon,youtube,Netflix}. These systems match users with personalized selections drawn from a massive amount of digital content to enhance user experience and ultimately maximize engagement on the platform \cite{Recent_dev_in_Rec_Sys,compr_review}. On the other hand, by consistently promoting content that aligns with user preferences, RS narrow the diversity of information to which users are exposed, thereby fostering echo chambers and  opinion polarization \cite{rs_polarization}.
Designing RS with high user satisfaction while  countering negative global effects such as opinion clusterization has proven challenging despite numerous research efforts \cite{diversify_any_recommender,set_oriented_pers_rec}. A key question is how to balance individual user satisfaction with prevention of harmful outcomes at the societal level \cite{micro_macro_op_effects}.

%Given the vast scale and heterogeneity of online content, RS require principled and efficient mechanisms to align user preferences with relevant items.
The fundamental principle underlying RS techniques remains consistent: leveraging historical user interaction patterns and profile information to generate personalized recommendations. In this context, collaborative filtering, content-based filtering, and hybrid approaches that integrate both methodologies \cite{Recent_dev_in_Rec_Sys} have emerged as dominant paradigms in the field. However, while personalization enhances user satisfaction at the individual level, it has reinforced negative macroscopic phenomena, such as opinion radicalization \cite{closed_loop_opinion,micro_macro_op_effects,dual_influence}.
% In response to these systemic challenges, a growing body of research has proposed diversification strategies as countermeasures. These include top-k diversification methods \textcolor{red}{citation}, diversified collaborative filtering frameworks \textcolor{red}{citation}, and reinforcement learning approaches that incorporate diversification rewards into their objective functions \textcolor{red}{citation}. Such interventions aim to balance the trade-off between personalization accuracy and recommendation diversity, thereby mitigating the formation of filter bubbles and echo chambers while maintaining user satisfaction. 
While conventional countermeasures may prove effective in isolated experimental settings, they fail to account for the dynamic response of content creators who strategically adapt their material to target potential audiences \cite{dual_influence, dean2024usercreators}. Furthermore, most approaches ignore social interactions users have, limiting their capability to operate at a macroscopic level.

Individuals, in fact, do not act in isolation, but are embedded in social contexts where interactions with others influence their preferences and behaviors \cite{opiniondyn_tutorial,weightedmedian}. 
In this regard, a recommendation paradigm arises, known as Social RS, where the social network graph is exploited together with the user-item rating matrix in order to make more accurate and personalized recommendations \cite{srs_firstpaper,survey_SRS}.
\textcolor{black}{In contrast to Social RS, our focus is not on neighbours' opinions as a predictor of future preferences. Instead, we aim to dynamically exploit the structure of the social network as a control tool to shape the long-term opinion dynamics in a way that mitigates opinion clusterization.}
 
However, a critical yet underexplored dimension of social RS is how the social network can be leveraged to mitigate harmful content without compromising user satisfaction \cite{srs_trustworthy}.
We argue that, opinion polarization being a collective phenomenon, \textcolor{black}{and thus influenced by social interactions}, for the RS to mitigate such undesired effects, enhancing content diversity \textcolor{black}{for a single user} is not enough. \textcolor{black}{The embedding of the user in a social network is a key component in designing a trustworthy RS }\cite{srs_trustworthy,network_aware_rec_sys_via_feedback}. %The authors in \cite{diversification} propose a way to diversify strategies in \emph{top-k} recommendation algorithms by balancing relevance and diversity.

% At the individual level, users affect one another through their social networks, which in turn shape their preferences. In response, content creators adapt their material based on the observed engagement, seeking to gain priority from the recommender algorithm and thereby enhance their visibility \cite{xxx}. This interaction creates complex dynamics that extend far beyond traditional user-item matching strategies.

This paper addresses a fundamental question: How can RS efficiently leverage the social interactions between users to mitigate global clusterization effects, while simultaneously maintaining high levels of user satisfaction? We show that by explicitly modeling the dynamical interplay of users in the social network, creators' content and the RS, it becomes possible to achieve a better balance between user satisfaction and diversity of opinions.

\textbf{Contributions} \quad
We propose a novel \textcolor{black}{theoretical} framework that models the RS landscape with dynamic interaction between users, content creators, and the RS, where users are embedded in social networks. We leverage the user's network structure to develop an optimization-based RS that mitigates opinion clusterization while maintaining high user satisfaction. Unlike previous approaches, we explicitly model the interaction between users, creators, and the RS.
% providing a comprehensive solution to the clusterization problem in modern social media platforms. 
Our main contributions are as follows:
\vspace{-2mm}

\begin{itemize}[leftmargin=1em, itemsep=0pt]
\item We propose a framework that captures the dynamic interplay between socially-connected users, strategic content creators, and the RS, \textcolor{black}{and characterize the relationship between content personalization and opinion polarization}.

\item  \textcolor{black}{We show that a RS that greedily optimizes for user satisfaction leads to opinion cluster formation among creators.}
%\item We extend the optimization based RS into a purely data-driven framework and show how the social network of the users can be used to incorporate the RS into a scalable framework

\item \textcolor{black}{We propose a social-network–aware recommender, $RS(d)$, where the parameter $d$ (number of user hops) controls the trade-off between user satisfaction and the extent of creator clustering, with low d leading to higher satisfaction but more clusters, and high $d$ reducing clusters at the cost of satisfaction.}

%\item We provide a new optimization-based RS that explicitly incorporates social connections to reduce clusterization effects while maintaining a high level of user satisfaction.

\item We test our algorithm experimentally and showcase that when only accounting for engagement maximization, RS  increase opinion clusterization effects over the users population. 
\end{itemize}

\section{Related Work}

\textbf{Negative impacts of RS.}
RS algorithms have been linked to several undesired societal phenomena, including opinion polarization, filter bubbles, and echo chambers. The study \cite{polarization2} shows that link recommendations between highly similar nodes lead to network topologies that exacerbate opinion polarization. The work in \cite{topic_diversification} proposes a method to balance diversity and personalization in recommendation lists, enabling exploration of the full spectrum of users’ interests and demonstrating improved user satisfaction. Similarly, \cite{learning_to_recommend,diversified_recommendations,avoiding_monotony} introduce formal optimization frameworks that incorporate diversity objectives into the recommendation process and propose novel metrics to assess diversification quality beyond traditional accuracy measures. These studies collectively show that diversification can be enhanced without severely compromising accuracy, and in some cases even improving it. However, this body of work adopts a static perspective on recommender algorithms, overlooking the dynamic interactions between users and the RS.

\textbf{Opinion dynamics.}
Opinion dynamics studies how opinions evolve and spread among interacting agents within a social network \cite{opiniondyn_tutorial,weightedmedian,altafini,multi_topic_FJ,FJ}. In our work, we explicitly account for network influence on user preferences and assume that both users’ and creators’ preferences evolve dynamically according to an extended version of the Friedkin–Johnsen model \cite{FJ}. This extension incorporates multiple topics \cite{multi_topic_FJ}, where opinions evolve under the joint influence of connected users, recommended content, and each user’s own prejudice.

\textcolor{black}{
\textbf{Link recommendations.} A line of work analyzes the impact of link recommendation over opinions, when opinions follow a Friedkin-Johnsen dyanmics \cite{RelevanceConflict,polarizationDisagreement,TowardsConsensus,kuhne2025optimizing,ChitraMusco_RS}. In particular, \cite{RelevanceConflict} study the impact of addition of links in a social network relates to the level of conflict in the network.  \cite{TowardsConsensus} study how a centralized planner can alter the structure of a social network to reduce polarization. 
% They also analyze the setting where users'internal opinions are adversarially chosen and relate the planner's problem to the maximization of the network Laplacian's spectral gap. 
The works from \cite{polarizationDisagreement,kuhne2025optimizing,ChitraMusco_RS} study the design of link recommendations to jointly minimize opinion polarization and disagreement subject to some budget constraints.
}

\textbf{Performative prediction.}
Performative predictions support decisions that can influence the outcomes they aim to predict \cite{performative_prediction}. This is the case for RS, whose goal is to predict relevant content for users. The design and evaluation of RS is often approached from a supervised machine learning perspective, treating viewer preferences and the content catalog as static. In practice, however, RS interact with and shape the behavior of both viewers and content creators. This interaction generates a feedback loop between the system and its users \cite{Recent_dev_in_Rec_Sys}.
A recent line of work makes this feedback loop explicit by modeling RS–user interactions and studying how users’ opinions evolve under the influence of recommended content \cite{dean2022preference,yao2024user,closed_loop_opinion,modelling_closed_loop_op_for_social_media,ctrl_strat_for_rec_sys,dual_influence,dean2024accounting,network_aware_rec_sys_via_feedback}.
On the other hand, works such as \cite{RS_gametheory,hron2022modeling,jagadeesan2023supply,divers_content,yao2023bad} focus on dynamic adaptation by creators while treating users as static. The position paper from \cite{dean2024usercreators} proposes a unifying framework that views user–creator–recommender interactions as a dynamical system. 
% In this setting, we draw inspiration from \cite{dual_influence}, where both users and creators co-evolve within a feedback loop.  
\cite{dual_influence} adopt a model-based approach and show that such dynamics lead to opinion polarization, while standard diversity-promoting strategies are insufficient to mitigate it.
In contrast to \cite{dual_influence}, we consider users embedded in a social network, influenced not only by recommended content from creators but also by the content shared from other connected users. We demonstrate that, in this setting, balancing content diversity and personalization counteracts opinion clusterization and polarization.

\textcolor{black}{To better position our paper, the Table in Section \ref{sec:related_work} of the Appendix  provides a schematic summary of the related work and how we compare with it.}

\section{Problem Setup}

We consider a setup where users interact bidirectionally with content creators, mediated by a RS. Similar to previous works \cite{modelling_closed_loop_op_for_social_media, closed_loop_opinion}, opinions are considered as the driving factor behind user preferences. Formally, we consider two sets of agents (modeled as dynamical systems) engaging in $n$ different topics at each time step $t$:
\begin{itemize}
    \item \textbf{Users:}  $\displaystyle \mathcal{U}^t =\{ \displaystyle u _0^t, \displaystyle u _1^t,...,\displaystyle u _{N-1}^t\}$, where $u _i^t \in [-1,1]^n$ represents the opinion of user $i$ at time $t$. The $k$-th entry of $\displaystyle \vu _i^t$ indicates the opinion of user $i$ on item $k$ at time $t$. We define the global user opinion's vector at time $t$ as $\textbf{u}^{t}=\begin{bmatrix} (u_0^t)^\top & \dots & (u_{N-1}^t)^\top\end{bmatrix}^\top$.  %Users interact with each other via a social network, which we model as a graph and detail in the following subsection.    
    \item \textbf{Content creators:} $\displaystyle \mathcal{C}^t =\{ \displaystyle c _0^t, \displaystyle c _1^t,...,\displaystyle c _{M-1}^t\}$, where $ c _j^t \in [-1,1]^n$ represents the opinion of creator $j$ at time $t$. The $k$-th entry of $\displaystyle c _j^t$ is the opinion of creator $i$ on item $k$ at time $t$. We define the global creator opinion's vector at time $t$ as $\textbf{c}^{t}=\begin{bmatrix} (c_0^t)^\top & \dots & (c_{M-1}^t)^\top\end{bmatrix}^\top$.
    %\item \textbf{Recommeder System:} ... matches creator content to users based on their preferences
\end{itemize}

At each timestep $t$, content creators publish material reflecting their current opinion vectors in $\mathcal{C}^t$. We extend the framework introduced by \cite{dual_influence}, and explicitly account for the social network effects in the opinion evolution and overall user behavioral trends. In particular, we assume closed-loop interactions among all agents, creating a dynamic feedback system where: (1) recommended content influences user opinions over time, (2) content creators adapt their content strategies based on audience implicit feedback and (3) users live in a social network where they influence each other over time. While there is a mutual influence between users and creators, the difference between these two classes of agents lies in the nature of their interactions: 

\begin{itemize}
    \item \textbf{User-user interactions:} mediated via the social network. This is, two users influence one another if they are connected directly through their social network.
    \item \textbf{Creator-user interactions:} mediated via a RS. Specifically, the RS presents each user with a subset of content creators. In turn, content creators only receive feedback from the users they reach through the RS.
\end{itemize}

The dynamics resulting from these interactions are modeled as:
\begin{subequations}\label{eqn:dynamics}
\begin{align}
    &\displaystyle \mathbf u^{t+1} = \green{f}(\mathbf u^{t}) + \pink{h^t}(\mathbf c^t), \\
    & \displaystyle \mathbf c^{t+1} = \blue{p^t}(\mathbf u^{t}) + \orange{q}(\mathbf c^t),\\
    & \quad c^t_j \sim \purple{\mathcal R}(u^t_i), &\qquad \forall i=1,\dots,N \quad j=1,\dots,M, \label{eqn:R_stochastic}
\end{align}
\end{subequations}
where user $i$'s decision of engaging with content from creator $j$ at time $t$, is sampled from a probability distribution $\purple{\mathcal{R}}$ defined by the recommendation assigned to user's $i$ by the RS selection at time~$t$. The functional relationships for $\green{f},\pink h,\blue p,\orange q$ and definitions in \eqref{eqn:dynamics} are given by the multi-topic extended Friedkin-Johnsen model \cite{multi_topic_FJ} extended to include the influence of the RS in the same fashion as in \cite{ctrl_strat_for_rec_sys}. Detailed descriptions of the functions are discussed in the following sections. For a complete discussion of the model, we refer to \Cref{app:multi_topic}.

%After identifying the values determining the users opinion transition, the parameters are normalized such that $ B_i + \sum_{j \text{ s.t }(j,i) \in \mathcal{E}} \mA_{i,j} = 1$. 

%Figure 1 illustrates our complete system architecture. In the following subsections, we provide detailed descriptions of each component and their interaction mechanisms.

\begin{figure}[h!] 
    \centering
    \includegraphics[width=0.8\linewidth]{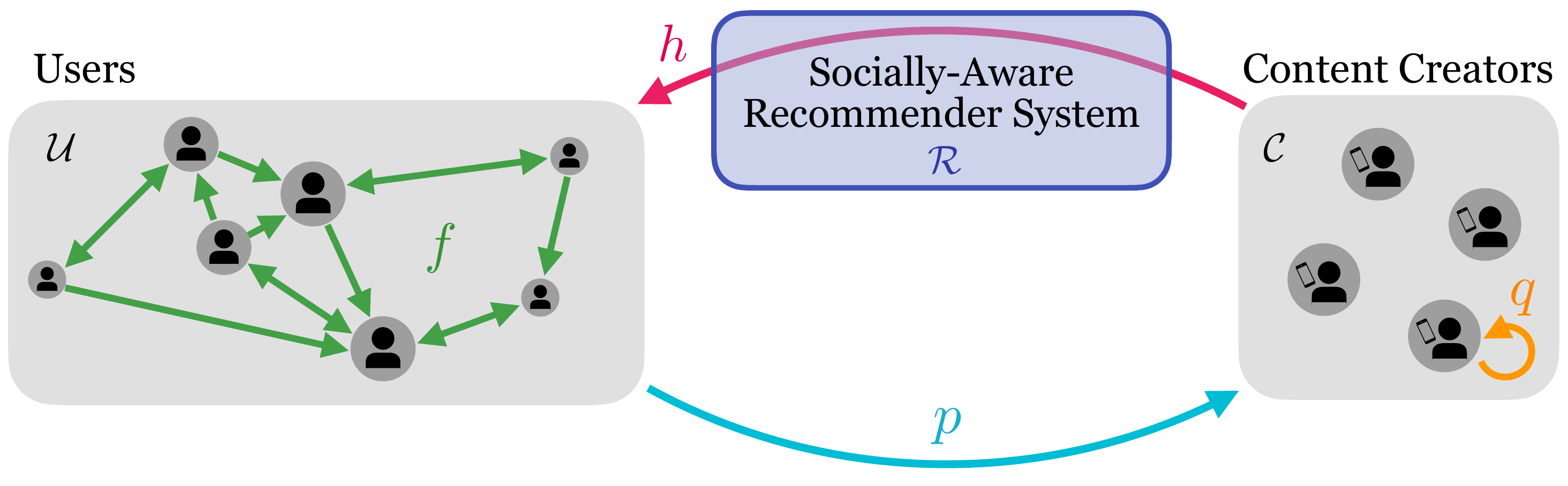}
    \caption{Overview of the dynamic framework described in equation \eqref{eqn:dynamics}. Users $\mathcal U$ influence each other via dynamics $\green{f(\cdot)}$, and are influenced by creators via $\pink{h^t(\cdot)}$, mediated by $\purple{\mathcal R}$. Content creators have internal opinion dynamics $\orange{q(\cdot)}$, and are influenced by users via $\blue{p^t(\cdot)}$.}
    \label{fig:overview}
\end{figure}

\subsection{User-user interaction} 
In our framework, users influence each other by means of interactions through their social network. 

\begin{definition}[Social Network]
    Let $ A \in \displaystyle \R^{N \times N}$ be an adjacency matrix, where $A_{ij} > 0$ indicates that user $j$ influences user $i$ with weight $A_{ij}$. The \emph{social network} is a  directed and weighted graph $\mathcal{G}(\mathcal{U}, \mathcal{E}, A)$, where each node $i \in \mathcal U$ corresponds to an individual user, and $\mathcal{E}$ denotes the set of edges representing social connections. An edge $(j,i) \in \displaystyle \mathcal{E}$ exists if and only if user $j$ has a direct influence on user $i$, i.e. $ A_{ij} \neq 0$.
    \label{def:social-network}
\end{definition}

Given Definition \ref{def:social-network} and the \cite{FJ} model for one single topic, i.e. $n=1$,
\begin{equation} \label{eqn:social-dyn}
    \green{f}(\mathbf{u}^{t}) = \green{(I_N-\Lambda)A}\mathbf{u}^{t} + \green{\Lambda} \mathbf{u}^{0},
\end{equation}
where $I_N$ is the $N$-dimensional identity matrix and $A\in [0,1]^{N\times N}$ is sub-stochastic and is the adjacency matrix of the social network. The matrix $\Lambda$ is a diagonal matrix whose elements ${\lambda}_i~\in~(0,1]$ capture the resistance to opinion change (``stubbornness") of users, and $ u_i^{0}$ the user's $i$ initial opinion (``prejudice"). Users are assumed a degree of critical thinking, i.e. the diagonal of $A$ is nonzero. We will adopt the multi-topic version of \eqref{eqn:social-dyn} from {\cite{multi_topic_FJ}} that is explained in detail in \Cref{app:multi_topic}.

We note that asymmetric social relationships, common in modern platforms with ``following" dynamics, are naturally captured by this representation. More specifically, for each user we can define the set of all users that have an influence on the user via its $d$-hop social network. The parameter $d$ controls the breadth of social context: larger values incorporate more distant connections. 

\begin{definition}[$d$-hop influencers]\label{dhop}
For a graph $\mathcal{G}(\mathcal{U},\mathcal{E},A)$, the $d$-hop influencers of user $i \in  \mathcal{U}$,
is $\text{in}_i(d) = \{j \in \mathcal{U}\mid  \text{dist}(j \rightarrow i) \leq d\}$, and $i \in \text{in}_i(d)$.
\end{definition}

\subsection{Creator-user interaction}

In our framework, creators and users influence each other by means of a RS. In particular, the RS functions as a probability distribution over content creators: for each user $i$, $c_j^t\sim \purple{\mathcal R}(u^t_i)$, where the distribution is dependent on each user's current opinion $u^t_i$. This is, the RS adjusts the probability distribution of content to the user's preferences; details on the specific recommendation strategies studied in this work are provided in the following sections.\footnote{A common choice of probability distribution is often the softmax of a utility function that maximizes confirmation bias \citep{conf_bias}. See \cite{choice_theory,harm_mitigation,Preference_amplification,Rec_effect_on_evolution_of_usr_choice} and references therein for additional details.}. 

\begin{definition}
    Let ${\mathcal{F}^t_1,\ldots,\mathcal{F}^t_{M}}$ be the set of disjoint partitions of $\mathcal U$ into subsets at time $t$. The \emph{user partition} $\mathcal{F}_j^t$  contains all users $i$ that consume content from creator $j$ at time $t$, i.e., $i \in \mathcal{F}_j^t$ if $c^t_j$ is sampled from $\purple{\mathcal R}(u^t_i)$ at time $t$. 
\end{definition}

This allows for the definition of how content creators and users influence each other. 

The influence of content creators towards users follows as
\begin{equation}\label{eq:creator_to_user}
    \pink{h^t}(\textbf{c}^t) = \pink{(I_N - \Lambda)B^t} \textbf{c}^t,  
\end{equation}
where $\pink{B^t} \in [0,1]^{N \times M}$, with $\pink{B^t}$ such that $[\green{A}\, \pink{B^t}]\vec{1}_{N+M}=\vec{1}_{N}$ similar to the setting in \cite{ctrl_strat_for_rec_sys}. $B^t$ describes the influence power of the recommendation on user $i$ opinions', where $B_{ij}^t\neq 0$ if $i \in \mathcal{F}_j^t$ at time $t$ and  $B_{ij}^t = 0$ otherwise.  This is, creator $j$ can only directly influence user's $i$ opinion at time $t$ if they consume their content. 

Similar to users, we assume a degree of critical thinking in the creators. This is captured by:
\begin{equation}\label{eq:creator_dynamics}
\orange{q}(\mathbf{c}^t) = \orange{(I_M-\Gamma)}\orange{E}  \mathbf{c}^t + \orange{\Gamma} \mathbf{c}^0,
\end{equation}
where $I_M$ is the $M$-dimensional identity matrix and $\orange{E},\orange{\Gamma} \in [0,1]^{M\times M}$ are diagonal matrices (no cross-talk between creators) that govern the temporal consistency of the creator's opinions, determining the influence of their previous stance. $\displaystyle c_j^{0}$ captures the creator's $j$ initial opinion.

The influence of users towards content creators follows as
\begin{equation}\label{eq:user_to_creator}
    \blue{p^t}(\mathbf{u}^{t}) = \blue{(I_M-\Gamma)C^t}  \mathbf{u}^{t},
\end{equation}
where $\blue{C^t} \in [0,1]^{M \times N}$, with $\blue{C^t}$ such that $[\orange{E}\, \blue{C^t}]\vec{1}_{M+N}=\vec{1}_{M}$. $C^t$ describes the influence power of user feedback on creator's $j$ opinions', where $C_{ji}^t\neq 0$ if $i \in \mathcal{F}_j^t$ at time $t$ and  $C_{ji} = 0$ otherwise. This is, creator $j$ only receives feedback from the set of users consuming their content.

\subsection{Recommender System}

The RS mediates directly the creator $\rightarrow$ user interaction, and indirectly the user $\rightarrow$ creator interaction. 
%It does so through a two-stage optimization approach fo each user $i$ by which at each time step $t$:
%\begin{enumerate}[label=(\roman*)]
%  \item It computes a reference recommendation $r_{i,\text{ref}}^t \in \sR^n$ for each user $i$.
%  \item It selects a subset $\mathcal{K}^t_i \subseteq \mathcal C^t$ of size $k < M$, consisting of the $k$ items from $\mathcal C^t$ that are closest to $r_{i,\text{ref}}$ in the $\ell_2$-norm, and assigns them a probability as per equation \eqref{eqn:choice-prob}.
%\end{enumerate}
The goal of a RS is to sort and present content to users to maximize their engagement. According to the confirmation bias theory \citep{conf_bias}, this goal can be achieved by recommending content that perfectly matches the users' existing opinion.
 However, such a greedy approach was shown to result in polarized opinions and clusterization behavior \cite{del2017modeling}. Hence, we study the problem of designing a RS that (a) maximizes users' satisfaction while (b) reducing clusterization effects on a global scale. We do so by having the RS explicitly account for the existence of a social network. To that end, we next provide formal definitions of satisfaction and clusterization.

\paragraph{Satisfaction.} Satisfaction is quantified by measuring the cumulative distance between a user's opinion and the selected content over the entire time sequence $\{0,...,T\}$.

Motivated by the fact that the engagement of users is driven by confirmation bias \cite{conf_bias}, the RS will recommend content to maximize user satisfaction, defined as follows.

\begin{definition}[User Satisfaction]\label{user_sat} The satisfaction of user $i$ (with opinion $u_i^T\in\mathcal{U}^T$) with creator $j$ (with opinion  $c_j^T\in\mathcal{C}^T$) at time $T$ is:
\begin{equation}
\text{sat}(u_i^T,c_j) =
\begin{cases}
-\dfrac{1}{T}\sum_{t=0}^{T-1}\|u_i^t - c_j^t\|_2, &  \text{if } c_j^t \sim \mathcal R_i^t, \\
0, & \text{otherwise.}
\end{cases}
\end{equation}
\end{definition}

Definition \ref{user_sat} formalizes the fact that users' engage more with content that is closely aligned with their own opinion. 
On a global scale, the RS wants to maximize global satisfaction defined as follows.
%The global satisfaction is then determined by averaging the user satisfaction across all users.
\begin{definition}[Global Satisfaction] \label{def:satisfaction}Let $\mathcal U^t$ be the set of $N$ users at time $t$, the global satisfaction at time $t$ is:
\begin{equation} \label{eq:satisfaction}
\text{sat}(\mathcal{U}^t) = \frac{1}{N}\sum_{i=0}^{N-1} \text{sat}(u_i^t).
\end{equation}
\end{definition}

\paragraph{Clusterization.} Opinion clusterization is quantified by the silhouette coefficient that each user has in its assigned cluster (as computed via $k$-means for the opinion vectors).\footnote{As is standard, the optimal number of clusters is determined by running $k$-means for various values of $k$ and selecting the value that yields the highest average silhouette.} A high silhouette coefficient indicates that a user's opinion is well-matched to its assigned cluster and poorly matched to neighboring clusters. If users can be clearly assigned to a cluster and thus have high opinion silhouettes, the opinion landscape is clusterized.
\begin{definition}[User Silhouette]\label{silhouette}
 Given an option $u_i^t \in \mathcal{F}_i$, the silhouette coefficient is defined as:
\begin{equation}
s(u_i^t) = \frac{b(u_i^t) - a(u_i^t)}{\max\{a(u_i^t), b(u_i^t)\}} \in [-1,1],
\end{equation}
where $a(u_i^t) = \frac{1}{|\mathcal{F}_i|-1} \sum_{u_j^t \in \mathcal{F}_i, j \neq i} ||u_i^t - u_j^t||_2$ is the average intra-cluster distance, and $b(u_i^t) = \min_{l \neq i} \frac{1}{|\mathcal{F}_l|} \sum_{\vu_k^t \in \mathcal{F}_l} ||u_i^t - u_k^t||_2$ is the minimum average outer-cluster distance.
\end{definition}

\begin{definition}[Global Clusterization]\label{de:global_clust}
 The global opinion clusterization of the set $\mathcal{U}^t$ is defined as
\begin{equation} \label{eq:clusterization}
    \text{cl}(\mathcal{U}^t) = \cfrac{1}{N}\sum_{i = 0}^{N-1} s(u_i^t).
\end{equation}
\end{definition}

In what follows, we present a RS design that ensures balance between satisfaction and clusterization by taking into account the social network dynamics.

\section{Recommender System Design}

In this section, we present the design of our \emph{socially-aware RS}. As is standard in personalized feed mechanisms employed by online platforms, we consider a two-stage algorithmic curation strategy that determines which subset of available content reaches each user. At each time step $t$, the RS:
\begin{enumerate}[label=(\roman*)]
    \item Computes a reference recommendation $\purple{r_i^t} \in \mathbb{R}^n$ for each user $i$ as a function of their opinion $u_i^t$, i.e. $\purple{r_i^t} := \purple{r}(u_i^t)$.
    \item Provides the \texttt{top-k} recommendations by presenting the user with the $k$ content items closest to this reference point through $k$-nearest-neighbor search:
\begin{equation}\label{eqn:set}
\purple{\mathcal{R}}(u_i^{t}) = \left\{c_j^{t}   
 \;\bigg|\; 
\lVert c_j^t - \purple{r_i^t} \rVert_2 \in k\text{-top } \underset{j} { \text{min}}
\{ \lVert c_j^t - \purple{r_i^t} \rVert_2\}_{j=1}^M
\right\}.
\end{equation}
\end{enumerate}

Each user $i$ samples one element (piece of content) from $\purple{\mathcal{R}}(u_i^{t})$ at every time step $t$, i.e., $c_j^t\sim \purple{\mathcal{R}}(u_i^{t})$ according to softmax with temperature parameter $ \delta^{-1}$. In what follows, we provide a theoretical result that informs how the choice of $\purple{r}(u_i^t)$ impacts user satisfaction and clusterization, and provide a design choice that, by mimicking the structure of the social network provably achieves a sweet spot between the two metrics.

\subsection{Recommendation Strategies for Satisfaction and Clusterization}

By studying the dynamics of the feedback mechanism in \ref{fig:overview}, it is possible to derive insights on how to design the RS to avoid clusterization while promoting user satisfaction. Consider the dynamics in \eqref{eqn:dynamics} where no RS is intervening and the user partitions are static. In this case, \eqref{eqn:dynamics} becomes: 
\begin{subequations}\label{eqn:simple_dynamics}
\begin{align}
    &\displaystyle \mathbf u^{t+1} = \green{(I_N-\Lambda)A}\mathbf{u}^{t} + \green{\Lambda} \mathbf{u}^{0} +  \pink{(I_N-\Lambda)}\pink{B}\mathbf c^t, \\
    & \displaystyle \mathbf c^{t+1} = \orange{(I_M-\Gamma)E}  \mathbf{c}^t + \orange{\Gamma} \mathbf{c}^0 +\blue{(I_M-\Gamma)}\blue{C}\mathbf u^{t},
\end{align}
\end{subequations}
where $\pink{B}$ and $\blue{C}$ are static, since the user partitions are static. This is, users engage with the same content creator over time. In what follows, we explore the influence of the dynamic components (social network $A$, creator's influence $B$, etc.) on the user's opinion evolution.

We also study how different recommendation strategies, influence the overall dynamic behavior and emphasize the key role the RS and the social connections have in this framework. To do that, we consider a recommendation strategy as in \eqref{eqn:set} with $k=1$. This implies that the RS becomes a deterministic map, i.e. $c_j^t = \purple{R}(u_i^t)$. Note that this relaxation allows for the bypassing of the stochastic sampling of content given a recommendation -- this assumption is solely used for the theoretical results; simulations consider the full stochastic setup. 

\iffalse
By studying the dynamics of the feedback mechanism in \Cref{fig:overview}, it is possible to derive insights on how to design the RS to avoid clusterization while promoting user satisfaction. To emphasize the impact of the network over opinion clusterisation, let us first consider the dynamics in \eqref{eqn:dynamics} with no interaction between users (i.e. $\green{A}$ is a diagonal matrix), an assumption that we will relax later, and a deterministic RS $k = 1$. In this case, \eqref{eqn:dynamics} becomes: 
\begin{subequations}\label{eqn:simple_dynamics}
\begin{align}
    &\displaystyle \mathbf u^{t+1} = \green{A}\mathbf{u}^{t} + \green{\Lambda} \mathbf{u}^{0} + \pink{B^t}\mathbf c^t, \\
    & \displaystyle \mathbf c^{t+1} = \blue{C^t}\mathbf u^{t} + \orange{E}  \mathbf{c}^t + \orange{\Gamma} \mathbf{c}^0,\\
    &\quad c_j^t = r_i^t(u_i^t).
\end{align}
\end{subequations}
Note that this relaxation allows for the bypassing of the stochastic sampling of content given a recommendation -- this assumption is solely used for the theoretical results; simulations consider the full stochastic setup. 
\fi

\subsubsection{Social Interactions mitigate the effect of Recommendations}

\begin{theorem}\label{thm}
Given the system with dynamics in \eqref{eqn:simple_dynamics}, the users' opinions reaches a steady state, i.e., $\exists \mathbf u^*\in\mathbb R^N$ such that $\lim_{t\rightarrow \infty}\ \mathbf{u}^t = \mathbf u^*$. Moreover, the influence of the social network ($A$) \textcolor{black}{towards each user $i\in\{1,\dots,N\}$, namely $\sum_j A_{i,j}$, and the recommended content, i.e. $B_i$} on $\mathbf u^*$ are complementary: increasing one decreases the other.
\end{theorem}
%\begin{proof}
%We refer the reader to \Cref{app:proof_th1}.
%\end{proof}

\iffalse
\begin{theorem}\label{clustering_thm}
    Given the dynamics in \eqref{eqn:simple_dynamics}, if a partition $j$ turns into static at a time $t_0$, i.e. $\mathcal{F}_j^t = \mathcal{F}_j^{t_0}$, $\forall t \geq t_0$, then there exists $ \text{ }u_{i}^*, c_{j}^* \text{ with } i \in \mathcal{F}_j$, such that:
    \begin{equation}
    \begin{aligned}
        &\lim_{t \rightarrow\infty} u_i^t = u_{i}^* \quad \forall i \in \mathcal{F}_j\\
        &\lim_{t \rightarrow\infty} c_j^t = c_{j}^* 
    \end{aligned}
    \end{equation} 
    And for all users $u_i \in \mathcal{F}_j$, $u_i^*= c_j^* + (I-A_{ii})^{-1}\Lambda_{ii}(u_{i}^0-c_j^0)$, $\forall i \in \mathcal{F}_j^{t_0}$.
    If 
    
\end{theorem} 
\begin{proof}
We refer the reader to 
\end{proof}
\fi 

The proof can be found in \Cref{app:proof_lemma1}. Intuitively, Theorem \ref{thm} states that when users consistently engage with the same creator over time, they gradually approach the creator's opinion. Creators, in turn, adapt to their fixed audience until an equilibrium is reached. The degree to which this happens depends on the user's and creator's stubbornness (as captured by $\Lambda$ and $\Gamma$, respectively) as well as the strength of the social connections (captured in $A$). The theorem shows that the social network and the recommended content execute two opposing forces over the steady state users' opinions. 

We note that the social network promotes dynamics partitions by propagating opinions through user connections. Yet, even connections within the same partition can mitigate the formation of static assignments through indirect network effects. Users within the same partition often maintain paths to users in other clusters, enabling opinion diffusion across partition boundaries.

%If the user has a high degree centrality, then he receives high social influence and the user will most likely not fall into one unique partition over time. Thus, the social network and greedy RS exert opposing forces: the former promotes opinion diversification through peer influence, while the latter reinforces existing preferences and promotes clustering. The critical insight is that global polarization emerges when social network effects are insufficient to counterbalance the RS's homogenizing influence. This tipping point, where clustering effects begin to dominate, depends on the relative strength of social connections versus recommendation influence. We showcase this phenomenon empirically in the experiments section.

\subsubsection{RS that Maximize Satisfaction increase Clusterization}

The main goal of a RS is to maximize user's engagement. As such, it is natural to consider RS designs that maximize user's satisfaction. A greedy RS is one that solely maximizes satisfaction by exploiting confirmation bias \cite{conf_bias}. That is, the RS provides the content creator that most closely aligns with users opinions, i.e. 
\begin{equation}\label{eqn:greedy_RS}
    R(u_i^t) = \text{argmin}_{c_j}(||c_j^t-u_i^t||_2).
\end{equation}
This strategy achieves the maximum satisfaction as per definition of \emph{user satisfaction} in \Cref{def:satisfaction}. 

\begin{lemma}\label{clustering_lemma}
\textcolor{black}{
Consider the system with dynamics as in \eqref{eqn:dynamics}. Let $A$ be a diagonal matrix and creators being stubborn, i.e. $\Gamma = I_M$. Then, a greedy RS as in \eqref{eqn:greedy_RS} induces a static user partition  $\mathcal{F}_1,...,\mathcal{F}_M$ and the users' opinions reach a steady state as per Theorem \ref{thm}. Moreover, for any user $i \in \mathcal{F}_j$, all opinions approach the recommended creator's content, namely $||u_i^t - c_j^t||_2$ decreases monotonically in time.  
}
% Consider the system with dynamics as in equation \ref{eqn:simple_dynamics}, and let $A$ be a diagonal matrix and $\Lambda \ll I_N$ with $N\gg M$. Then, a greedy RS as in \eqref{eqn:greedy_RS} increase clusterization in the proximity of content creators as in equation \ref{eq:clusterization}.
  %Given the system with dynamics as in equation \ref{eqn:simple_dynamics}, in the absence of social interactions,  i.e. $A$ is a diagonal matrix, for $N\gg M$ and for highly susceptible users, i.e. $\Lambda \approx 0$, greedy recommendations increase clusterization (equation \ref{eq:clusterization}) in the proximity of content creators.% Namely, there exists $t_0\geq 0$ such that $\mathcal{F}_j^t=\mathcal{F}_j^{t_0}$, for all $ t \geq t_0$.
\end{lemma}

The proof can be found in \Cref{app:proof_lemma1}. Intuitively, by static assignment of users to creators, users will fall into filter bubbles where they engage with creators who provide users with content that closely aligns with their opinions. This effect is aggravated when the users do not have social interactions to counterbalance the effect of the RS, as per Theorem \ref{thm}, since static user-creator mappings lead to the formation an equilibrium which fosters clusterization. Moreover, the closer a user's opinion gets to the creators' opinion, the smaller the chances for them to exit the bubble. The critical insight is that global polarization emerges when social network effects are insufficient to counterbalance the RS's homogenizing influence. This tipping point, where clustering effects begin to dominate, depends on the relative strength of social connections versus recommendation influence. We showcase this phenomenon empirically in the experiments section.

% Our claim follows from Lemma \ref{clustering_lemma} and by noting that greedy RS promote the same or very similar creators repeatedly. Since users adapt to creators' opinions over time (with the limiting case where they coincide with the creators' opinions), global opinion clusters will form around the creators' opinions.
% Intuitively, a greedy RS drives users towards similar content creators, who in turn adapt their content to match their audience's preferences, creating a feedback loop that reinforces clustering.

% \begin{proof}
% \begin{enumerate}
%     \item Show that maximizing user satisfaction (eq 7) is the same as $r_i^t= argmin_{c_j} ||c_t^t-u_i^t||^2$ s.t. $u^{t+1}=\sum_j A_{ij}u^t_j + B r_i^t$ and $c_{i,j} = \sum_j C_{ij}u^t_j + B c_j^t $ and $r_i = c_j$ WHERE (THIS IS THE KEY, I THINK) $C_ij = 0$ if $i\notin \mathcal F_j$ and not $0$ otherwise. If the fact that satisfaction is over time gives you trouble, just add the assumption that we are considering "immediate satisfaction", this is 1 time step.
%     \item I think you can easily show that the problem above is asymptotically stable, over $t$, which I think (should) lead to the global silhouette increasing monotonically over time. [If it doesn't, hack: we change the definition and define the clusterization metric to be a meassure for assymptotic stability so this works].
% \end{enumerate}
% \end{proof}

\subsubsection{RS that reduce Clusterization decrease Satisfaction}
We have shown that targeting only satisfaction, in the absence of social interactions, promotes opinion clusterization around creators. To counter the static assignment of each user to creators, the RS has to depart from engagement maximization and promote content diversity.

\begin{corollary}\label{diversification_lemma}
    Let $c_j^t=R(u_i^t)$ in \eqref{eqn:simple_dynamics} $\forall i,j$ result from a non-greedy RS, that is, $R(u_i^t) \neq \text{min}_j(||c_j^t-u_i ^t||_2)$. Then the user satisfaction in \eqref{eq:satisfaction} is suboptimal.
\end{corollary} 
\begin{proof} 
The proof follows directly from \eqref{eq:satisfaction}. That is $c_k^t = R(u_i^t)$ with $||u_i^t-c_k^t||_2 > {\rm min}_j ||u_i^t-c_j^t||_2$, for some $j \in 0,...,M-1$, will lead to reduced satisfaction of user $i$.
\end{proof}

Together with Lemma \ref{clustering_lemma}, this statement promotes the search of a tradeoff between maximizing users satisfaction while countering opinion clusterization.

\subsection{A Socially-Aware Recommender System Mimics the Social Network}
Given that social connections naturally mitigate clustering, we propose a recommendation strategy that explicitly leverages network structure. Rather than optimizing solely for individual preferences, our approach expands on the greedy RS in a socially-aware manner. This is, the RS incorporates the opinions of users within a social neighborhood besides the individual opinion of the user. To do so, we will leverage the concept of \emph{$d$-hop influencers} from \Cref{dhop}.  %More specifically, for each user we can define the set of all users that have an influence on the user via $d$-steps.
%\begin{definition}[$d$-hop influencers]
%For a graph $\mathcal{G}(\mathcal{U},\mathcal{E},A)$, the $d$-hop influencers of user $i \in  \mathcal{U}$,
%is $\text{in}_i(d) = \{j \in \sU \mid  \text{dist}(j \rightarrow i) \leq d\}$, and $i \in \text{in}_i(d)$.
%\end{definition}
%The parameter $d$ controls the breadth of social context: larger values incorporate more distant connections. 
In particular, we design a RS that mimics the social influence by leveraging the mean opinion of the $d$-hop influencers of each user, and using this as the RS recommendation reference. 
% \textcolor{red}{ The proposed strategy takes advantage of the fact that users are homophilic, that is, users have the tendency to connect with other users with similar interests (citation). }
\begin{definition}[$d$-hop socially-aware RS]\label{def:d-hop}
Let $\text{in}_i(d)$ be the $d$-hop influencers of user $i$. The $d$-hop socially aware RS produces a recommendation reference that is given as:
\begin{equation}
    r(u_i^t) = \cfrac{1}{|\text{in}_i(d)|}\sum_{j \in \text{in}_i(d)}u_j^t.
\end{equation}
\end{definition}
\vspace{-3mm}
The design parameter $d$ controls the trade-off between personalization and diversification: $d=0$ recovers the greedy strategy in Lemma \ref{clustering_lemma}, while larger values incorporate broader social influence. 
% Since this approach explicitly mimics social influence mechanisms, we have the following:
% \begin{lemma}
%    Let $R(u_i)$ be a $d$-hop socially-aware RS with parameter $d$. Then, global clusterization as per \Cref{de:global_clust} decreases as $d$ increases.
% \end{lemma}

% The proof can be found in \Cref{app:lemma2}. Intuitively, by the dynamics in \eqref{eqn:simple_dynamics}, 
If the RS provides the $d$-hop neighborhood average opinion then users' opinion are pulled towards the average opinion of their neighborhood, including neighbors from other clusters. This mechanisms will reduce the separation between cluster. 
Increasing $d$ enhances content diversity at the cost of reduced user satisfaction, as recommendations deviate further from individual preferences. In the next section, we empirically analyze this trade-off and identify parameter values that minimize global clustering while maintaining user satisfaction.

\section{Experiments}
We evaluate the performance of our RS using the satisfaction and clusterization metrics as in \eqref{def:satisfaction} and \eqref{de:global_clust}. In a first experiment we study the closed loop interaction of our RS with a synthetic social network. We then demonstrate our findings on the Facebook-ego dataset \cite{ego_facebook_data}, which comprises the social network of $4039$ users.
% Our experimental setup assumes that recommender influence dominates over social network effects. 
% We further assume that user and creator stubbornness is low, $\Lambda\ll I_N$ and $\Gamma\ll I_M$  to accelerate convergence of the opinion dynamics. 
Throughout this section, the temperature parameter governing the users choice is chosen as $\delta^{-1} = 0.5$. We present our main findings below. 
% Additional experimental results and detailed analysis can be found in \Cref{appendix:k_analysis}.

\subsection{Experimental Setup For Synthetic Dataset}
We consider a network of $N=600$ users and $M=50$ content creators, with an average user in-degree of $11$. All interaction parameters governing the user-creator and creator-user dynamics in \eqref{eqn:dynamics} are detailed in \Cref{appendix:model_parameters}. For clarity of representation, we set the opinion dimension to $n=2$, with all user and creator opinions initialized uniformly at random within $[-1,1]^n$. The social network topology is randomly generated, with the probability of an edge of being present between users decreasing with their opinion distance.  As a result, users with closer opinions are more likely to be connected. This setup aligns with the homophily principle in social networks, which leads to more contact between similar users \cite{Homophilly_Pherson}. The specific connection probability function, is detailed in \Cref{appendix:connection_prob_synthetic}.

\subsection{Experimental Results on Synthetic Dataset}
The RS follows a \texttt{top-$k$} recommendation strategy with $k=5$. The impact of different $k$ and varying interaction parameters in \eqref{eqn:dynamics} on all performance metrics for the following settings is analyzed in \Cref{appendix:k_analysis_synthetic}.

\subsubsection{Clusterization and Satisfaction with different RS Strategies}
% In figure \cref{fig:dynamics_lam_0.0} the opinion landscapes of the creators is displayed after various different timesteps with $d=0$ (greedy RS). The user opinion clusters after k-means are displayed as elliposids with $2$ times the standard deviation along the axis with maximum and minimum variation, where the transparency is increased if the global clusterization is reduced. If the global opinion clusterization is lower than $0.5$ we do not display the clusters at all, as our results show that below this threshhold, clusters are very hard to identify. After $15$ timesteps, distinct clusters form around the creators. This effect increases for larger time horizons. 

We investigate three distinct RS strategies: (i)  RS only optimizes for user engagement (greedy RS with $d=0$), (ii) RS  only accounts for opinion diversity (\emph{socially aware RS} with $d=6$, high diversification), and (iii) hybrid RS that balances satisfaction and diversity (socially aware RS with $d=3$, intermediate localization). Strategy (i) only accounts for personalized recommendations, strategy (ii) emphasizes social interactions by considering broad network effects, while strategy (iii) seeks trades-off between user satisfaction and content diversity.

Figure~\ref{fig:dynamics_lam_0.0} shows the evolution of user and creator opinion landscapes at different time steps under the three proposed RS strategies. User opinion clusters, obtained via k-means clustering, are visualized as ellipsoids with axes corresponding to twice the standard deviation along the principal components. The ellipsoid transparency decreases with global clusterization. Clusters with global clusterization values below 0.5 are omitted, as they lack structural definition. The last column in \cref{fig:dynamics_lam_0.0} shows the negative global clusterization and user satisfaction over time. Ideally, one would like both curves to increase over time. We notice how the greedy RS (i) leads to the formation of distinct user clusters around creators, with this polarization effect intensifying over longer time horizons. The socially-aware RS (ii) reduces clusterization but at the cost of significantly lower user satisfaction. The hybrid RS (iii) maintains high satisfaction while effectively mitigating clusterization effects, achieving a balance between the two objectives.
\begin{figure*}[htbp]
    \centering
    \setlength{\tabcolsep}{2pt} % Minimal spacing between columns
    
    \begin{tabular}{c@{\hspace{0.1cm}}c@{\hspace{0.1cm}}c@{\hspace{0.1cm}}c}
        % Header row
        $t = 5$ & $t = 20$ & $t = 50$ & \orange{-Cl} / \blue{Sat}\\[1mm]
        
        % First row with images
        \includegraphics[width=0.23\textwidth]{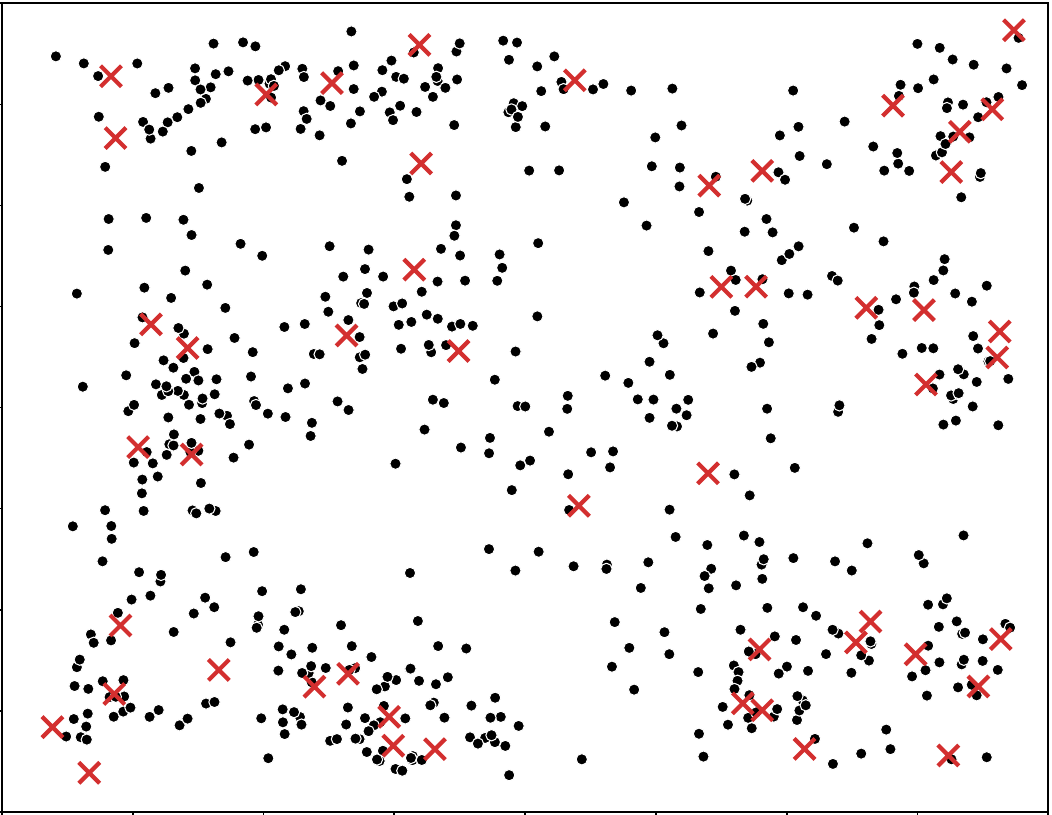} &
        \includegraphics[width=0.23\textwidth]{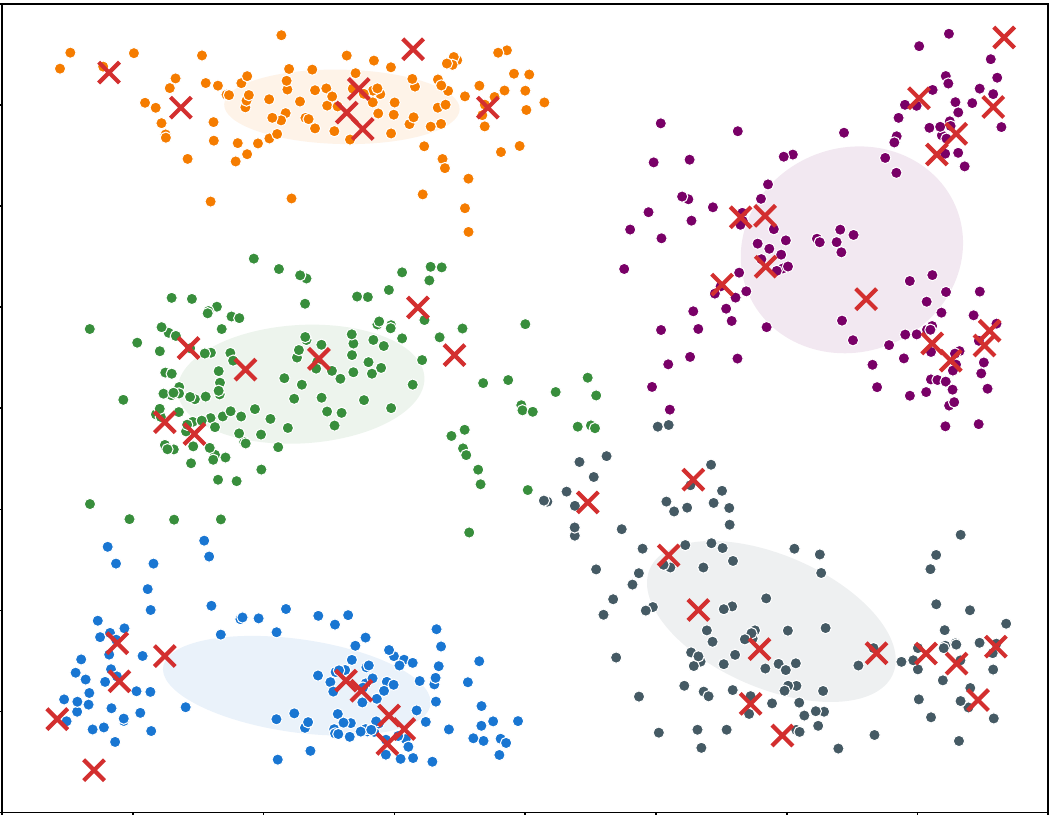} &
        \includegraphics[width=0.23\textwidth]{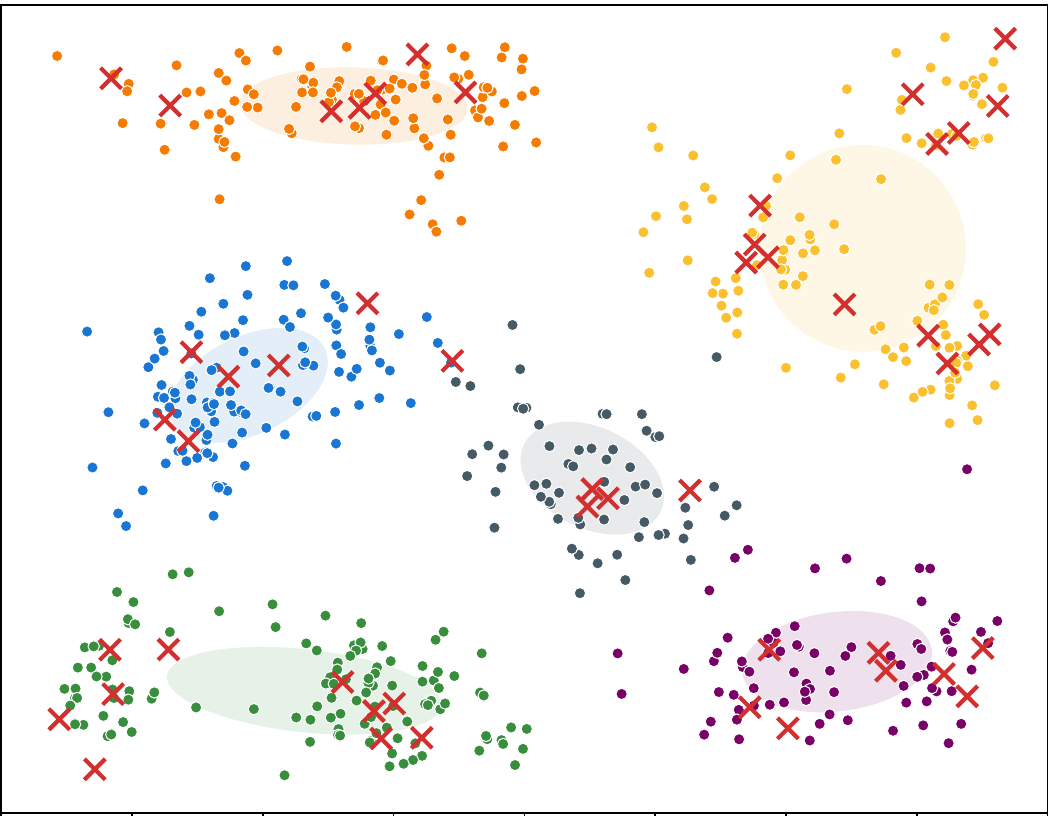} &
        \includegraphics[width=0.23\textwidth]{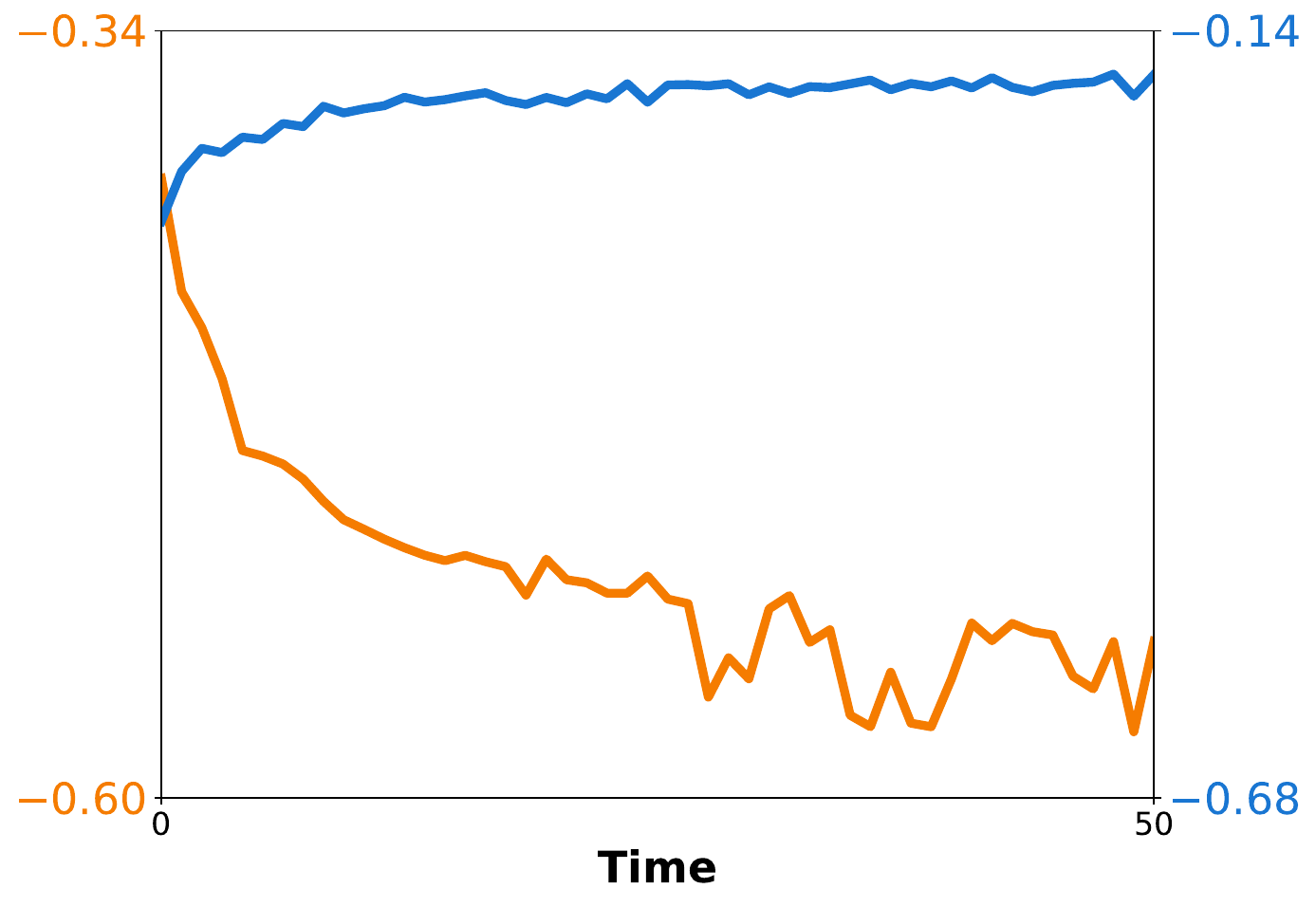} 
    \end{tabular}
    \begin{tabular}{c@{\hspace{0.1cm}}c@{\hspace{0.1cm}}c@{\hspace{0.1cm}}c}
        % Header row
        
        % First row with images
        \includegraphics[width=0.23\textwidth]{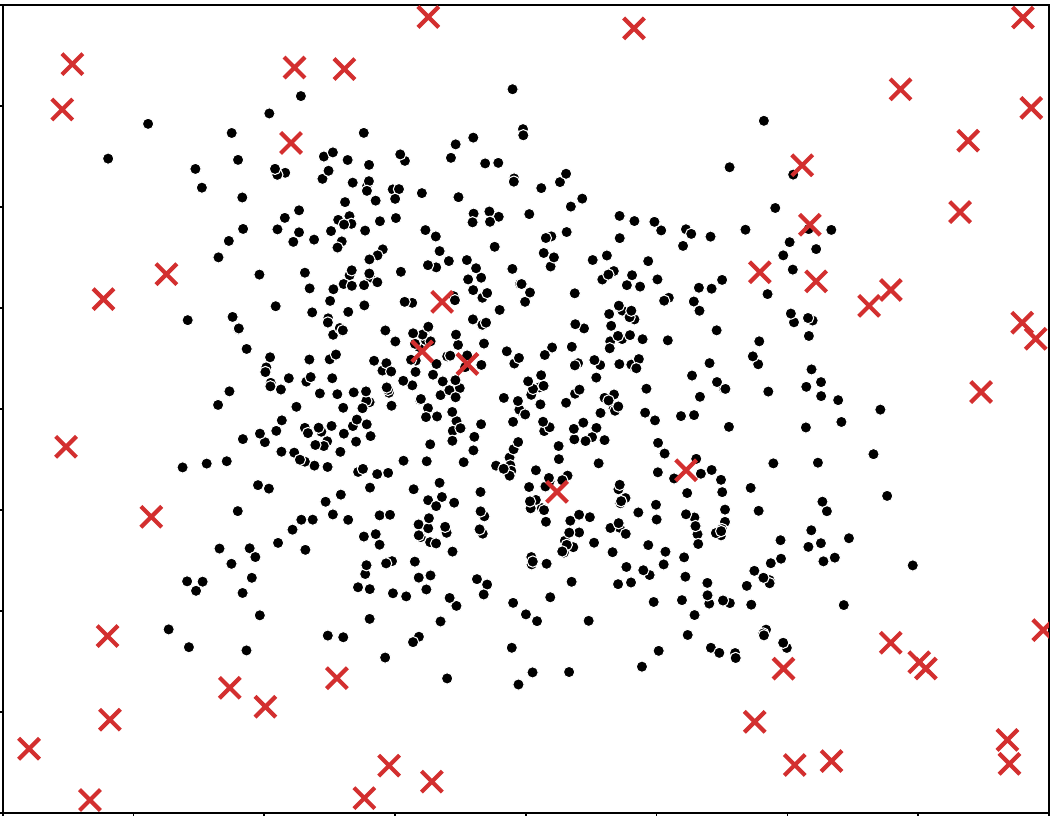} &
        \includegraphics[width=0.23\textwidth]{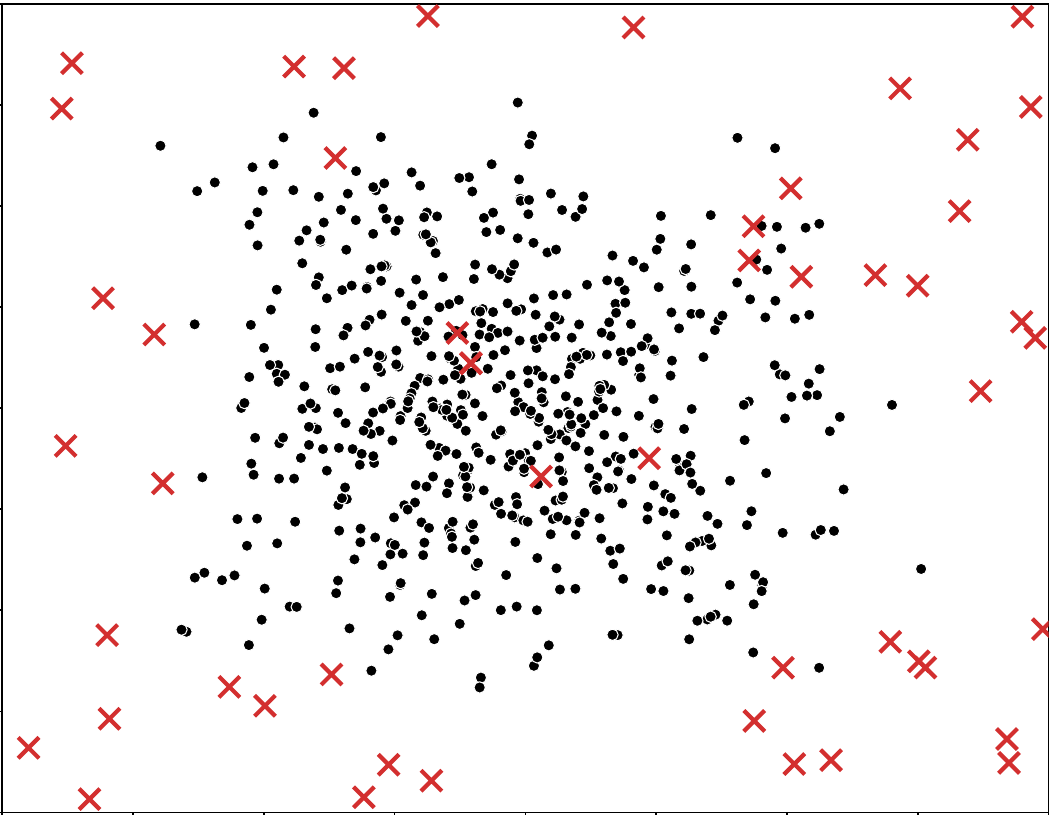} &
        \includegraphics[width=0.23\textwidth]{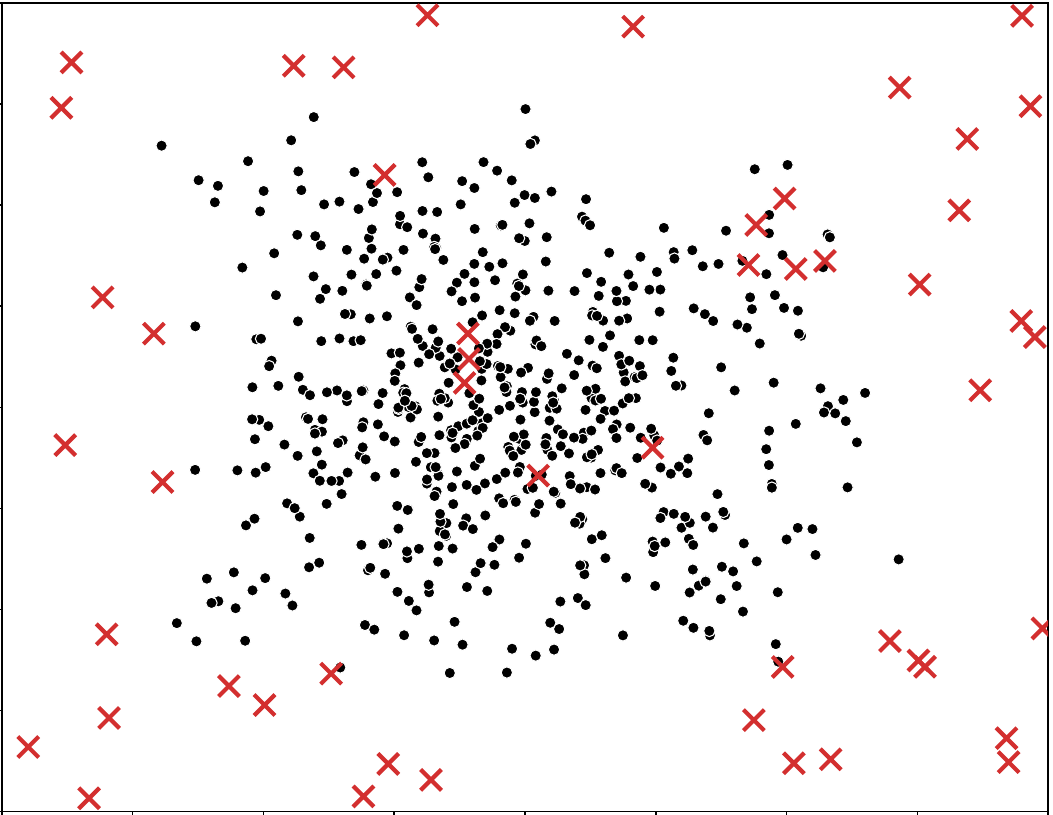} &
        \includegraphics[width=0.23\textwidth]{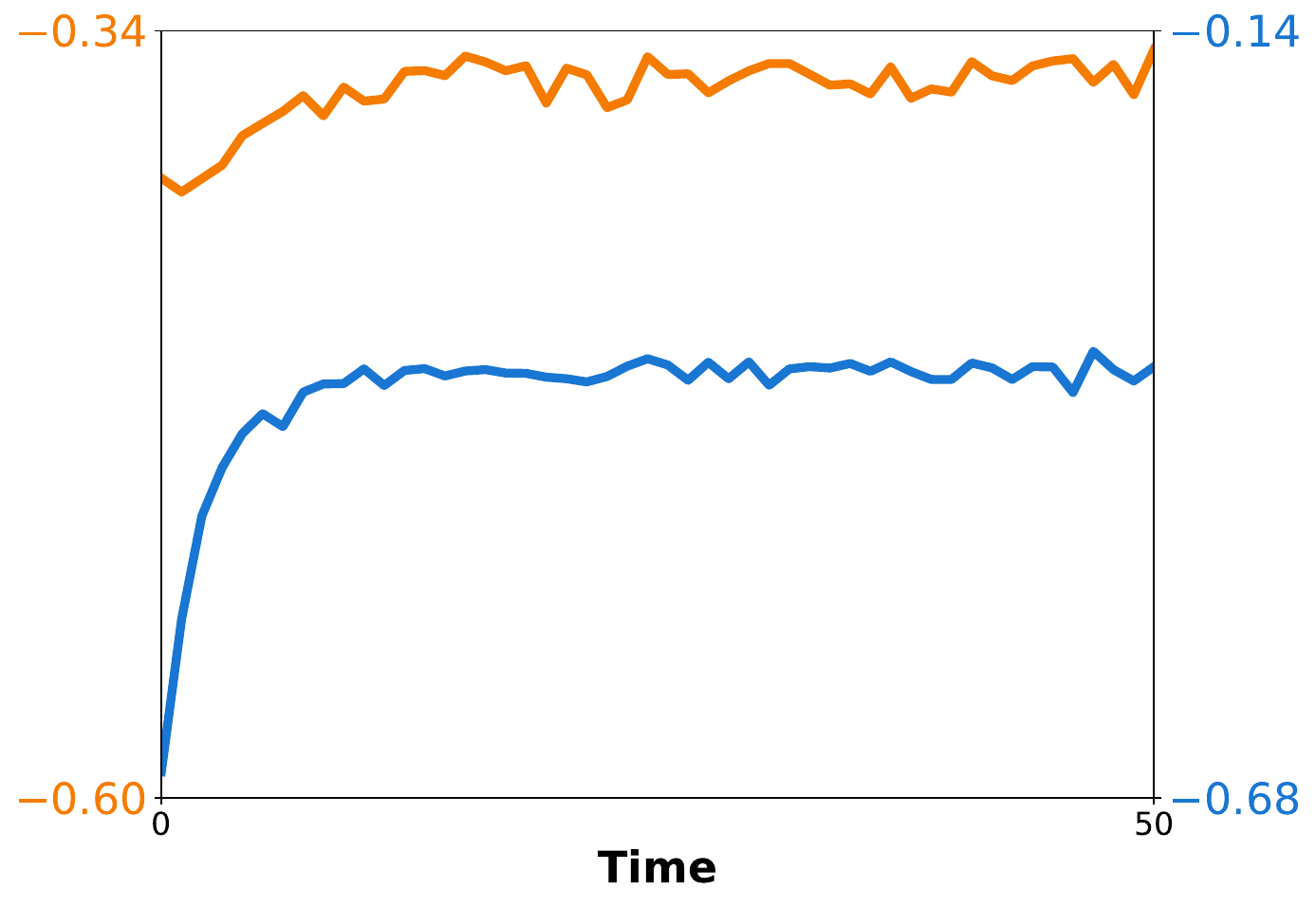} \\[1mm]
    \end{tabular}
    \begin{tabular}{c@{\hspace{0.1cm}}c@{\hspace{0.1cm}}c@{\hspace{0.1cm}}c}
        % Header row
        
        % First row with images
        \includegraphics[width=0.23\textwidth]{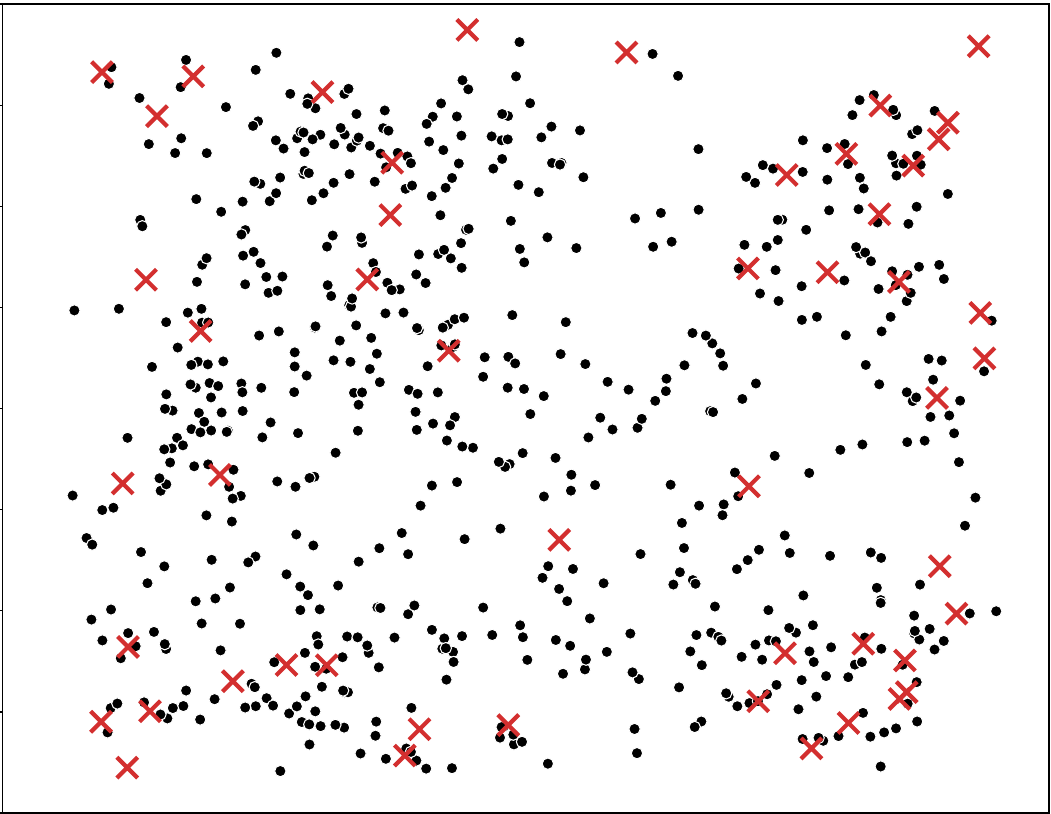} &
        \includegraphics[width=0.23\textwidth]{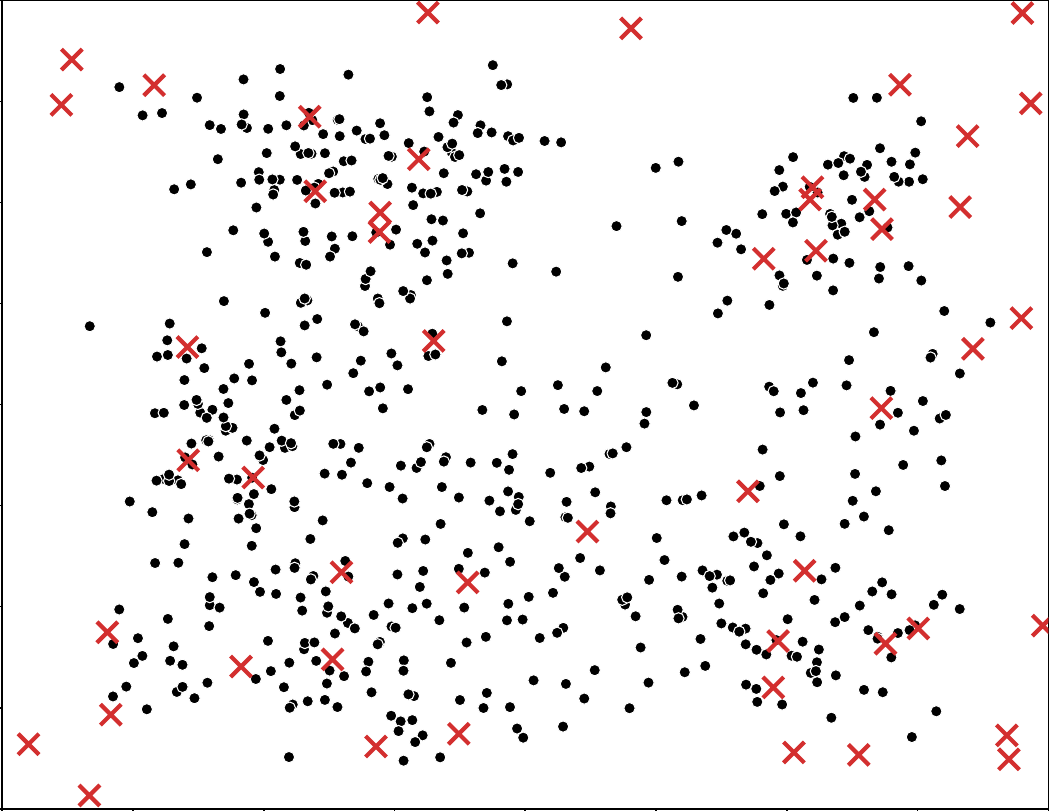} &
        \includegraphics[width=0.23\textwidth]{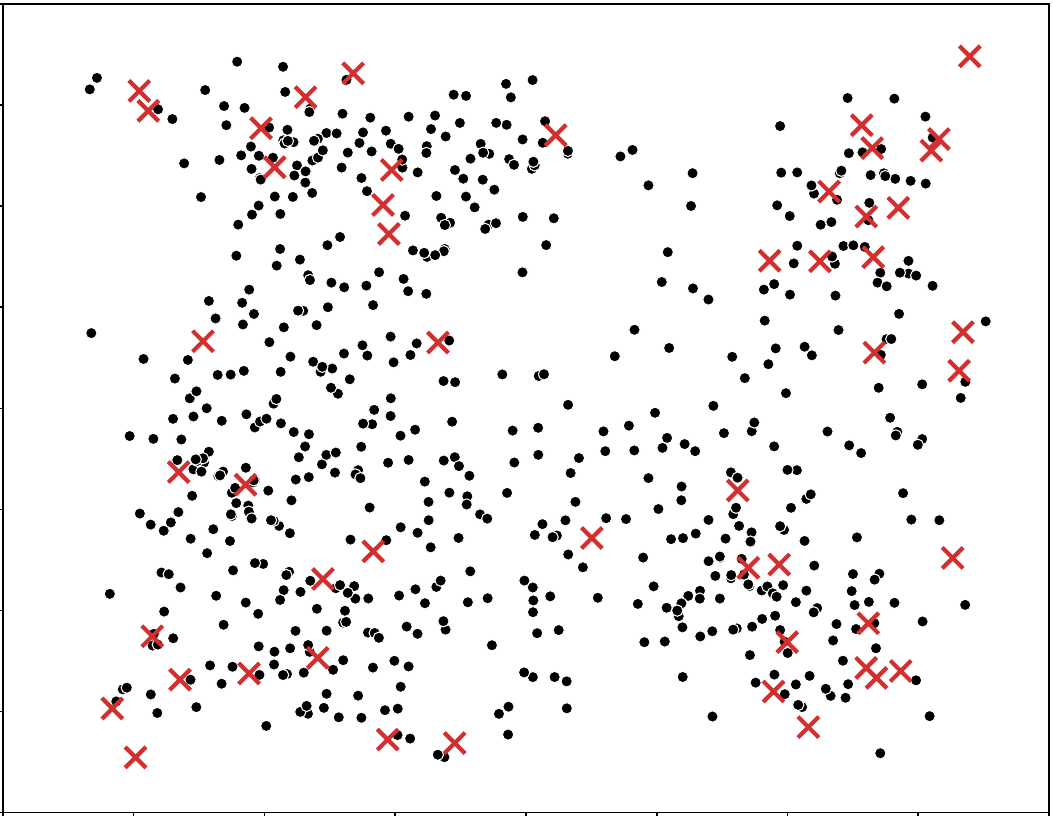} &
        \includegraphics[width=0.23\textwidth]{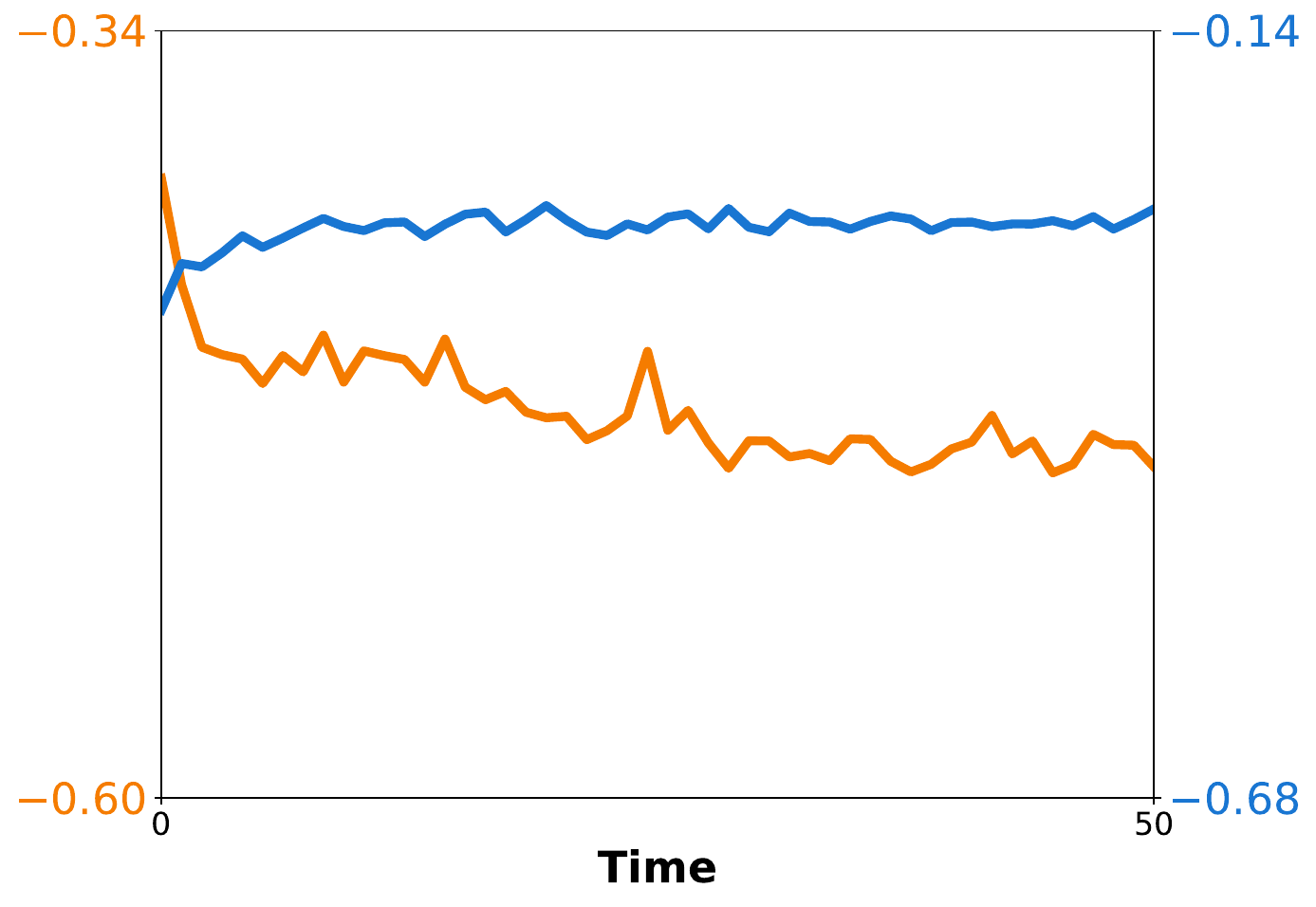} \\[1mm]
    \end{tabular}
    
    \caption{Snapshots of the opinion environment simulated with a localized region (i) $d=0$, first row, (ii) $d=6$, second row (iii) $d=3$, third row. \textcolor{red}{$\times$} denotes the creator, $\bullet$ the users respectively.}
    \label{fig:dynamics_lam_0.0}
\end{figure*}

 %The negative global opinion clusterization and the global user satisfaction is displayed  in \Cref{fig:varying_lambda} as function of the localization parameter $d$, for $t=50$ and $t=500$, respectively, together with their variance \footnote{We denote the variance of the satisfaction and the clusterization over the users in the network as in \eqref{eq:satisfaction} and \eqref{eq:clusterization}.}. When increasing the localized region, the global clusterization effect decreases but so does the user satisfaction. We can observe a sweet spot between the two for $d =3$, where the global clusterizaion effects strongly decrease while the global user satisfaction stays high.
% \vspace{-0.5cm}
 \Cref{fig:varying_lambda} displays negative global opinion clusterization and user satisfaction as functions of the localization parameter $d$ at $t=50$ and $t=500$, along with their variance\footnote{Variance is computed over users as in \eqref{eq:satisfaction} and \eqref{eq:clusterization}.}. As $d$ increases, both clusterization and satisfaction decline. A sweet spot occurs at $d=3$, where clusterization drops sharply while satisfaction remains high.
% \vspace{-0.3cm}
\begin{figure}[h!]
    \centering
    \begin{subfigure}[b]{0.36\textwidth}
        \centering
        \includegraphics[width=\textwidth]{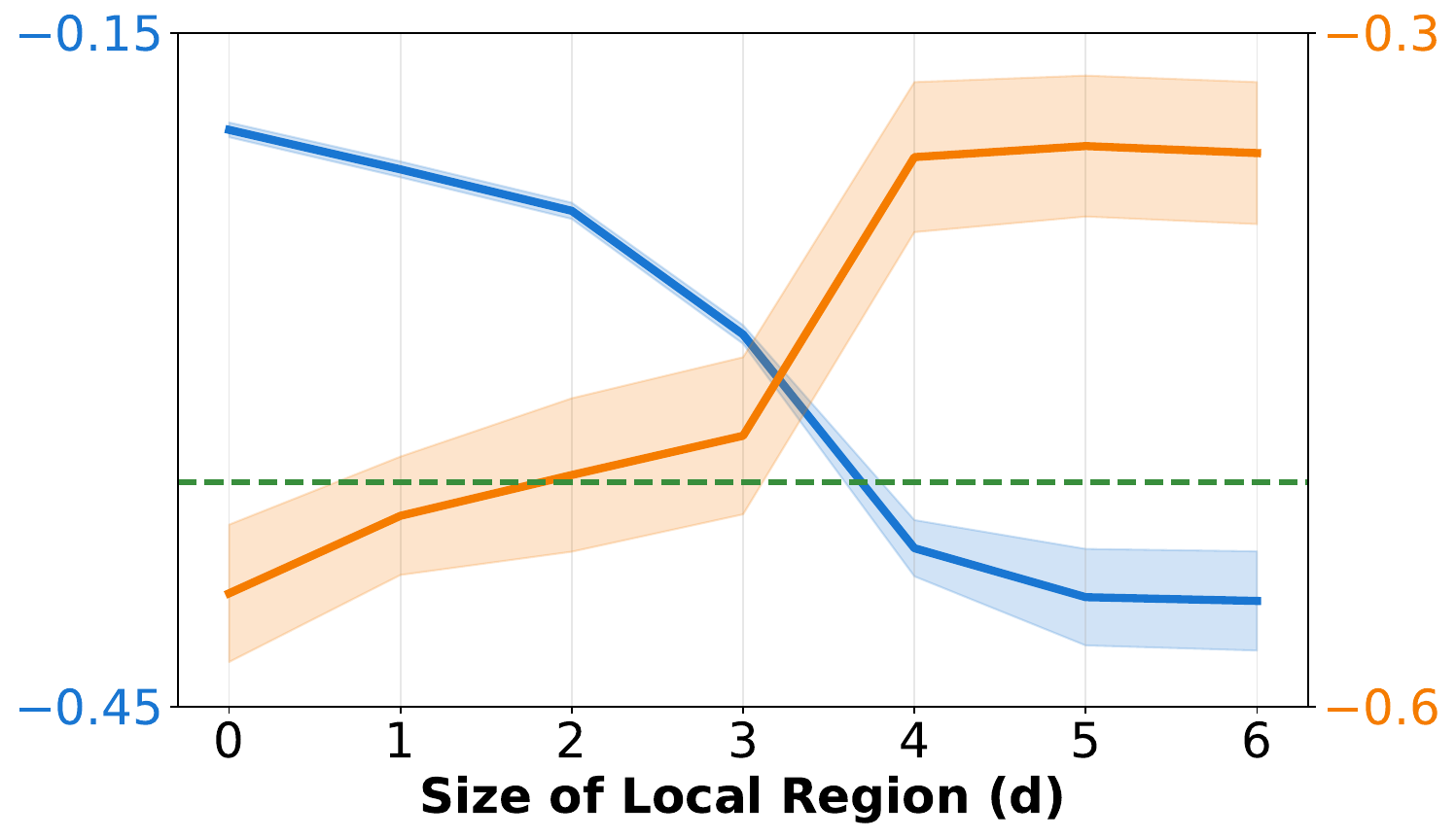}
        \caption{After 50 timesteps}
        \label{fig:varying_lambda_50}
    \end{subfigure}
    \hfill
    \begin{subfigure}[b]{0.36\textwidth}
        \centering
        \includegraphics[width=\textwidth]{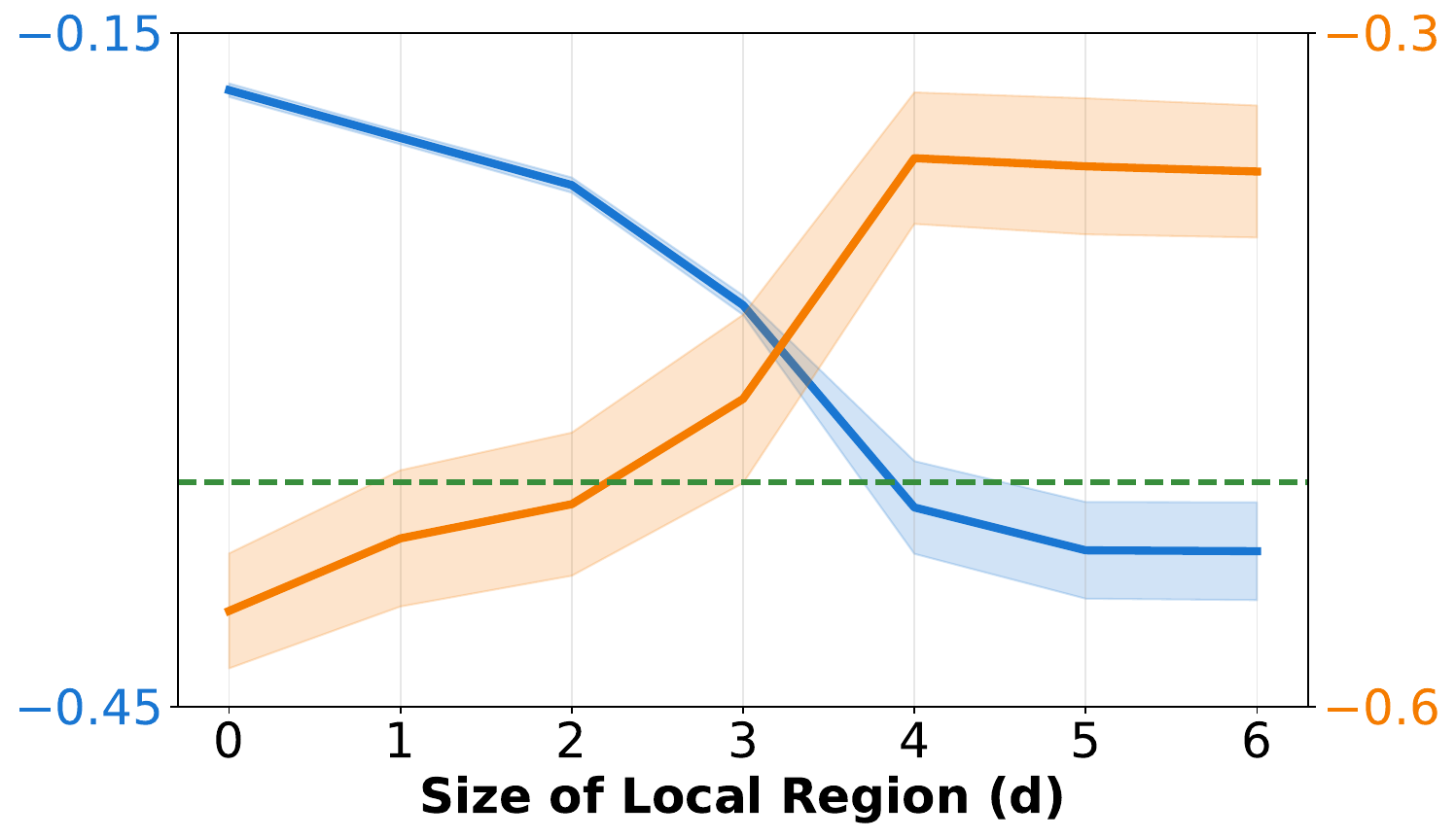}
        \caption{After 500 timesteps}
        \label{fig:varying_lambda_500}
    \end{subfigure}
    \hfill
    \begin{subfigure}[b]{0.25\textwidth}
        \centering
        \raisebox{1.5cm}{%  % Adjust this value to center vertically
        \small
            \begin{tabular}{cl}
                \textcolor[RGB]{249,124,0}{\rule{0.2cm}{0.1cm}} & Neg. Clusterization \\[0.2cm]
                \textcolor[RGB]{25,118,210}{\rule{0.2cm}{0.1cm}} & Satisfaction \\ [0.2cm]
                \textcolor[RGB]{56,142,60}{\rule{0.2cm}{0.1cm}} & Clusterization Thresh.
            \end{tabular}
        }
        \caption*{}  % Empty caption to maintain alignment
    \end{subfigure}
    \caption{Global clusterization and global user satisfaction plotted as $d$ varies after (a) 50 and (b) 500 timesteps. Clusterization thresh is set to $-0.5$, for which clusters are no longer distinguishable.}
    \label{fig:varying_lambda}
\end{figure}

\subsection{Experimental Setup for Real Dataset}
\textcolor{black}{
The ego-Facebook dataset comprises the social network of $4039$ users. The resulting social graph $\mathcal{G}(\mathcal{U},\mathcal{E},W)$ has an average degree of $45$ and is bidirectional, i.e. $(i,j)\in \mathcal{E}$ if and only if $(j,i)\in \mathcal{E}$. We consider $M=120$ content creators with random initial opinions and set the opinion dimension to $n=3$. All interaction parameters governing the user-creator and creator-user dynamics in \eqref{eqn:dynamics} are detailed in \Cref{app:real_data_params}. To generate initial user opinions that reflect the homophilic structure inherent in social networks, we employ spectral clustering \cite{spectral_clustering} to identify network communities, then assign similar opinions to users within the same community. The details can be found in \Cref{app:Ego_facebook}.}

\subsection{Experimental Results on Real Dataset}
\textcolor{black}{
The RS operates under a \texttt{top-$k$} recommendation strategy with $k=5$. 
% The impact of varying $k$ on all performance settings is analyzed in \cref{app:variational_k_real_dataset}. 
\Cref{fig:real_Dataset_3d} displays the user opinion landscapes after $t=20$ time-steps for (i) a greedy RS with $d=0$ and (ii) a hybrid RS with $d=3$. For illustration purposes, we omitted the creators opinions in the $3$-dimensional figure on the left side and only display every fourth opinion for the the users and the creators opinions respectively in the figures showing the projections on the $xy$-, $xz$- and $yz$-planes.}
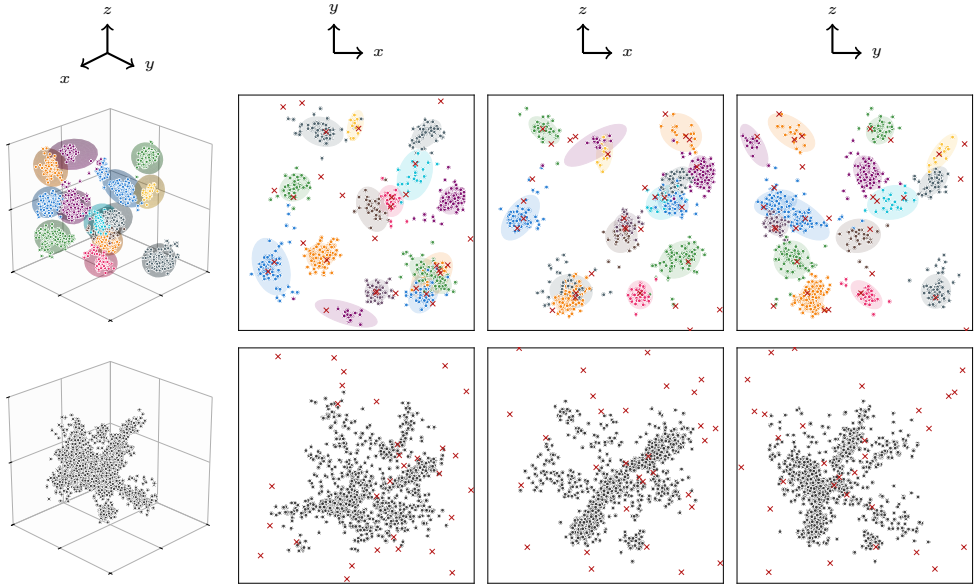
\begin{figure*}[htbp]
    \centering
    \setlength{\tabcolsep}{2pt} % Minimal spacing between columns
    
    \begin{tabular}{c@{\hspace{0.1cm}}c@{\hspace{0.1cm}}c@{\hspace{0.1cm}}c}
        % Header row
        $ \begin{tikzpicture}[baseline=-0.5ex, scale=0.3]
            % Draw axes
            \draw[->, thick] (0,0) -- (1.2,-0.6) node[right] {\tiny $y$};
            \draw[->, thick] (0,0) -- (0,1.3) node[above] {\tiny $z$};
            \draw[->, thick] (0,0) -- (-1.2,-0.6) node[below left] {\tiny $x$};
        \end{tikzpicture}$ & 
        $ \begin{tikzpicture}[baseline=-0.5ex, scale=0.3]
            % XY plane axes
            \draw[->, thick] (0,0) -- (1.3,0) node[right] {\tiny $x$};
            \draw[->, thick] (0,0) -- (0,1.3) node[above] {\tiny $y$};
        \end{tikzpicture}$ & 
        $ \begin{tikzpicture}[baseline=-0.5ex, scale=0.3]
            % XZ plane axes
            \draw[->, thick] (0,0) -- (1.3,0) node[right] {\tiny $x$};
            \draw[->, thick] (0,0) -- (0,1.3) node[above] {\tiny $z$};
        \end{tikzpicture}$ & 
        $ \begin{tikzpicture}[baseline=-0.5ex, scale=0.3]
            % YZ plane axes
            \draw[->, thick] (0,0) -- (1.3,0) node[right] {\tiny $y$};
            \draw[->, thick] (0,0) -- (0,1.3) node[above] {\tiny $z$};
        \end{tikzpicture}$\\[1mm]
        
        % First row with images
        \includegraphics[width=0.23\textwidth]{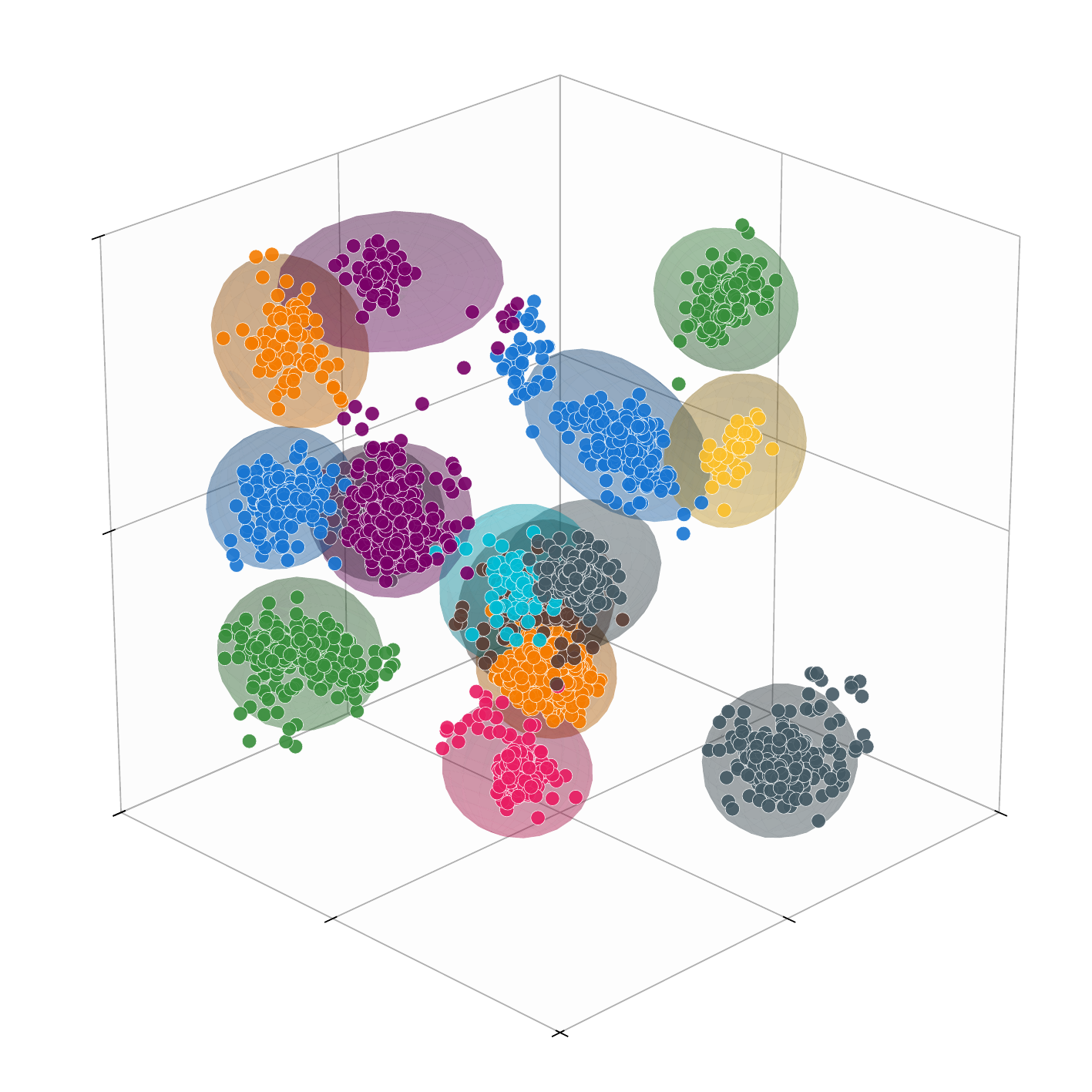} &
        \includegraphics[width=0.23\textwidth]{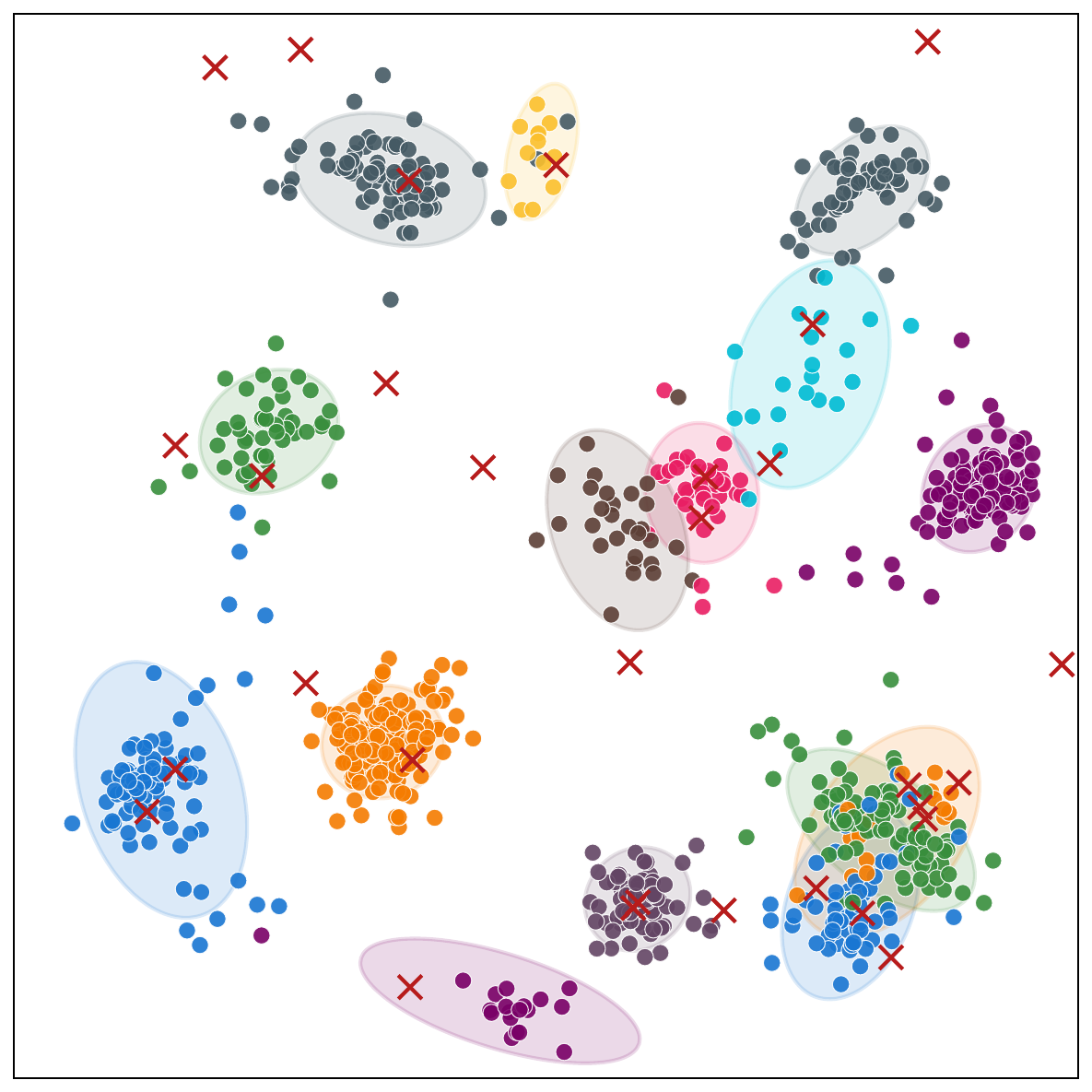} &
        \includegraphics[width=0.23\textwidth]{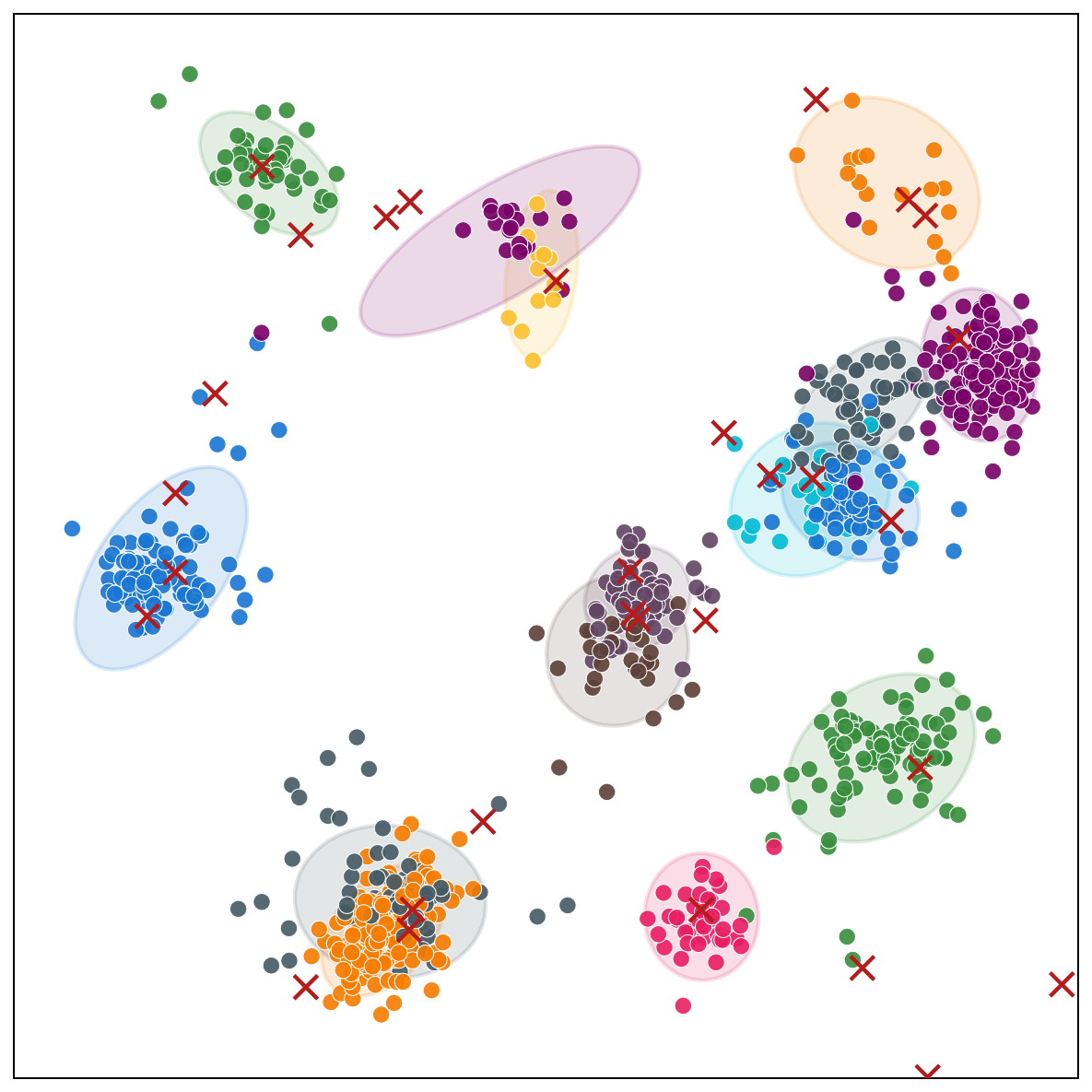} &
        \includegraphics[width=0.23\textwidth]{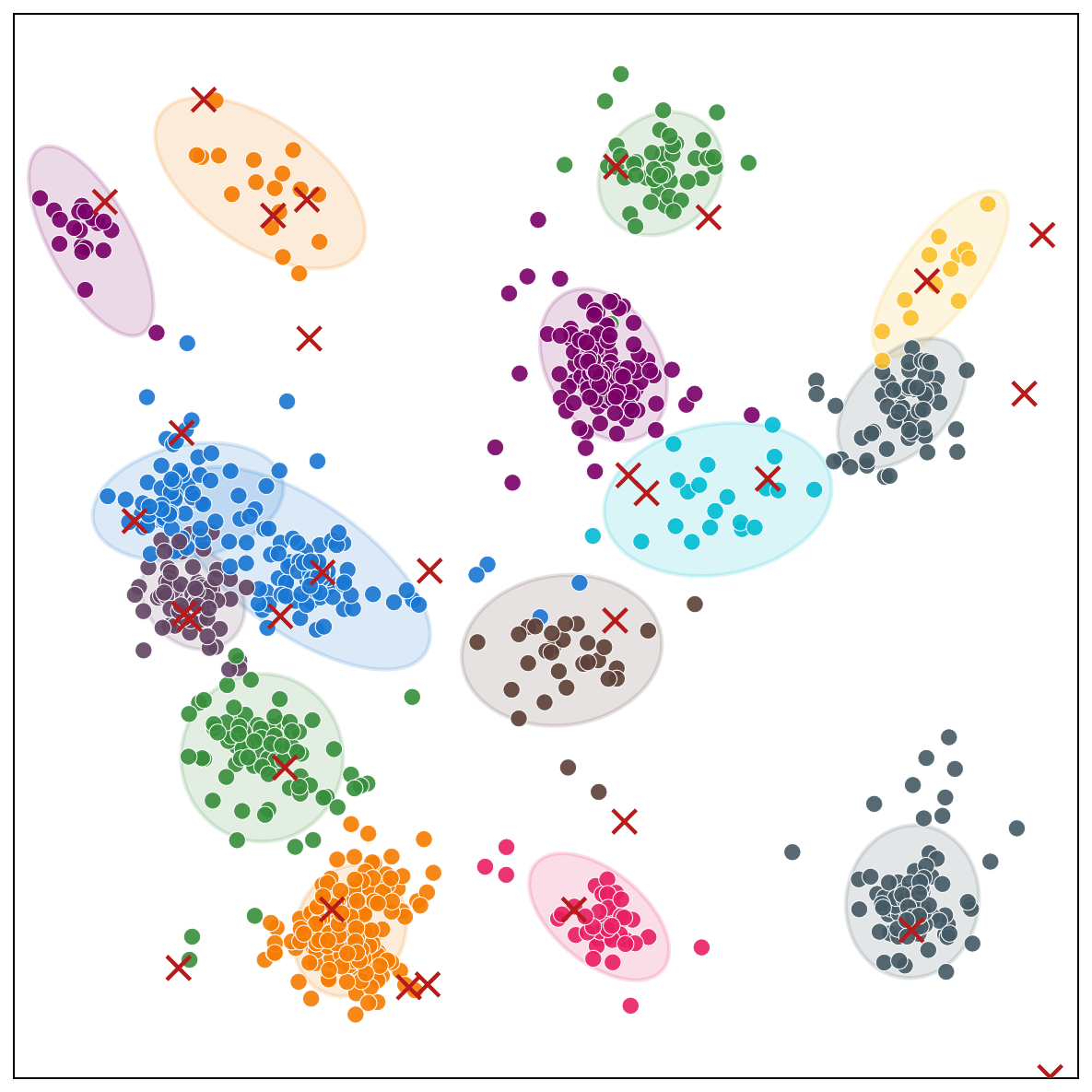} 
    \end{tabular}
    \begin{tabular}{c@{\hspace{0.1cm}}c@{\hspace{0.1cm}}c@{\hspace{0.1cm}}c}
        % Header row
        
        % First row with images
        \includegraphics[width=0.23\textwidth]{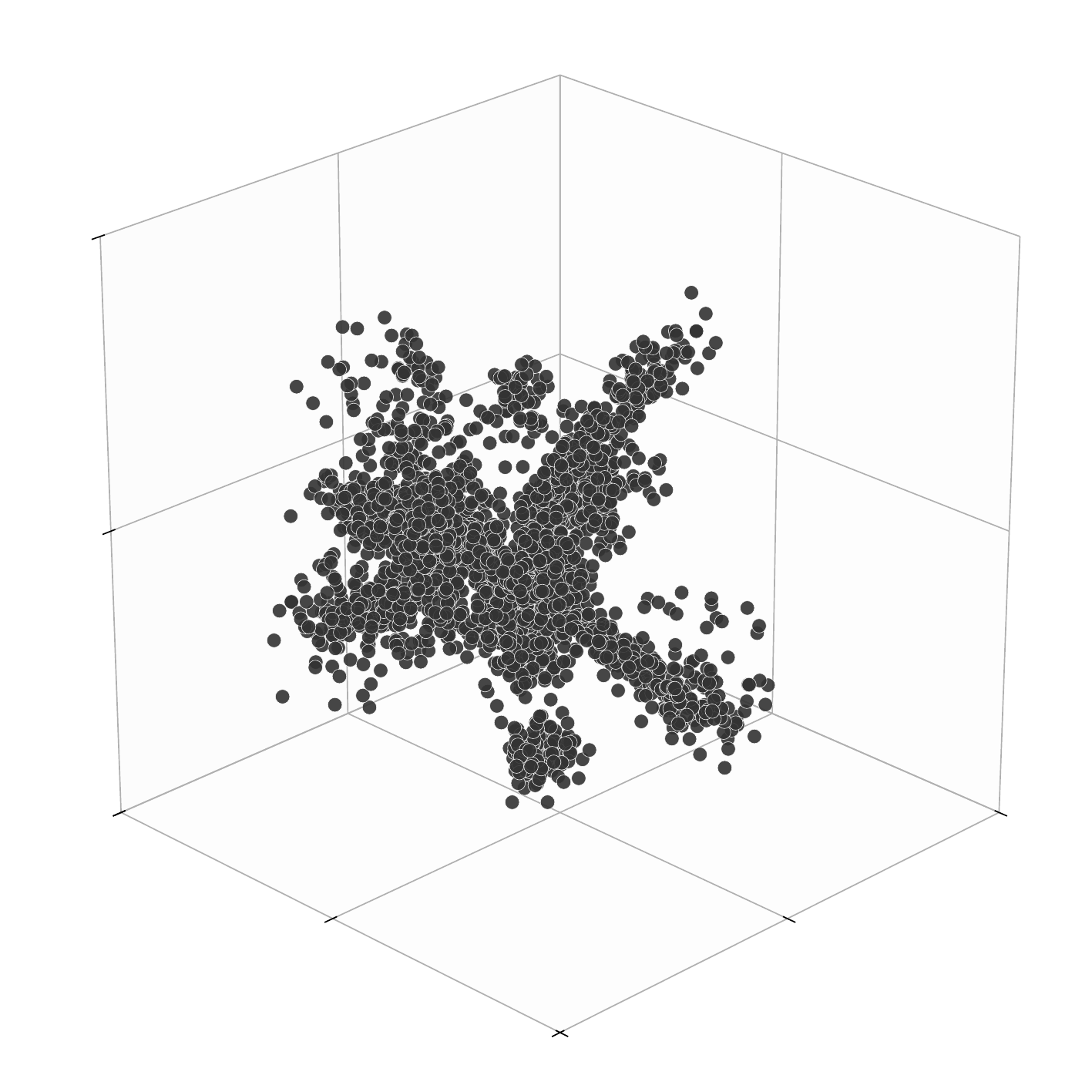} &
        \includegraphics[width=0.23\textwidth]{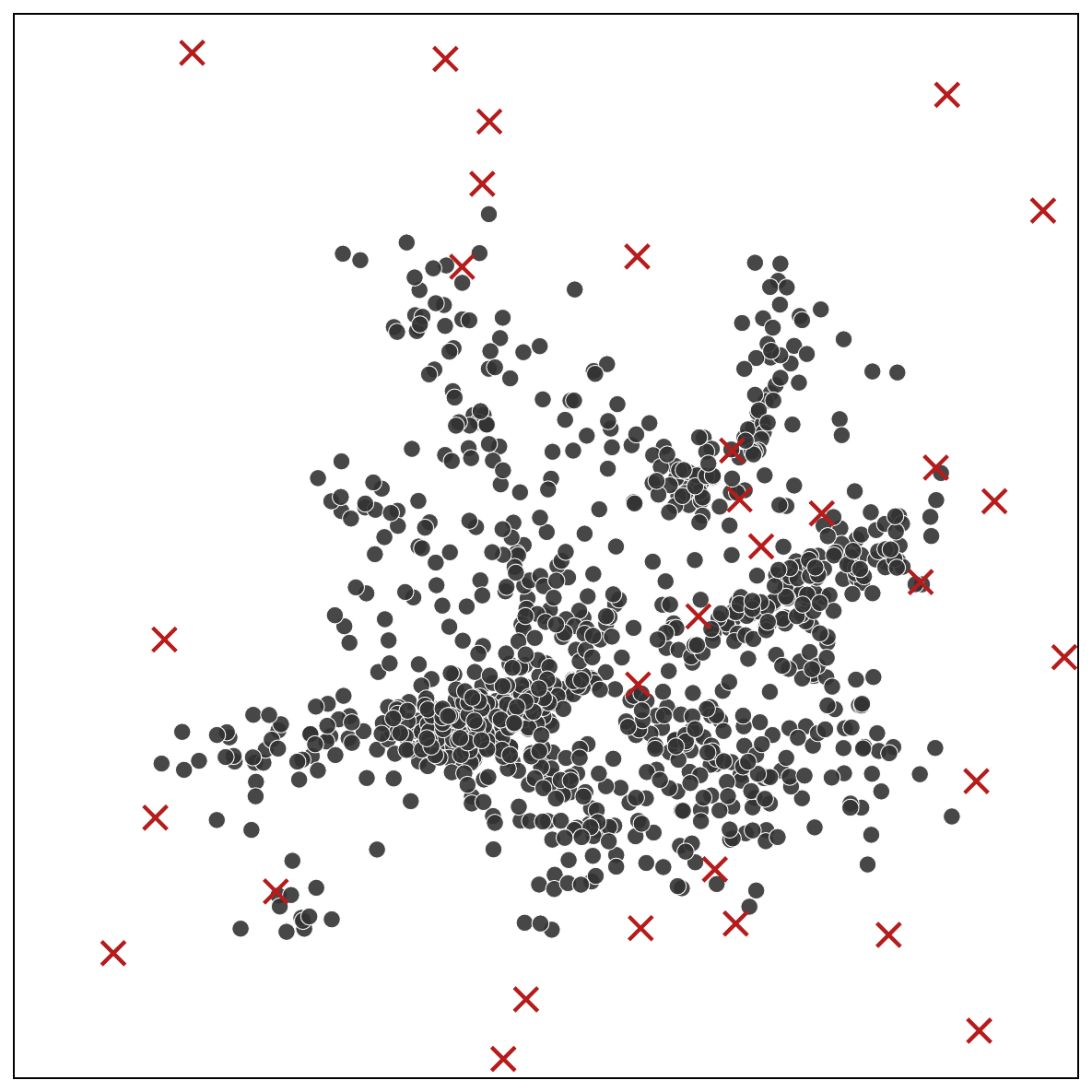} &
        \includegraphics[width=0.23\textwidth]{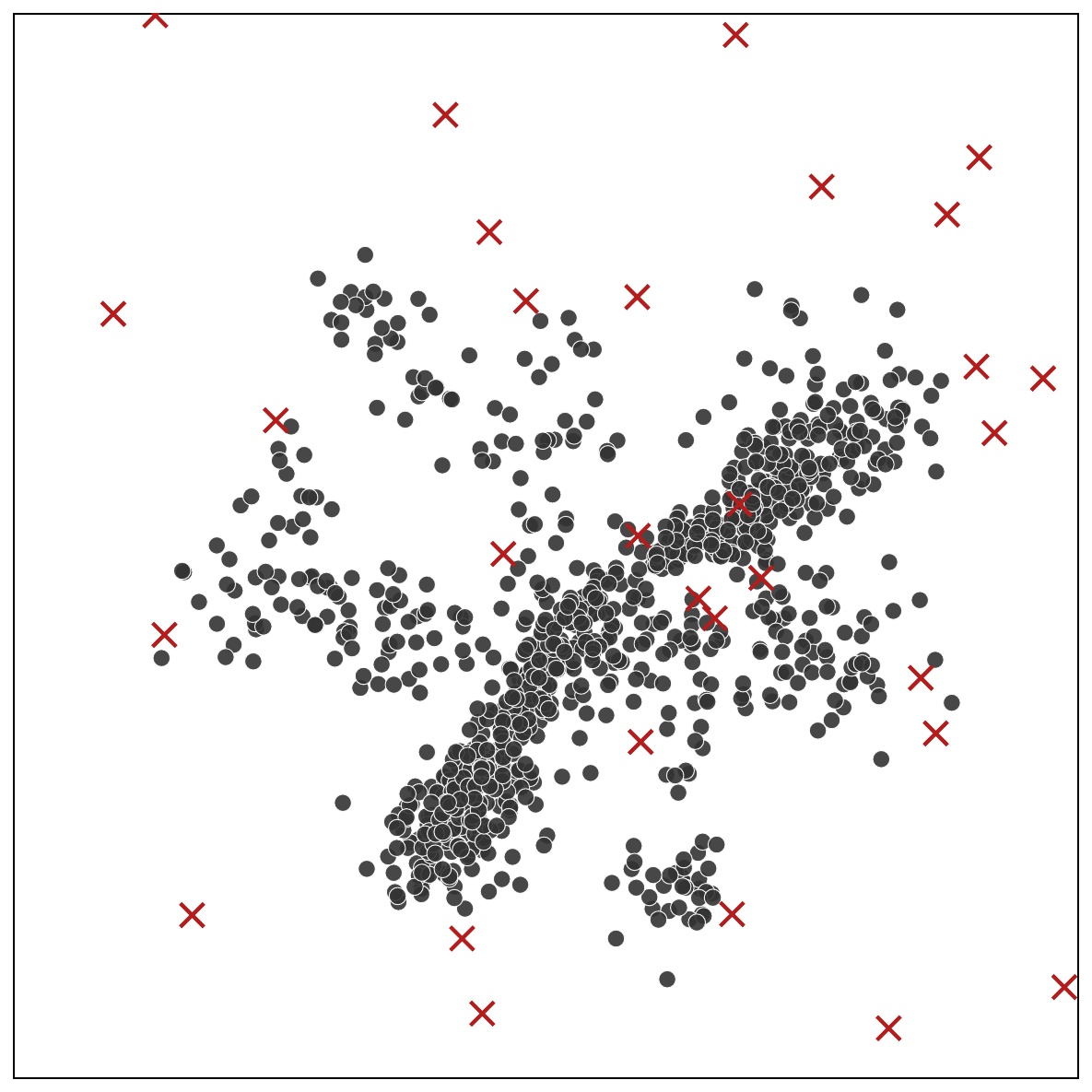} &
        \includegraphics[width=0.23\textwidth]{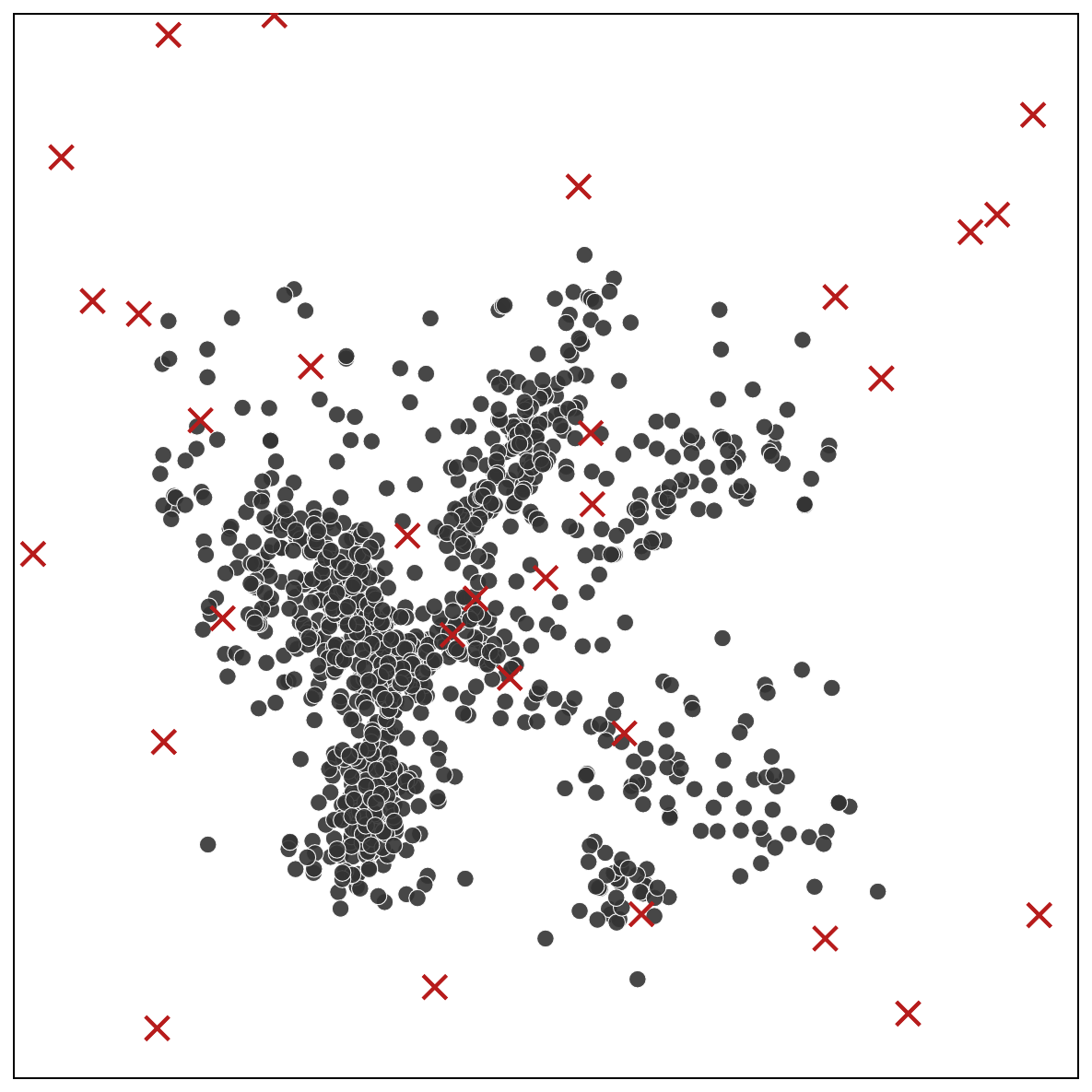} \\[1mm]
    \end{tabular}
    \caption{Snapshots after $t=20$ timesteps of the opinion environment simulated with a localized region (i) $d=0$, first row, (ii) $d=3$, second row. \textcolor{red}{$\times$} denotes the creator, $\bullet$ the users respectively.}
    \label{fig:real_Dataset_3d}
\end{figure*}

\newpage
\bibliography{iclr2026_conference}
\bibliographystyle{iclr2026_conference}

\appendix

\section{Multi-Topic Friedkin-Johnsen Model}\label{app:multi_topic}
The multi-topic FJ model \cite{multi_topic_FJ} extends the classical scalar model in \eqref{eqn:social-dyn} to the case where each user holds opinions on multiple topics simultaneously. Let $u_i^t \in \mathbb{R}^n$ denote the opinion vector of user $i$ at time $t$, where each entry corresponds to a distinct topic. Stacking all $N$ users’ opinions into a single vector ${\bf u} \in \mathbb{R}^{Nn}$, the opinion update rule can be written as
\begin{equation} 
    f(\mathbf{u}^{t}) = ((I_N-\Lambda)\hat{A}\otimes C)\mathbf{u}^{t} + (\Lambda\otimes I_n) \mathbf{u}^{0},
\end{equation}
where $\hat{A}\in \mathbb{R}^{N \times N}$ is the row-stochastic influence matrix describing interpersonal influence in the network, $\Lambda \in \mathbb{R}^{N \times N}$ is the diagonal susceptibility matrix capturing how attached each influence is to their own prejudice $u^0 \in \mathbb{R}^{Nn}$ versus the social influence, and $\otimes$ denotes the Kronecker product. The matrix $C$ is a correlation matrix among different topics.  We consider the case of uncorrelated topics, and hence set $C = I_n$. In this special case, the model reduces to
\begin{equation} \label{eqn:social-dyn_multi_t}
    \green{f}(\mathbf{u}^{t}) = \green{((I_N-\Lambda)\hat{A}\otimes I_n)}\mathbf{u}^{t} + \green{(\Lambda\otimes I_n)} \mathbf{u}^{0}.
\end{equation}

\section{Environment Variables for Synthetic Dataset} 
\subsection{Parameters for user-creator Dynamics}\label{appendix:model_parameters}
The parameters governing the dynamics in \eqref{eqn:dynamics} are sampled independently from uniform distributions with bounds given in Table~\ref{tab:simulation_parameters}.
\begin{table}[htbp]
\centering
\caption{Simulation Parameters for Uniform Distribution Sampling}
\label{tab:simulation_parameters}
\begin{tabular}{lcc}
\toprule
\textbf{Parameter} & \textbf{Lower Bound} & \textbf{Upper Bound} \\
\midrule
\multicolumn{3}{l}{\textit{User Parameters}} \\
User Stubbornness $\Lambda_i$ & 0.0 & 0.5 \\
User Self-Influence $A_{ii}$& 0.5 & 0.8 \\
Recommender Influence $B_{ij}$ & 0.2 & 0.8 \\
Neighbor Influence $A_{ij}$ & 0.025 & 0.05 \\
\midrule
\multicolumn{3}{l}{\textit{Creator Parameters}} \\
Creator Stubbornness  $\Gamma_j$ & 0.0 & 0.5 \\
Creator Self-Influence $E_j$& 0.5 & 0.8 \\
User-Creator Influence $C_{j}$& 0.2 & 0.8 \\
\bottomrule
\end{tabular}
\end{table}

The user-creator influence is evenly distributed among the audience set of creator $j$; specifically, for creator $j$ with audience set $\mathcal{F}_j$, each user $i \in \mathcal{F}_j$ exerts influence $C_{ji} = C_j / |\mathcal{F}_j|$. The overall social influence on user $i$ is determined by summing the influences from all neighbors. Each user is influenced by exactly one creator. Thus, referring to the FJ model in \cref{app:multi_topic}, we obtain the stochastic constraints: $A_{ii} + B_{ij} + \sum_{j =0}^{N-1} A_{ij} = 1$ for users and $C_j + E_j = 1$ for creators. 

\subsection{User-User interaction Probability}\label{appendix:connection_prob_synthetic}
For each user $j$, we assume the connection to any other user with a probability that is given as:
$$\text{Prob(user $j$ influences user $i$)} = \exp(-\delta ||u^0_i - u_j^0||_2^2)$$
where $u_i^0, u_j^0 \in [-1,1]^n$ are the opinion vectors of users $i$ and $j$ at time $0$, and $\delta > 0$ is a parameter controlling the connectivity of the network. Different parameters of $\delta$ lead to different number of connections. Different choices of parameter $\delta$ and the resulting average node degrees are displayed in Table~\ref{tab:gamma_connectivity}. We choose $\delta=9$, to recover $11$ neighbors.
\begin{table}[htbp]
\centering
\caption{Network connectivity for different parameters $\delta$ with $N= 600$ users initialized randomly}
\label{tab:gamma_connectivity}
\begin{tabular}{cc}
\toprule
\textbf{Parameter} $\delta$ & \textbf{Average Connections} \\
\midrule
6 & 21 \\
7 & 17 \\
8 & 14 \\
9 & 11 \\
\bottomrule
\end{tabular}
\end{table}

\section{Results for variational environment on synthetic data}\label{appendix:k_analysis_synthetic}

\subsection{Varying the Simulation Parameters for the dynamics}
We increase the number of social interactions by setting $\delta=6$, yielding an average of 21 connections per user. The expanded social network exposes users to a broader spectrum of opinions through peer interactions. \Cref{fig:varying_lambda_50_soc1} and \cref{fig:varying_lambda_500_soc1} present the global clusterization and satisfaction metrics after 50 and 500 timesteps, respectively. The results demonstrate that increased social connectivity mitigates clusterization, even under the greedy recommender system ($d=0$).
\begin{figure}[h!]
    \centering
    \begin{subfigure}[b]{0.36\textwidth}
        \centering
        \includegraphics[width=\textwidth]{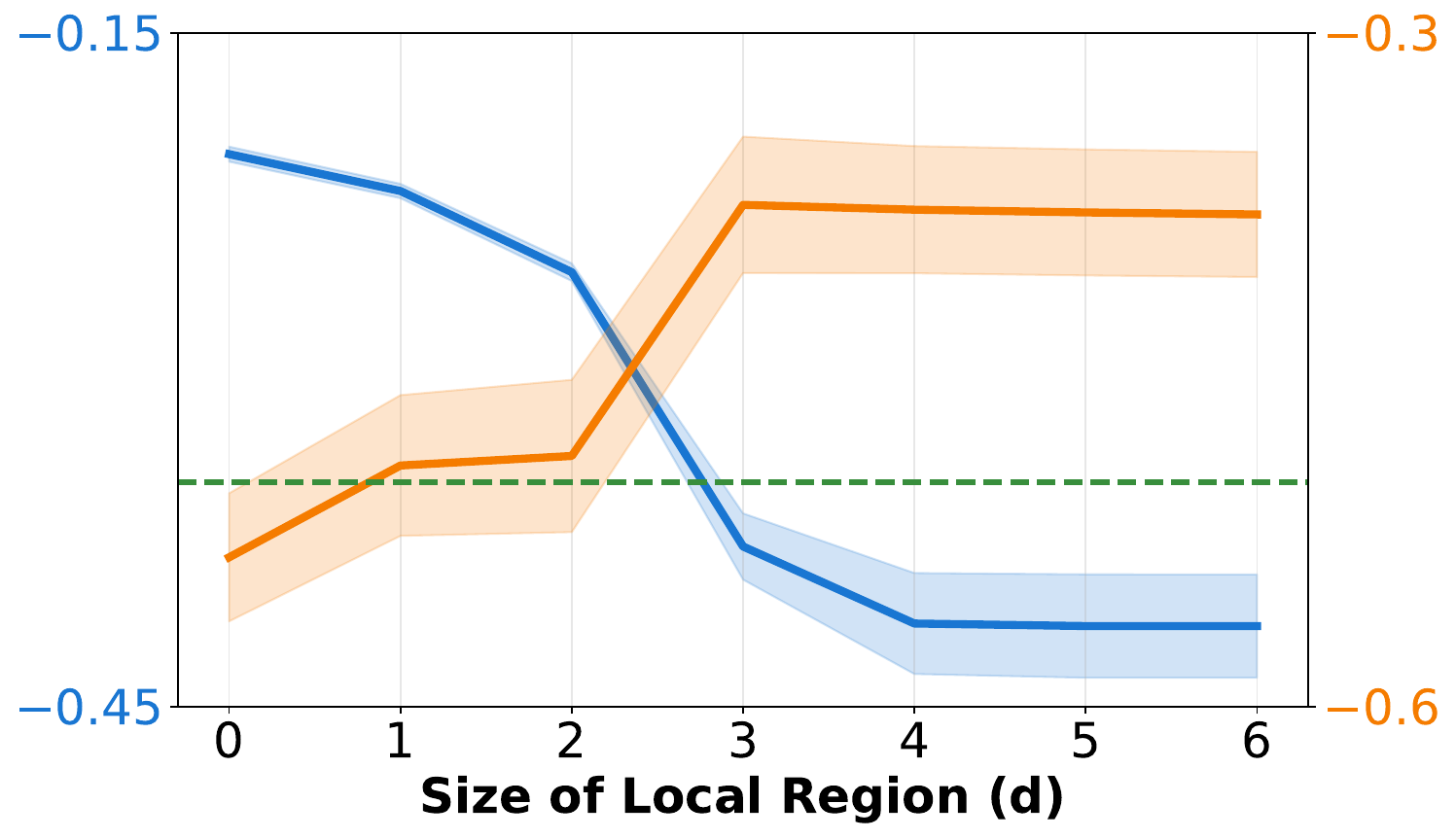}
        \caption{After 50 timesteps}
        \label{fig:varying_lambda_50_soc1}
    \end{subfigure}
    \hfill
    \begin{subfigure}[b]{0.36\textwidth}
        \centering
        \includegraphics[width=\textwidth]{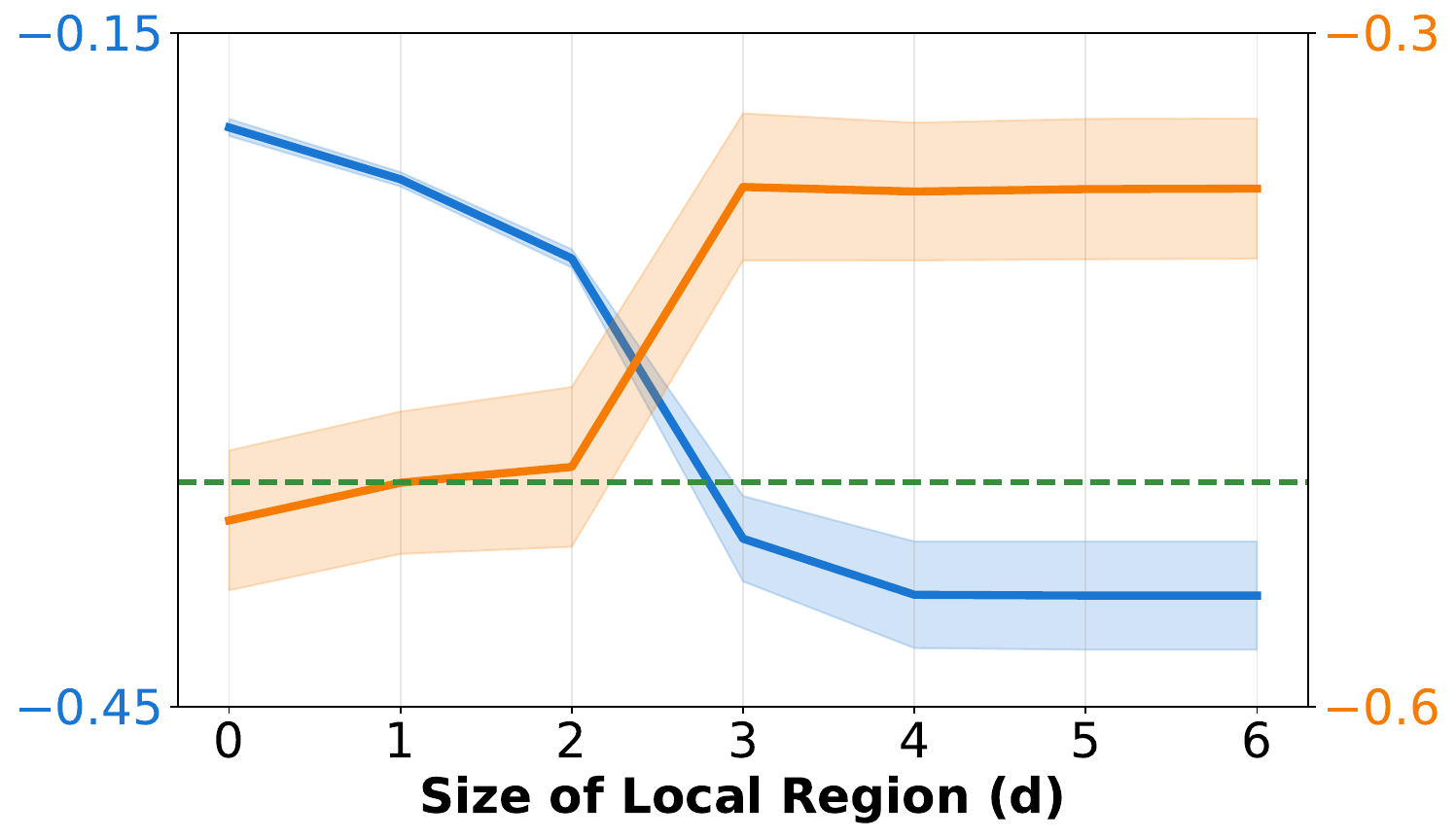}
        \caption{After 500 timesteps}
        \label{fig:varying_lambda_500_soc1}
    \end{subfigure}
    \hfill
    \begin{subfigure}[b]{0.25\textwidth}
        \centering
        \raisebox{1.5cm}{%  % Adjust this value to center vertically
        \small
            \begin{tabular}{cl}
                \textcolor[RGB]{249,124,0}{\rule{0.2cm}{0.1cm}} & Neg. Clusterization \\[0.2cm]
                \textcolor[RGB]{25,118,210}{\rule{0.2cm}{0.1cm}} & Satisfaction \\ [0.2cm]
                \textcolor[RGB]{56,142,60}{\rule{0.2cm}{0.1cm}} & Clusterization Thresh.
            \end{tabular}
        }
        \caption*{}  % Empty caption to maintain alignment
    \end{subfigure}
    \caption{Global clusterization and global user satisfaction plotted as $d$ varies after (a) 50 and (b) 500 timesteps with more social interactions as opposed to \cref{fig:varying_lambda}. Clusterization thresh is set to $-0.5$, for which clusters are no longer distinguishable.}
    \label{fig:varying_lambda_soc}
\end{figure}

\subsection{Varying k for top-k}

The opinion dynamics under socially-aware recommender systems are examined for localization parameters $d \in \{0, 3\}$ and recommendation set sizes $k \in \{1, 2, 3, 4\}$. Figures~\ref{fig:dynamics_k1}--\ref{fig:dynamics_k4} display the corresponding opinion landscapes.
For $k=1$, the dynamics becomes stationary after $t=20$ (\cref{fig:dynamics_k1}). Cluster visualization is omitted given the low global clusterization score.
The number of clusters increases as $k$ decreases, reflecting reduced creator-user interaction diversity. The recommender system with $d=3$ produces fewer clusters than $d=0$ across all values of $k$.

\begin{figure*}[htbp]
    \centering
    \setlength{\tabcolsep}{2pt} % Minimal spacing between columns
    
    \begin{tabular}{c@{\hspace{0.1cm}}c@{\hspace{0.1cm}}c@{\hspace{0.1cm}}c}
        % Header row
        $t = 5$ & $t = 20$ & $t = 50$ & $t=500$\\[1mm]
        
        % First row with images
        \includegraphics[width=0.23\textwidth]{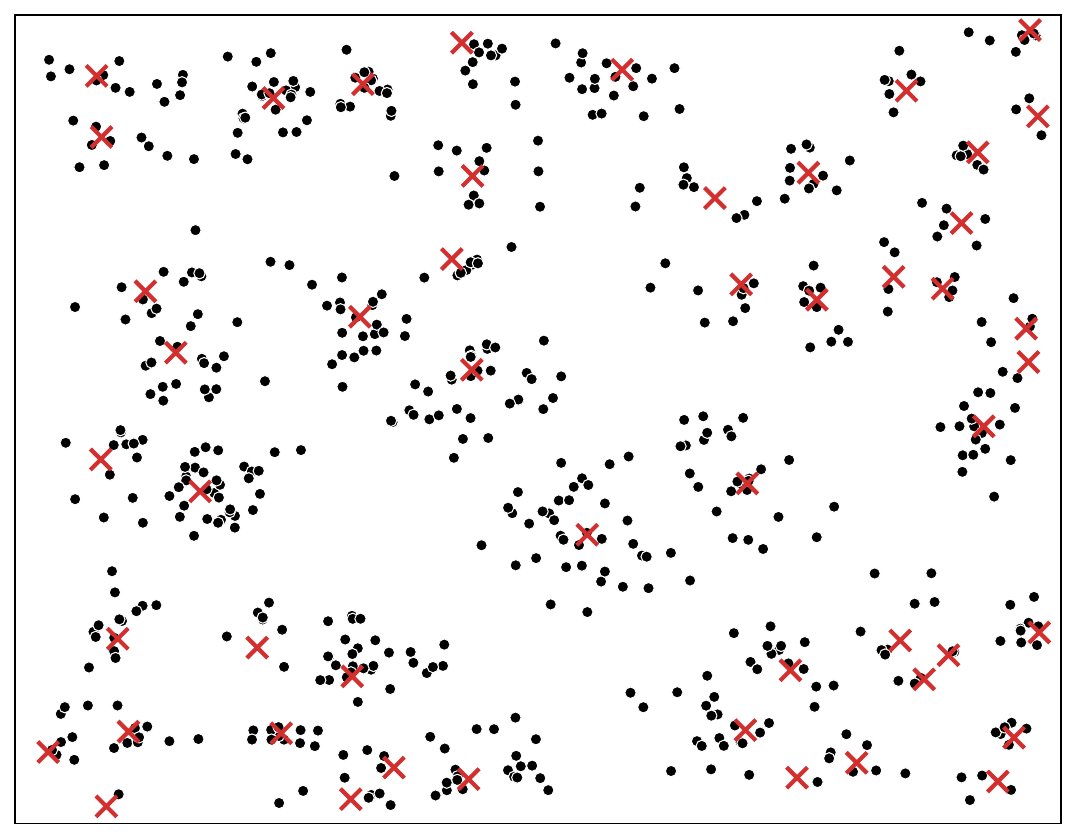} &
        \includegraphics[width=0.23\textwidth]{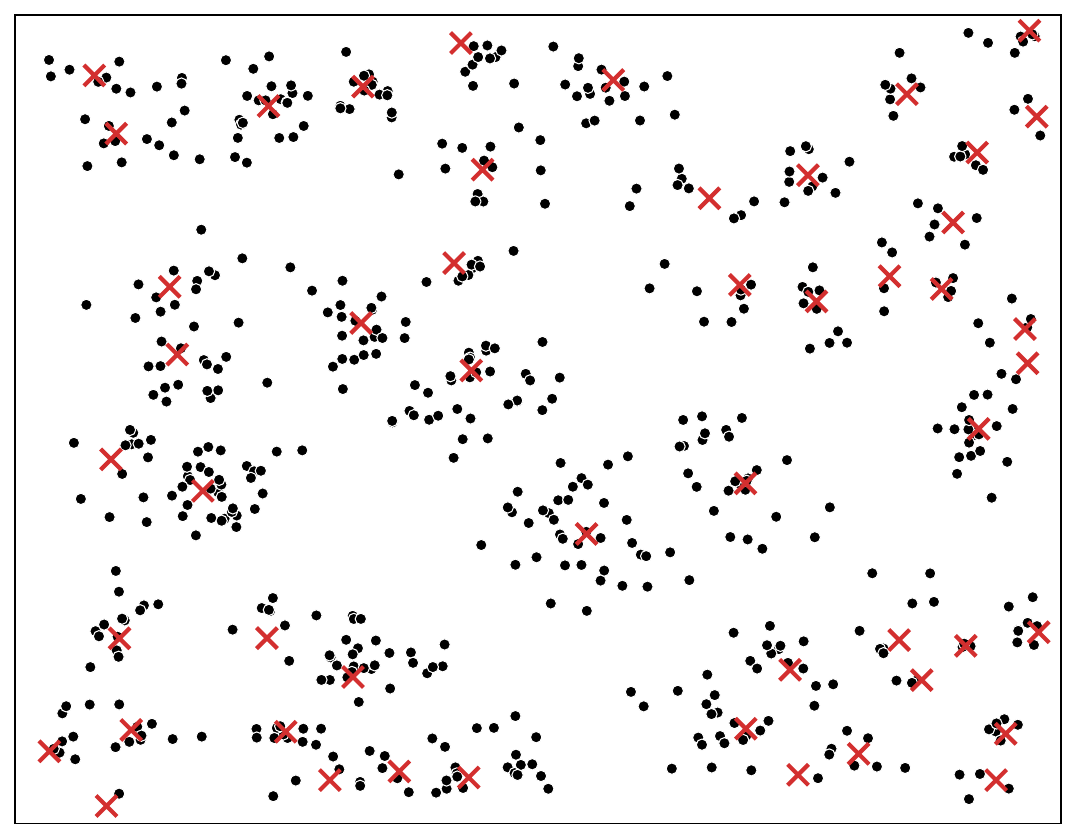} &
        \includegraphics[width=0.23\textwidth]{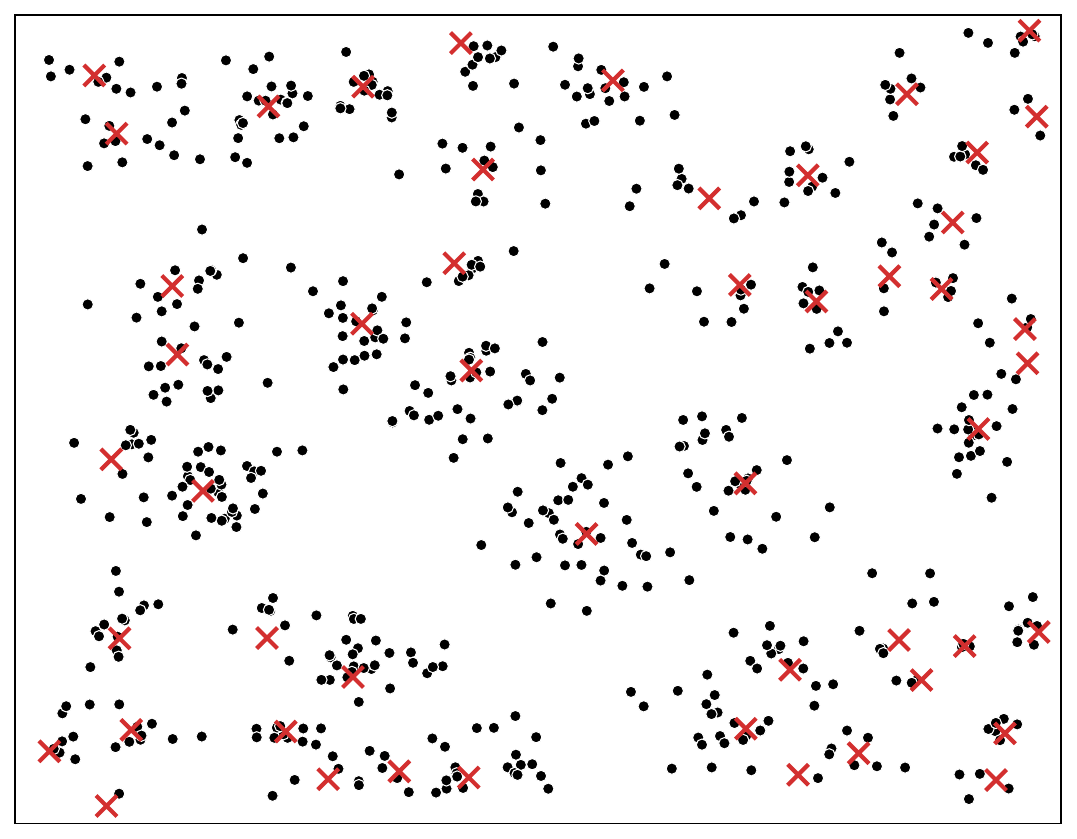} &
        \includegraphics[width=0.23\textwidth]{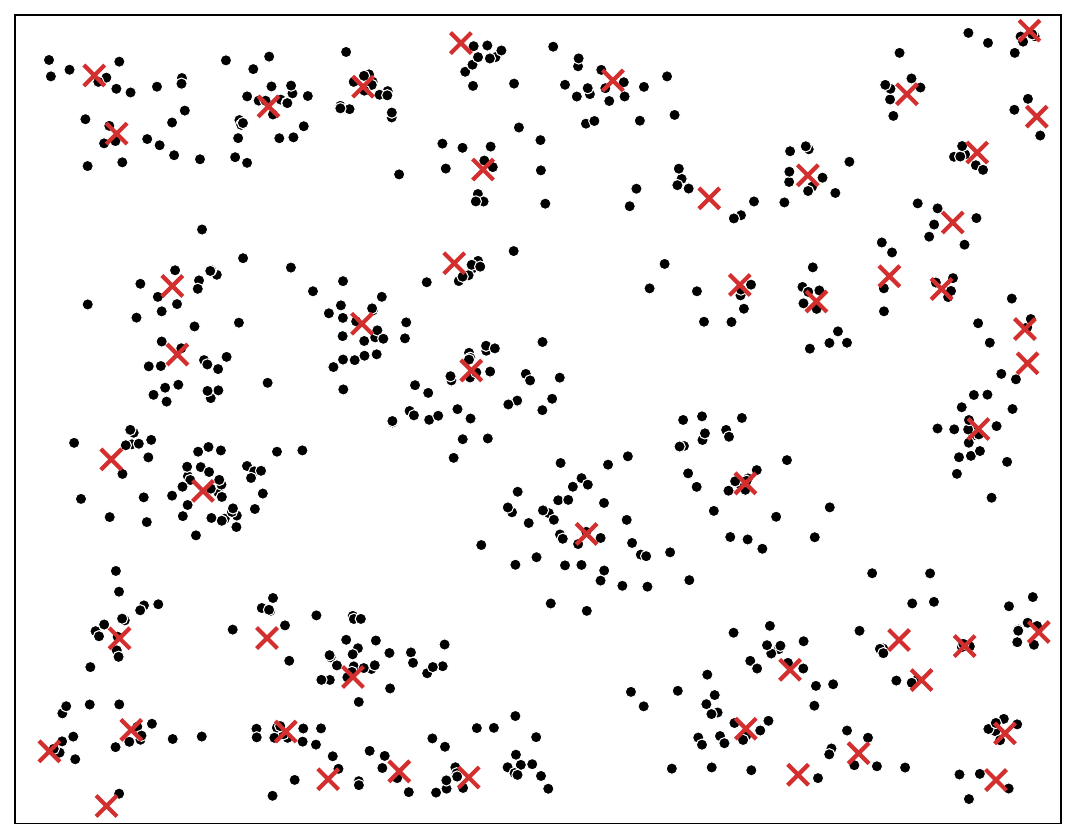}\\[1mm] 
        \end{tabular}
    \begin{tabular}{c@{\hspace{0.1cm}}c@{\hspace{0.1cm}}c@{\hspace{0.1cm}}c}
        % Header row
        
        % First row with images
        \includegraphics[width=0.23\textwidth]{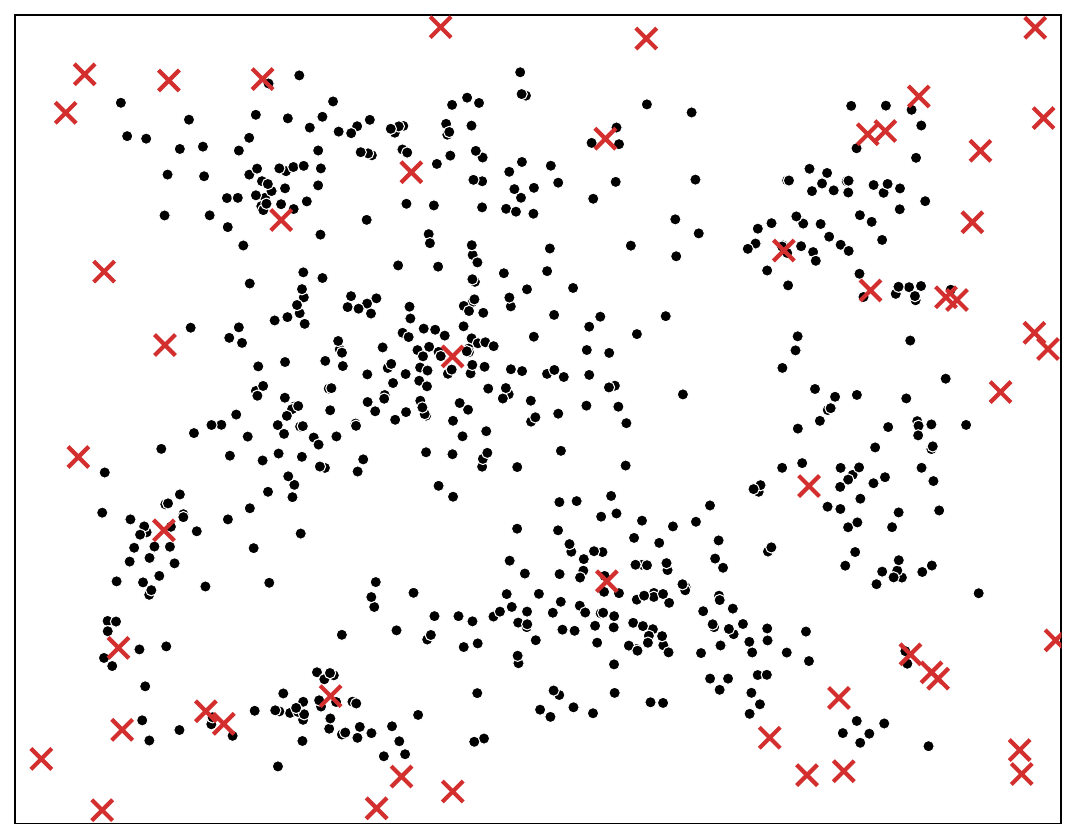} &
        \includegraphics[width=0.23\textwidth]{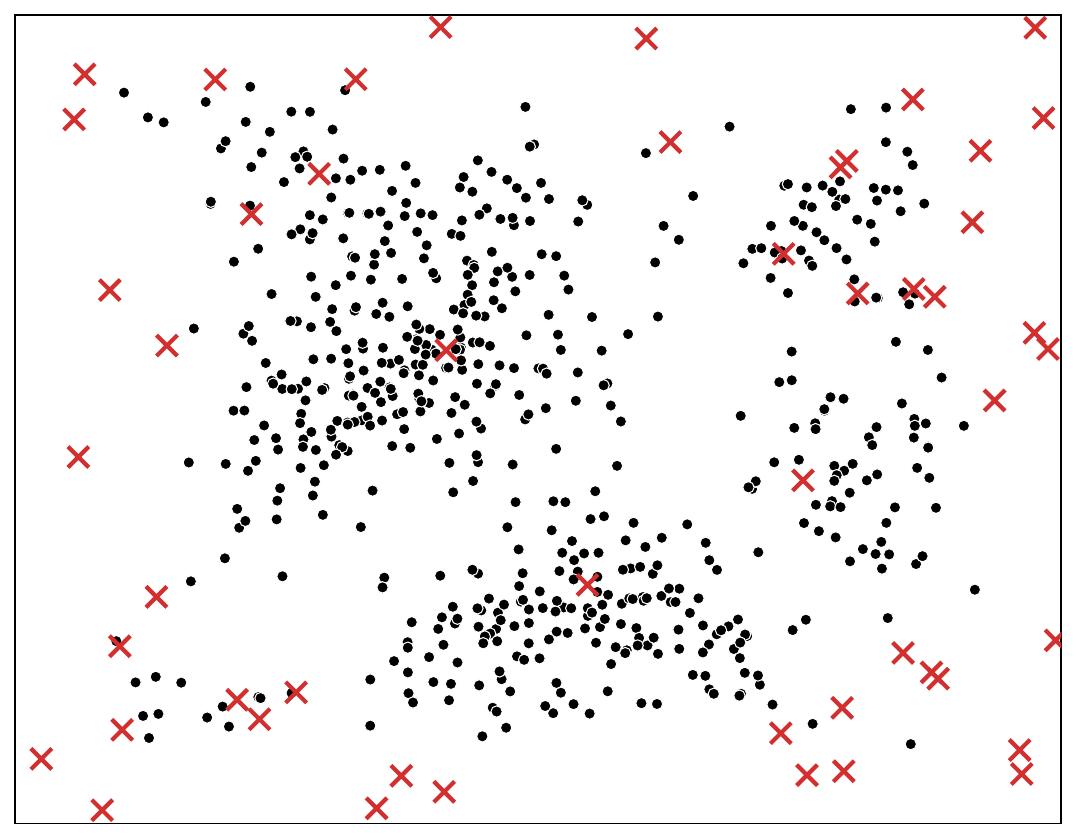} &
        \includegraphics[width=0.23\textwidth]{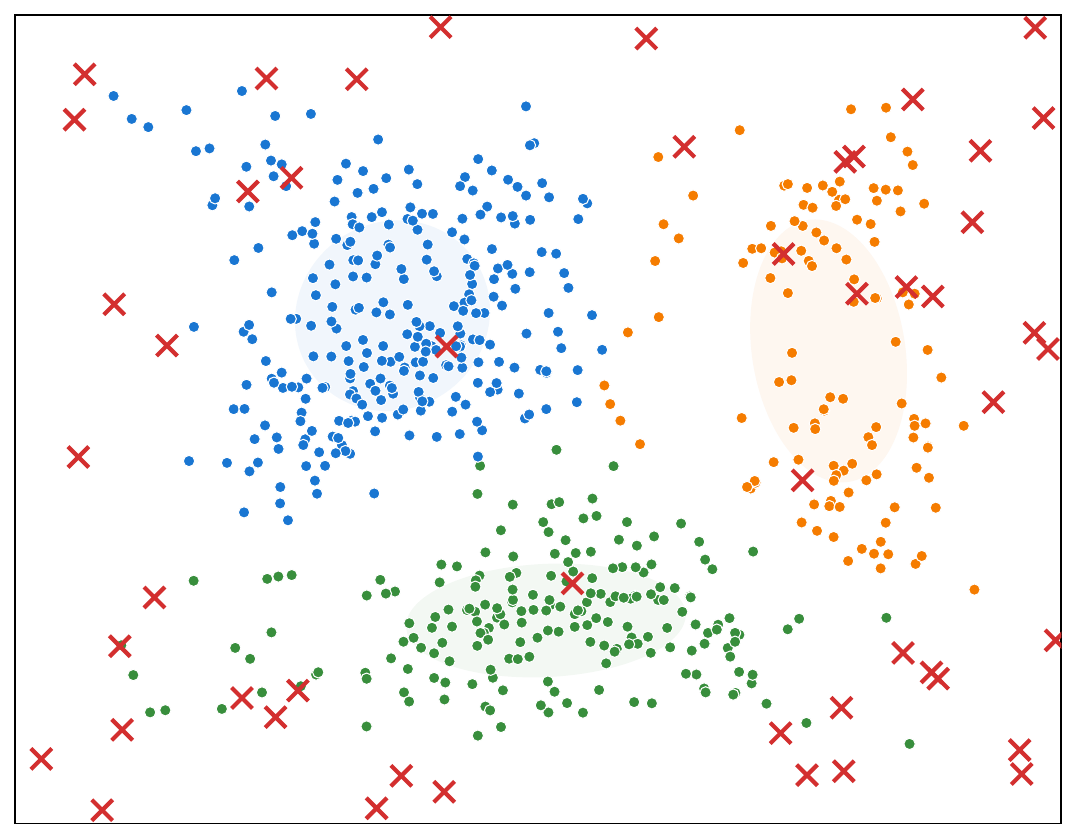} &  
        \includegraphics[width=0.23\textwidth]{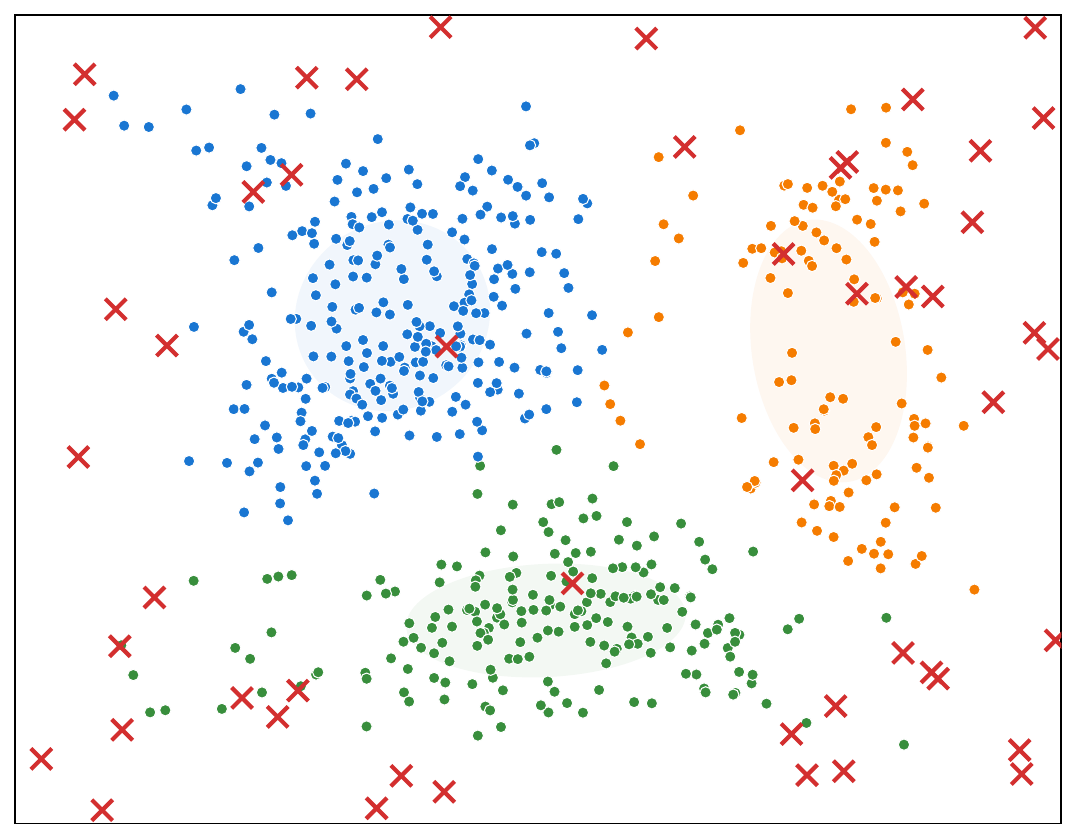}\\[1mm]
    \end{tabular}
    
    \caption{Snapshots of the opinion environment with $k=1$ simulated with a localized region $d=0$, first row,  $d=3$, second row. \textcolor{red}{$\times$} denotes the creator, $\bullet$ the users respectively.}
    \label{fig:dynamics_k1}
\end{figure*}

\begin{figure*}[htbp]
    \centering
    \setlength{\tabcolsep}{2pt} % Minimal spacing between columns
    
    \begin{tabular}{c@{\hspace{0.1cm}}c@{\hspace{0.1cm}}c@{\hspace{0.1cm}}c}
        % Header row
        $t = 5$ & $t = 20$ & $t = 50$ & $t = 500$\\[1mm]
        
        % First row with images
        \includegraphics[width=0.23\textwidth]{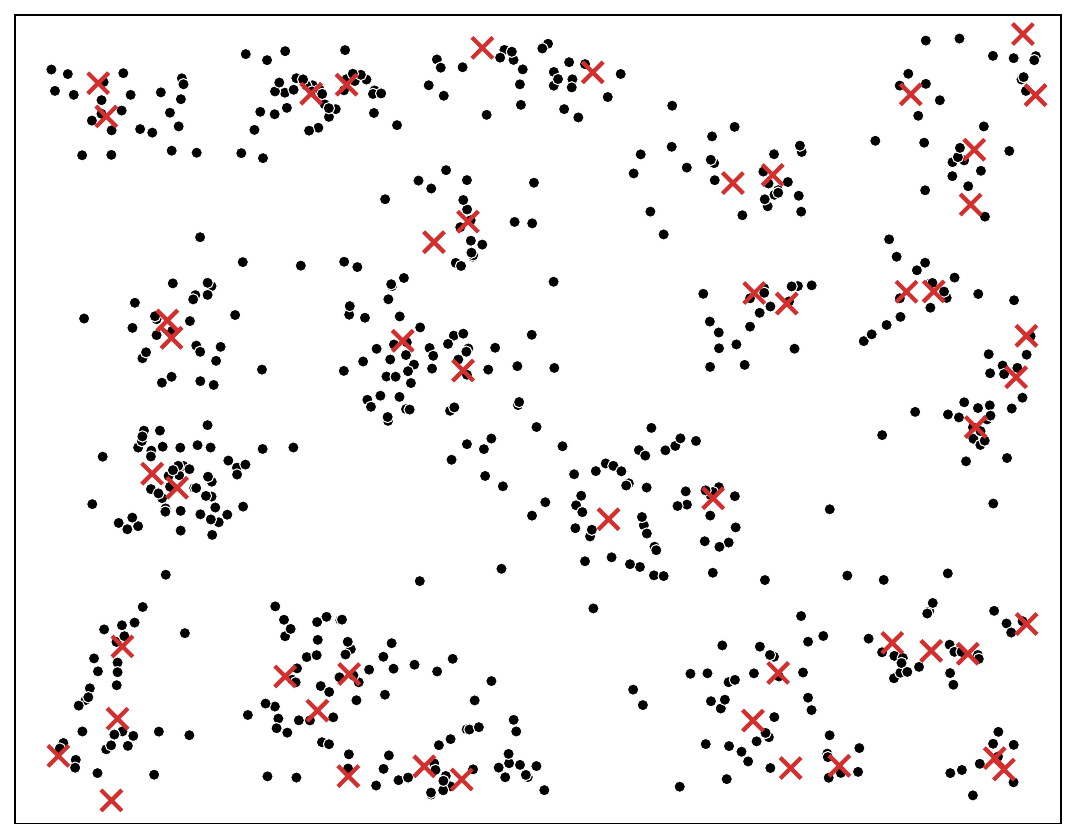} &
        \includegraphics[width=0.23\textwidth]{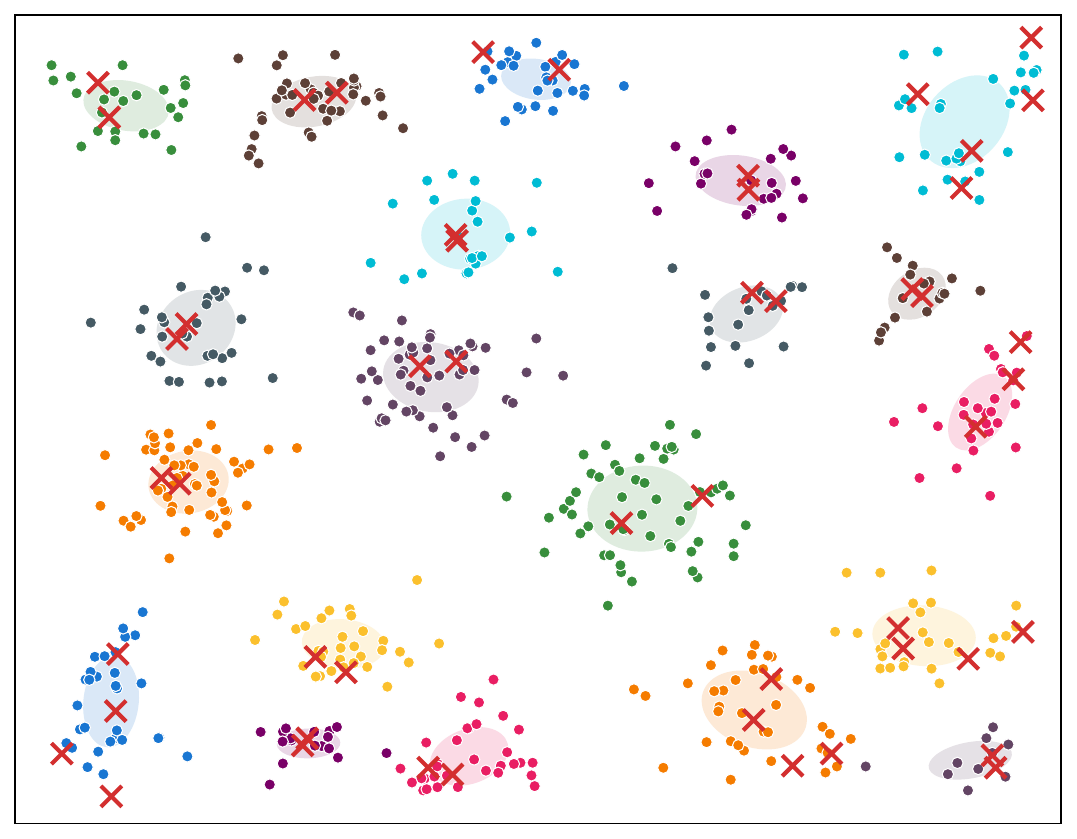} &
        \includegraphics[width=0.23\textwidth]{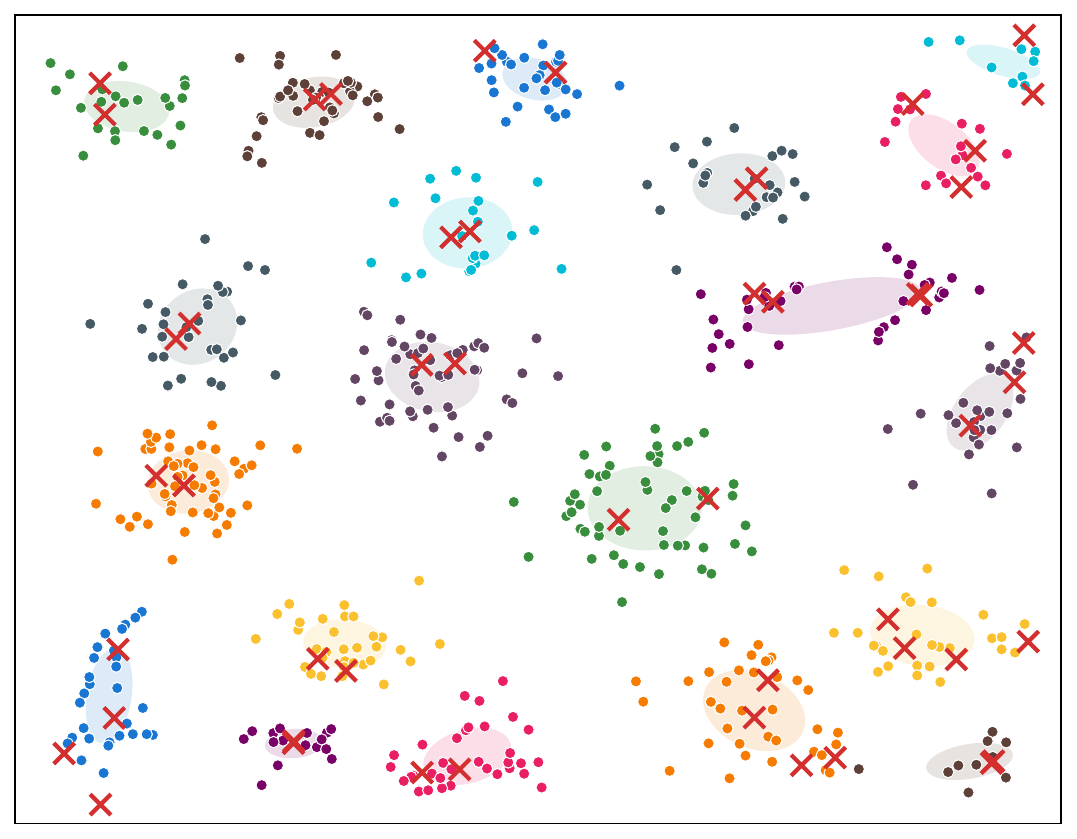} &
        \includegraphics[width=0.23\textwidth]{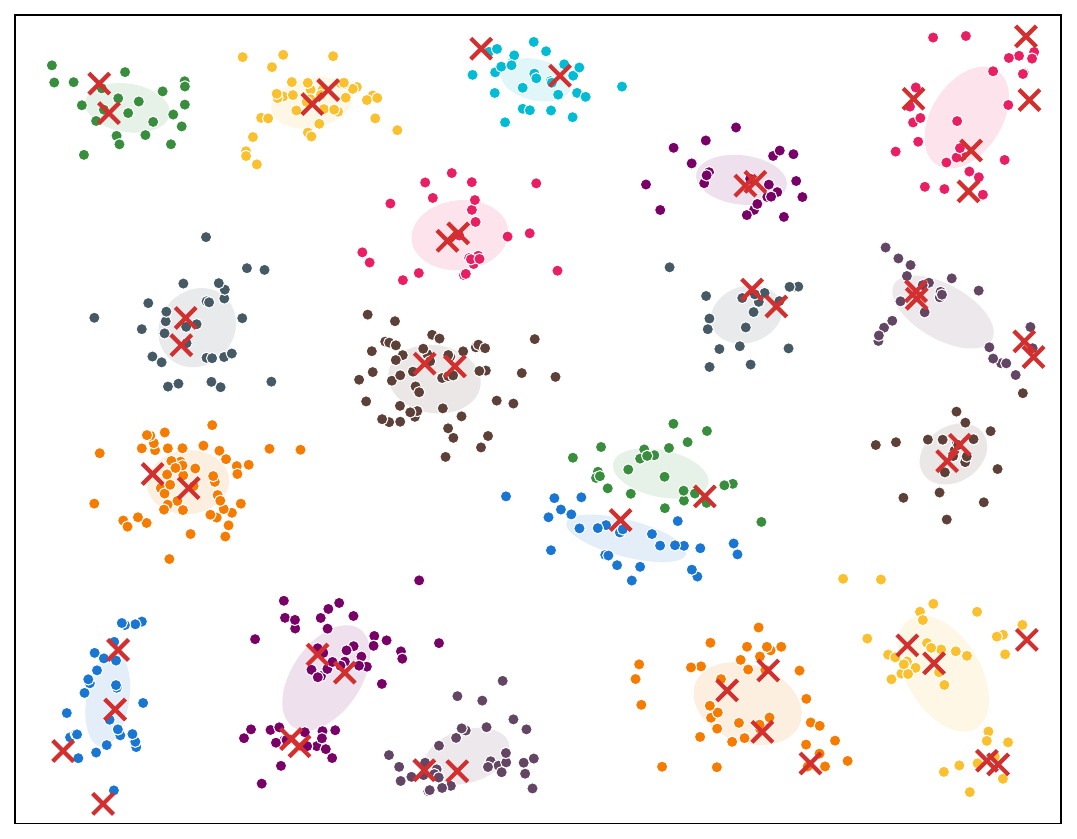}\\[1mm]
    \end{tabular}
    \begin{tabular}{c@{\hspace{0.1cm}}c@{\hspace{0.1cm}}c@{\hspace{0.1cm}}c}
        % Header row
        
        % First row with images
        \includegraphics[width=0.23\textwidth]{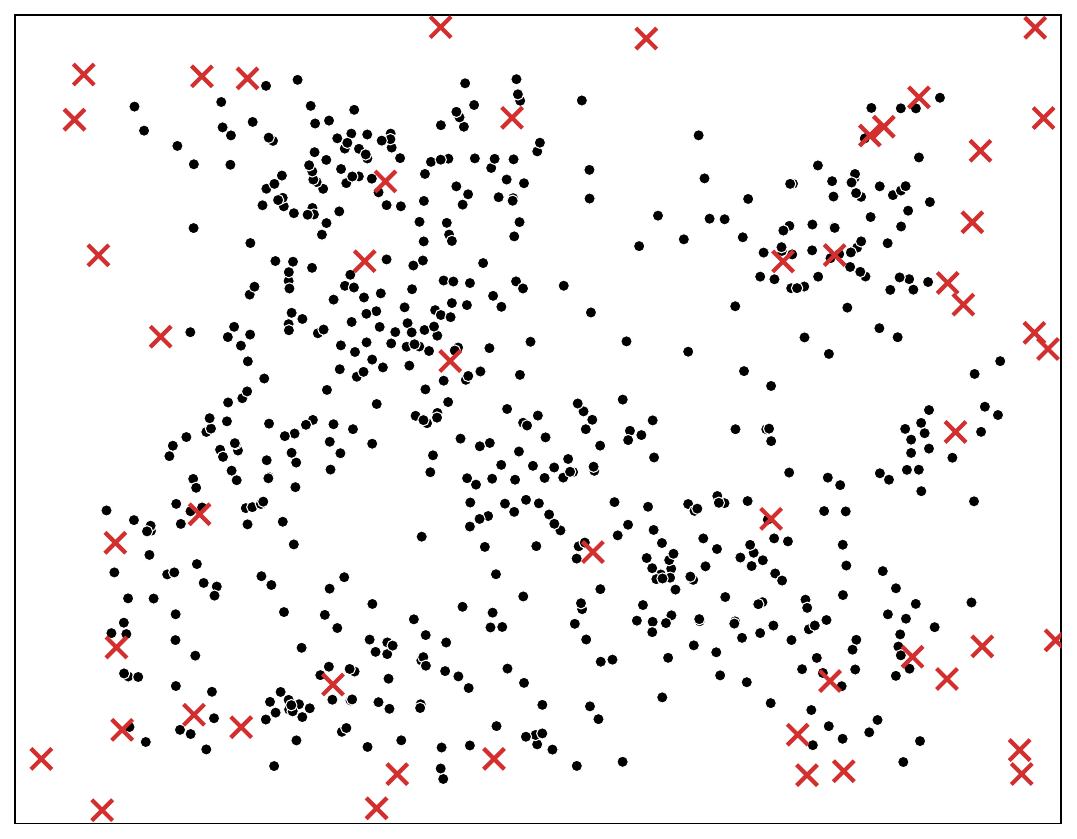} &
        \includegraphics[width=0.23\textwidth]{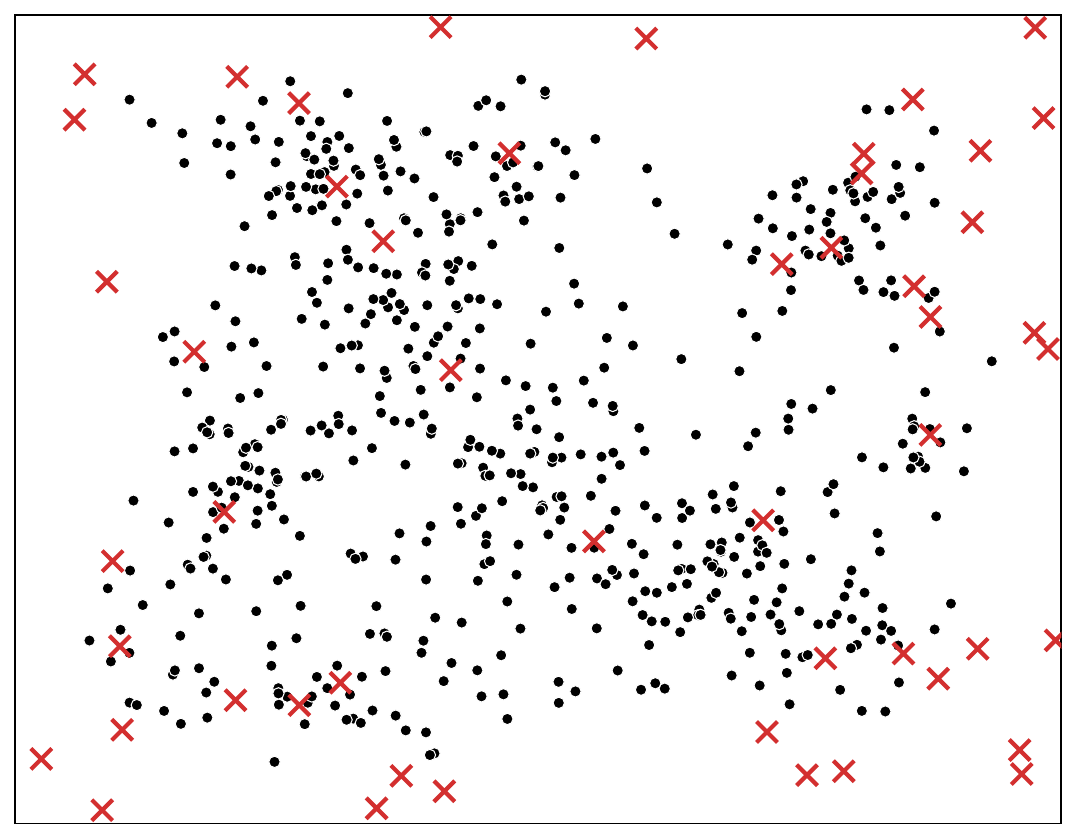} &
        \includegraphics[width=0.23\textwidth]{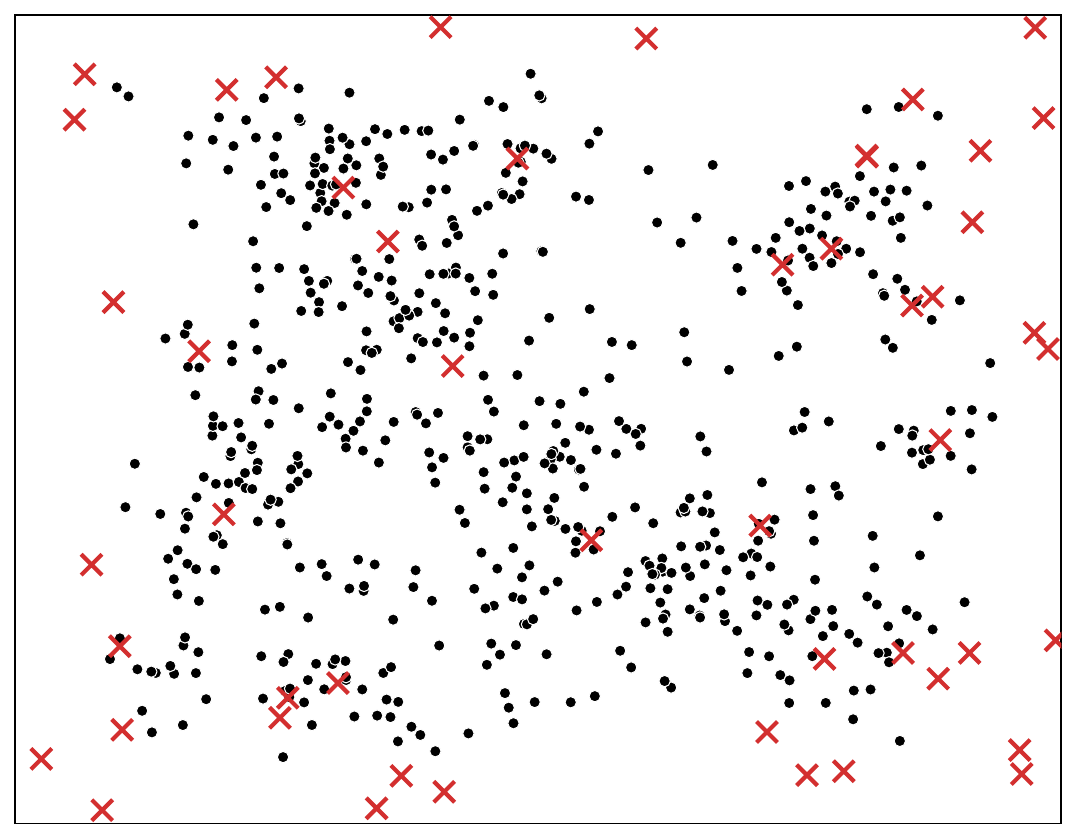} &
        \includegraphics[width=0.23\textwidth]{app_plot_k2_d3hor50_clus_T5_.pdf}\\[1mm]
    \end{tabular}
    
    \caption{Snapshots of the opinion environment with $k=2$ simulated with a localized region $d=0$, first row,  $d=3$, second row. \textcolor{red}{$\times$} denotes the creator, $\bullet$ the users respectively.}
    \label{fig:dynamics_k2}
\end{figure*}

\begin{figure*}[htbp]
    \centering
    \setlength{\tabcolsep}{2pt} % Minimal spacing between columns
    
    \begin{tabular}{c@{\hspace{0.1cm}}c@{\hspace{0.1cm}}c@{\hspace{0.1cm}}c}
        % Header row
        $t = 5$ & $t = 20$ & $t = 50$ & $t=500$\\[1mm]
        
        % First row with images
        \includegraphics[width=0.23\textwidth]{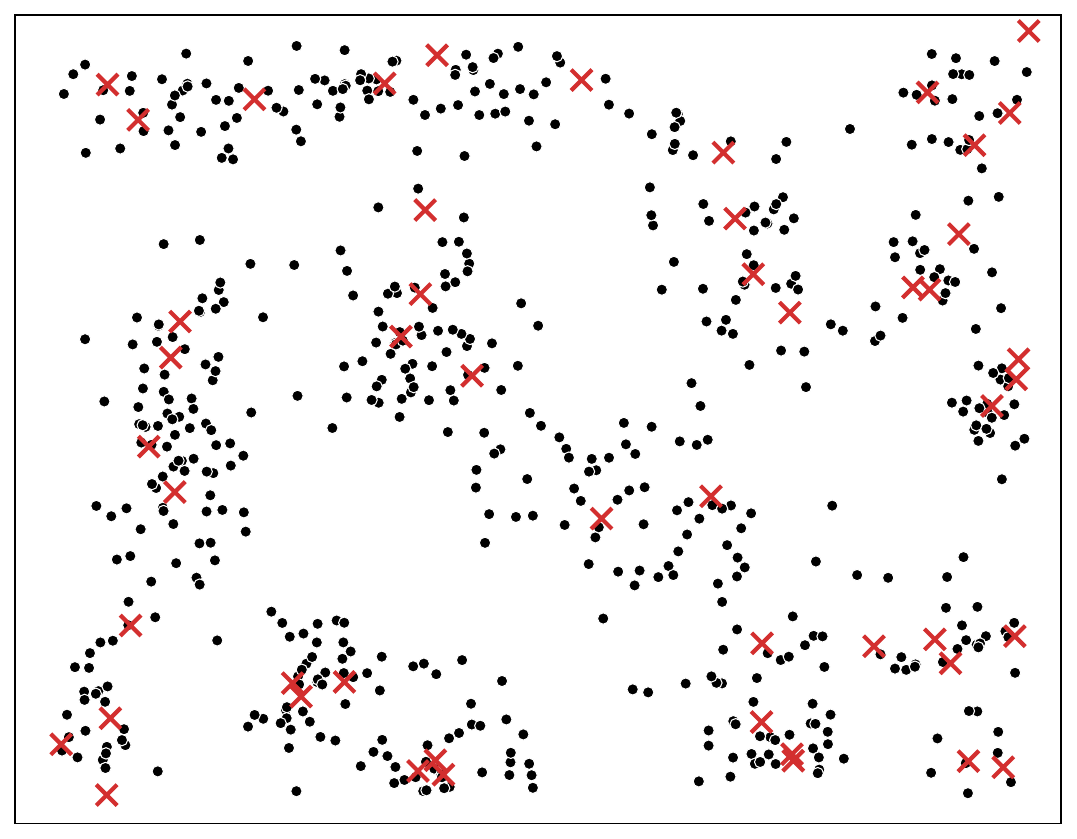} &
        \includegraphics[width=0.23\textwidth]{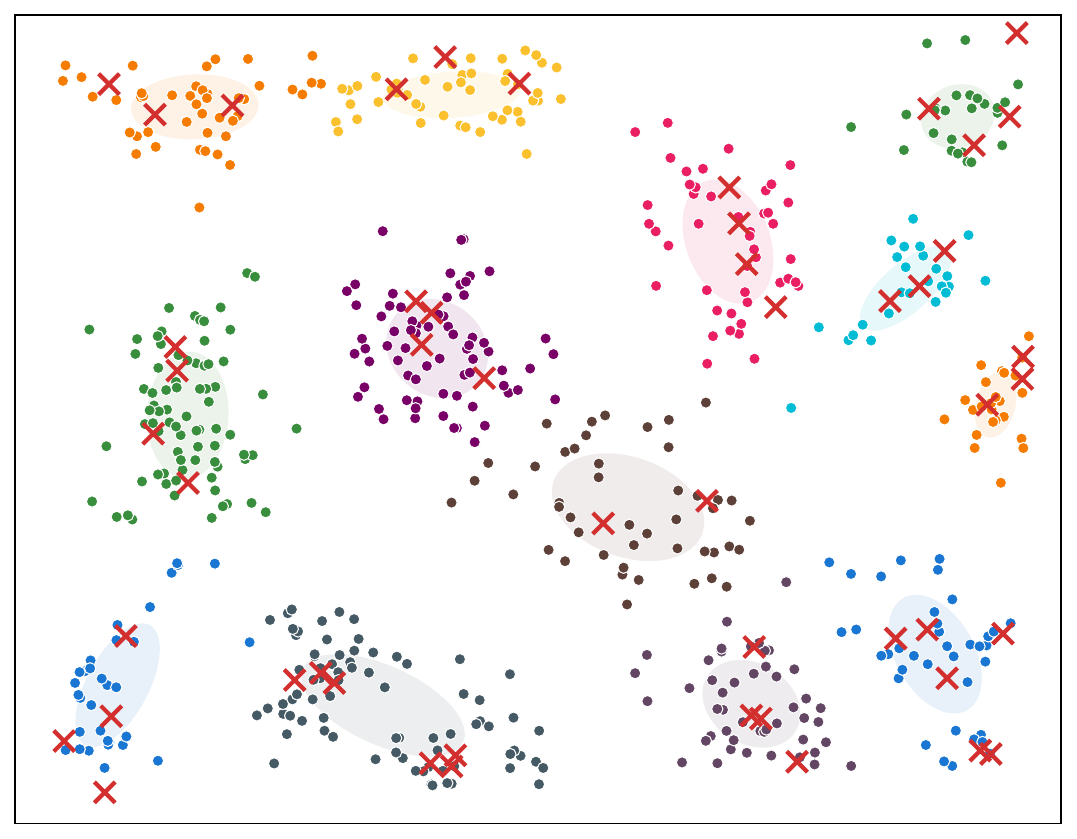} &
        \includegraphics[width=0.23\textwidth]{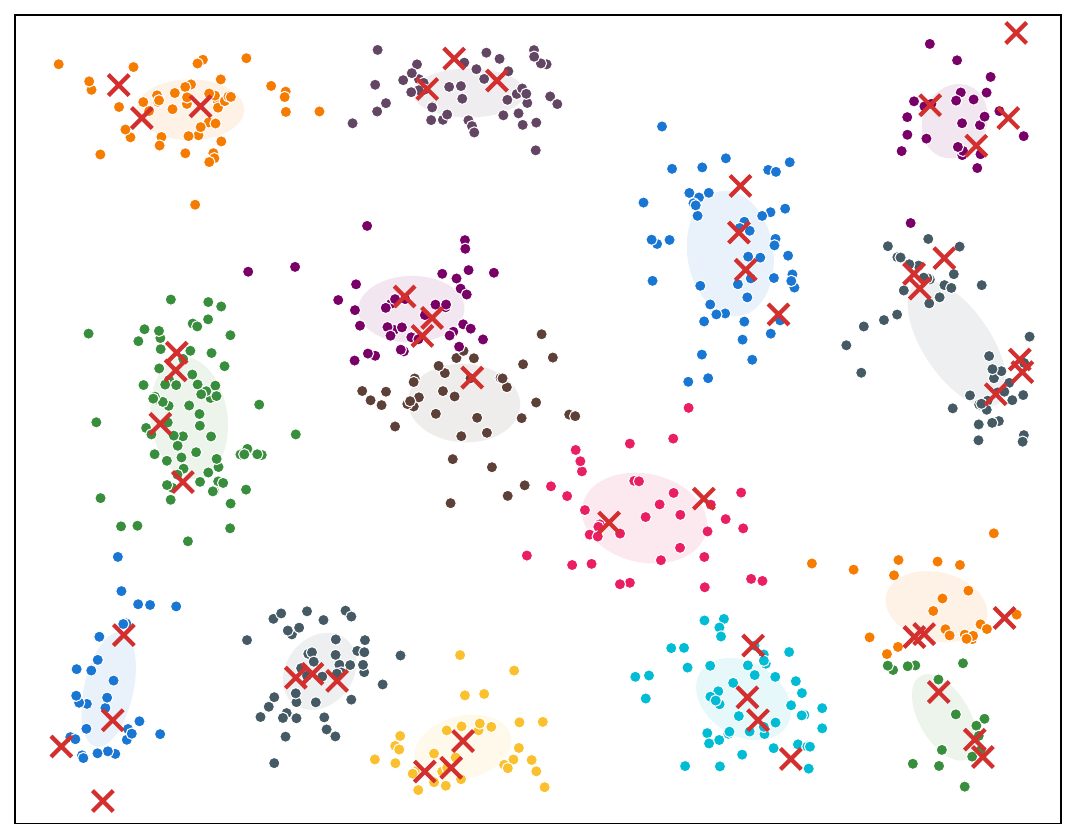} &
        \includegraphics[width=0.23\textwidth]{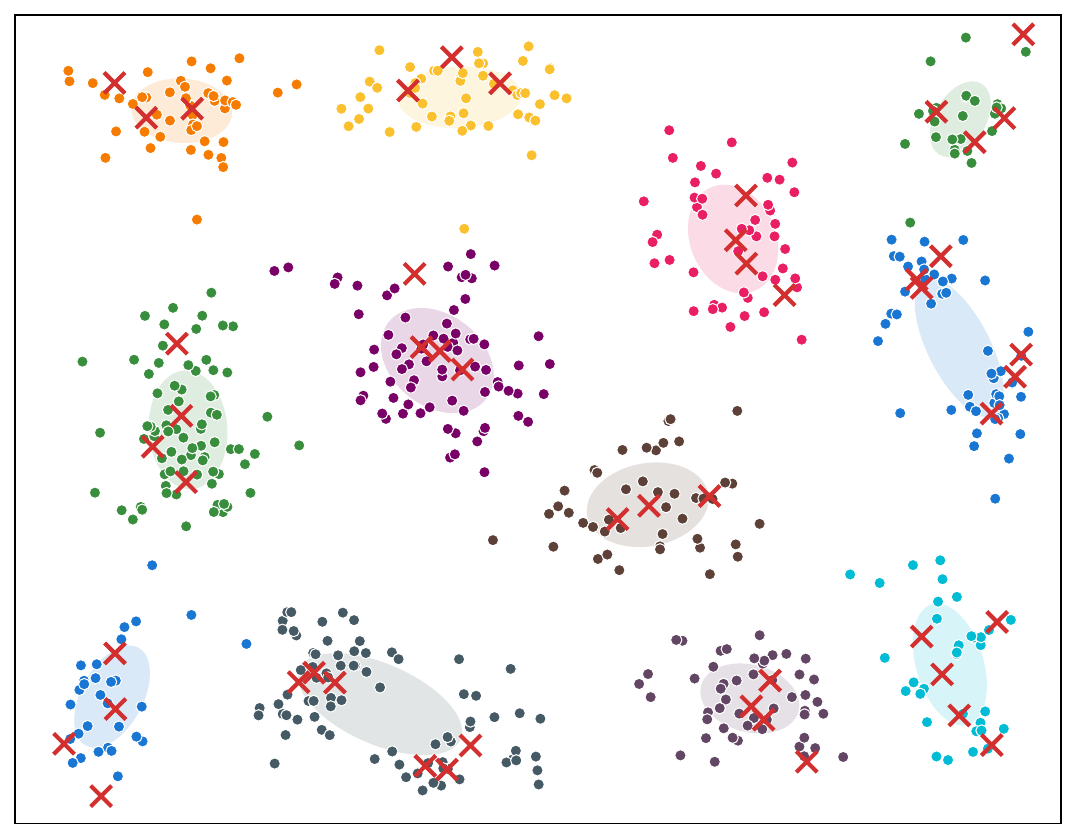} \\[1mm]
    \end{tabular}
    \begin{tabular}{c@{\hspace{0.1cm}}c@{\hspace{0.1cm}}c@{\hspace{0.1cm}}c}
        % Header row
        
        % First row with images
        \includegraphics[width=0.23\textwidth]{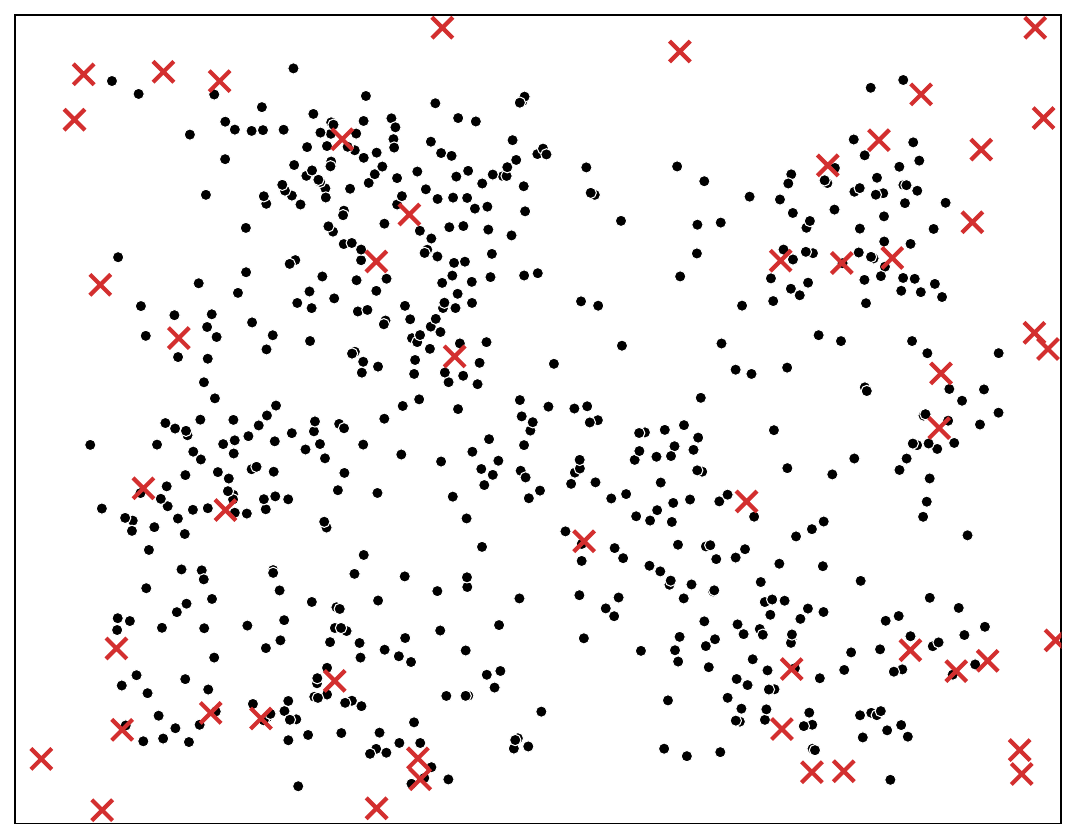} &
        \includegraphics[width=0.23\textwidth]{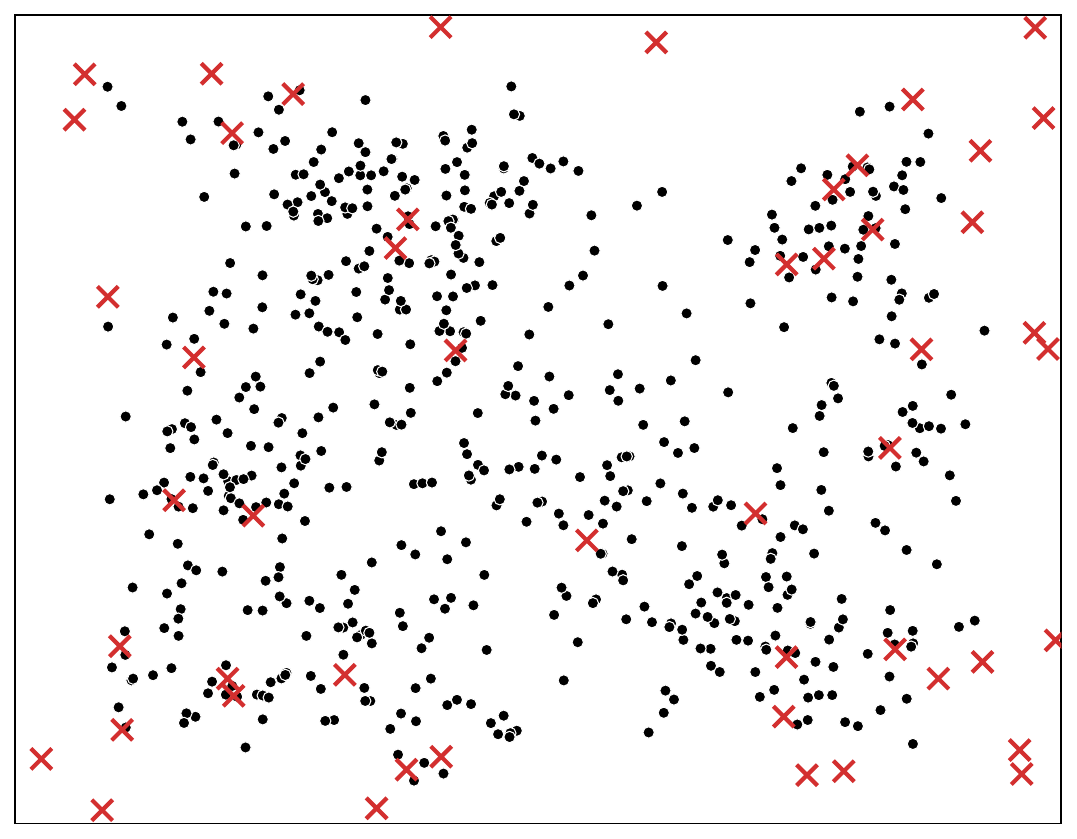} &
        \includegraphics[width=0.23\textwidth]{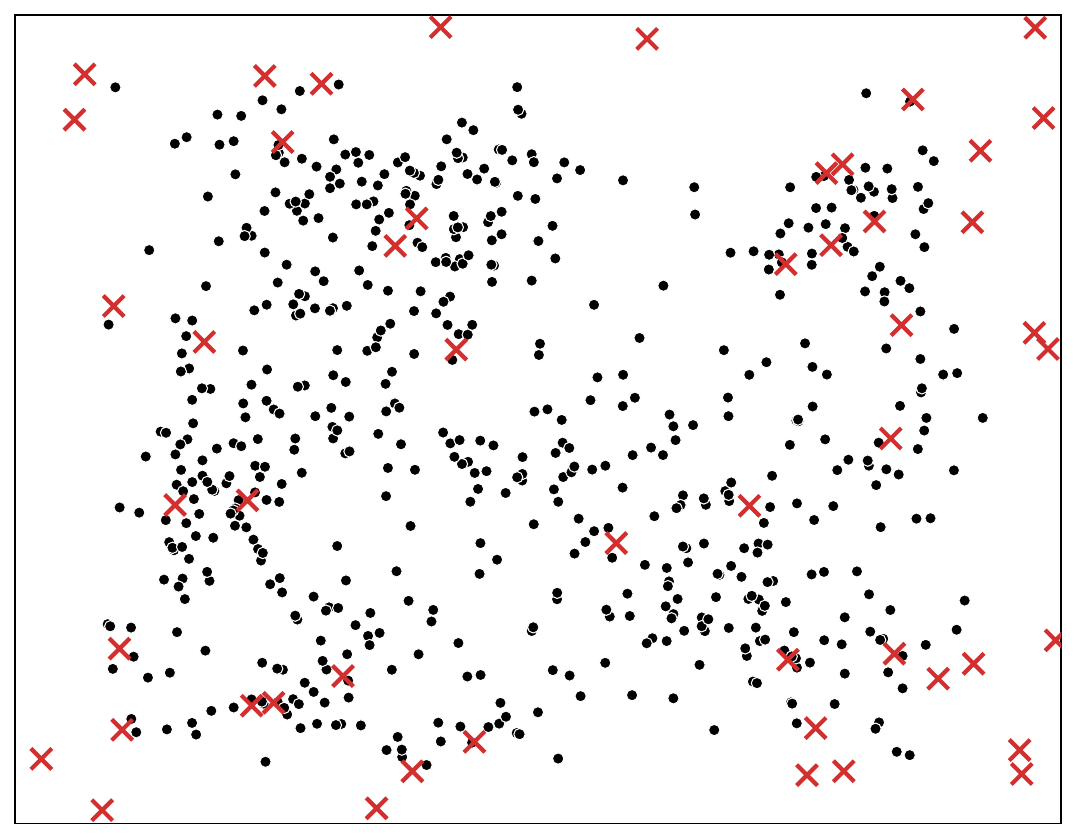} &
        \includegraphics[width=0.23\textwidth]{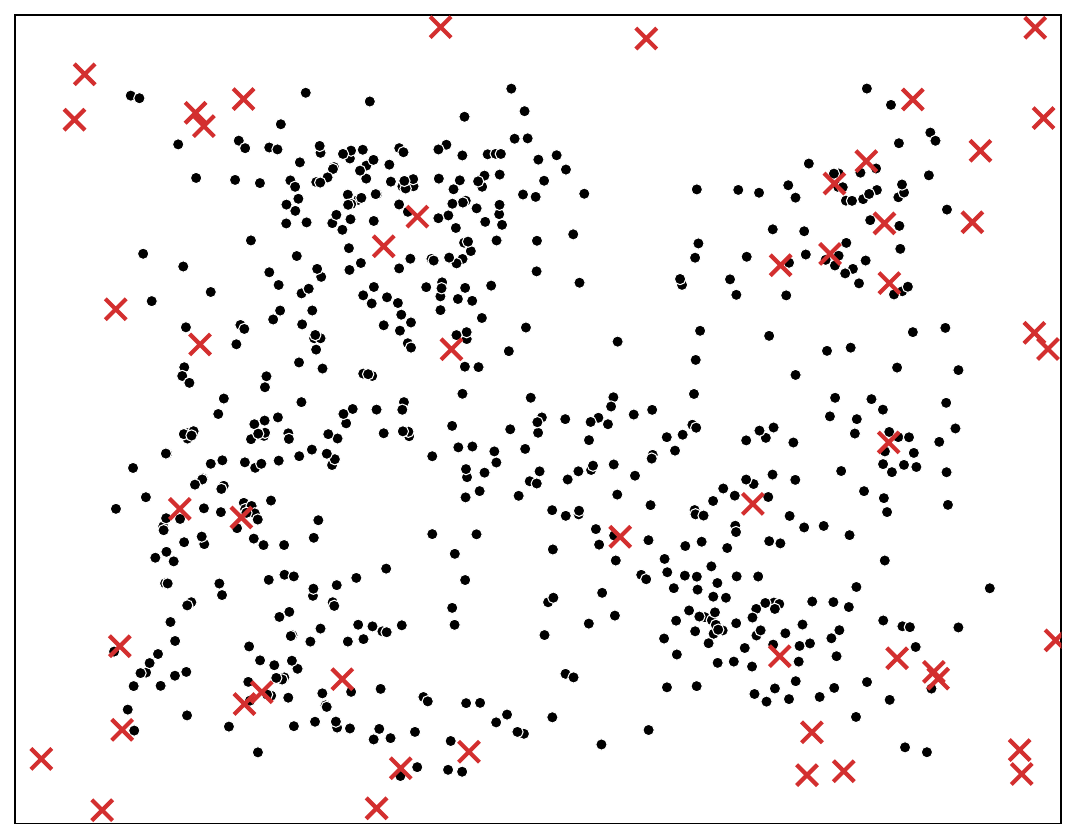} \\[1mm]
    \end{tabular}
    
    \caption{Snapshots of the opinion environment with $k=3$ simulated with a localized region $d=0$, first row,  $d=3$, second row. \textcolor{red}{$\times$} denotes the creator, $\bullet$ the users respectively.}
    \label{fig:dynamics_k3}
\end{figure*}

\begin{figure*}[htbp]
    \centering
    \setlength{\tabcolsep}{2pt} % Minimal spacing between columns
    
    \begin{tabular}{c@{\hspace{0.1cm}}c@{\hspace{0.1cm}}c@{\hspace{0.1cm}}c}
        % Header row
        $t = 5$ & $t = 20$ & $t = 50$ & $t=500$\\[1mm]
        
        % First row with images
        \includegraphics[width=0.23\textwidth]{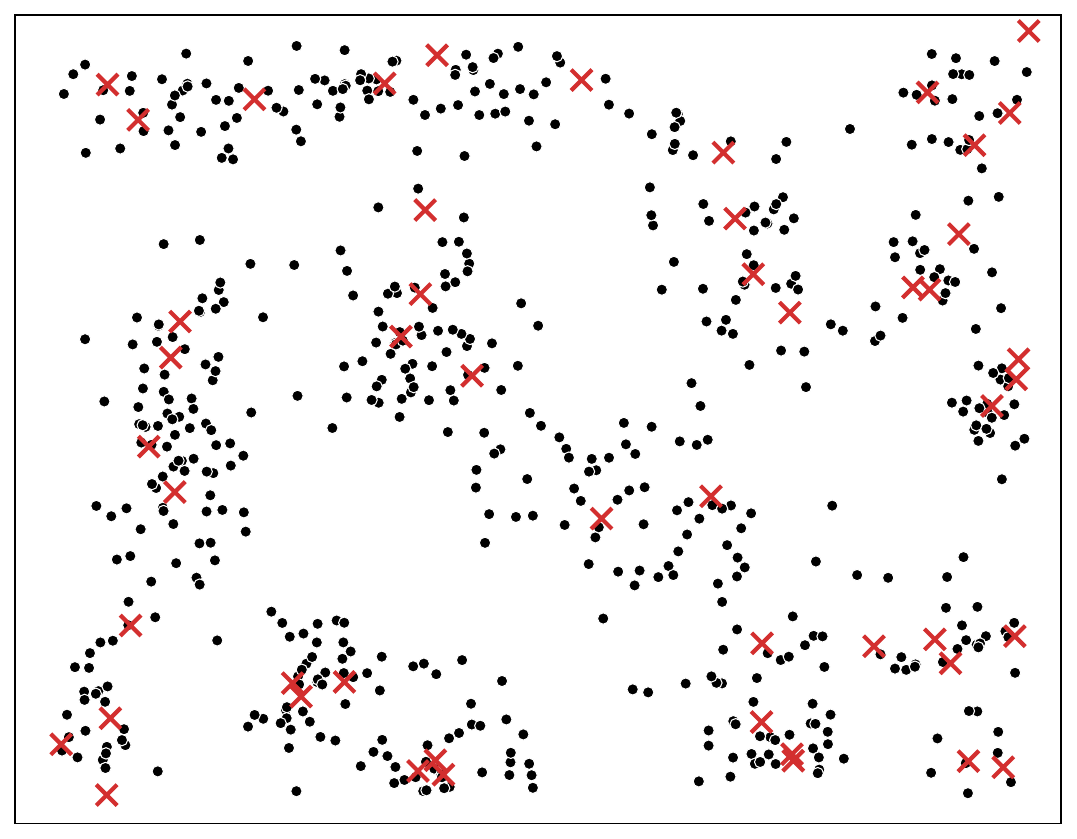} &
        \includegraphics[width=0.23\textwidth]{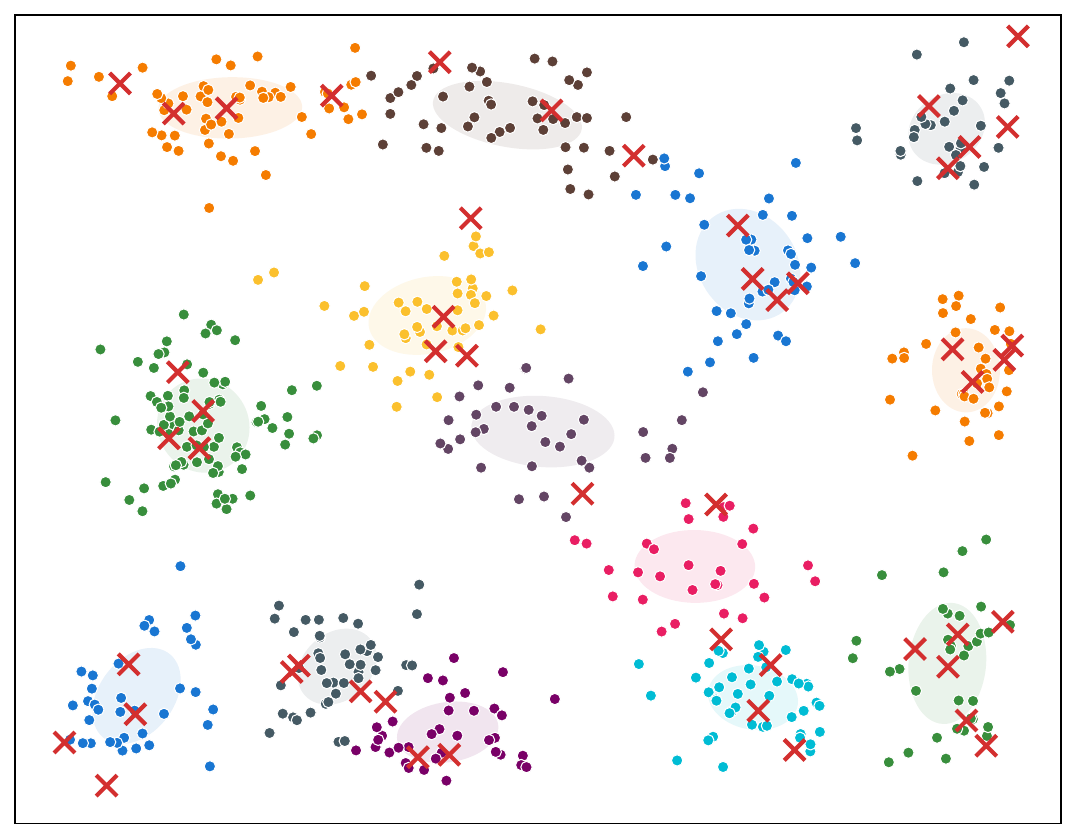} &
        \includegraphics[width=0.23\textwidth]{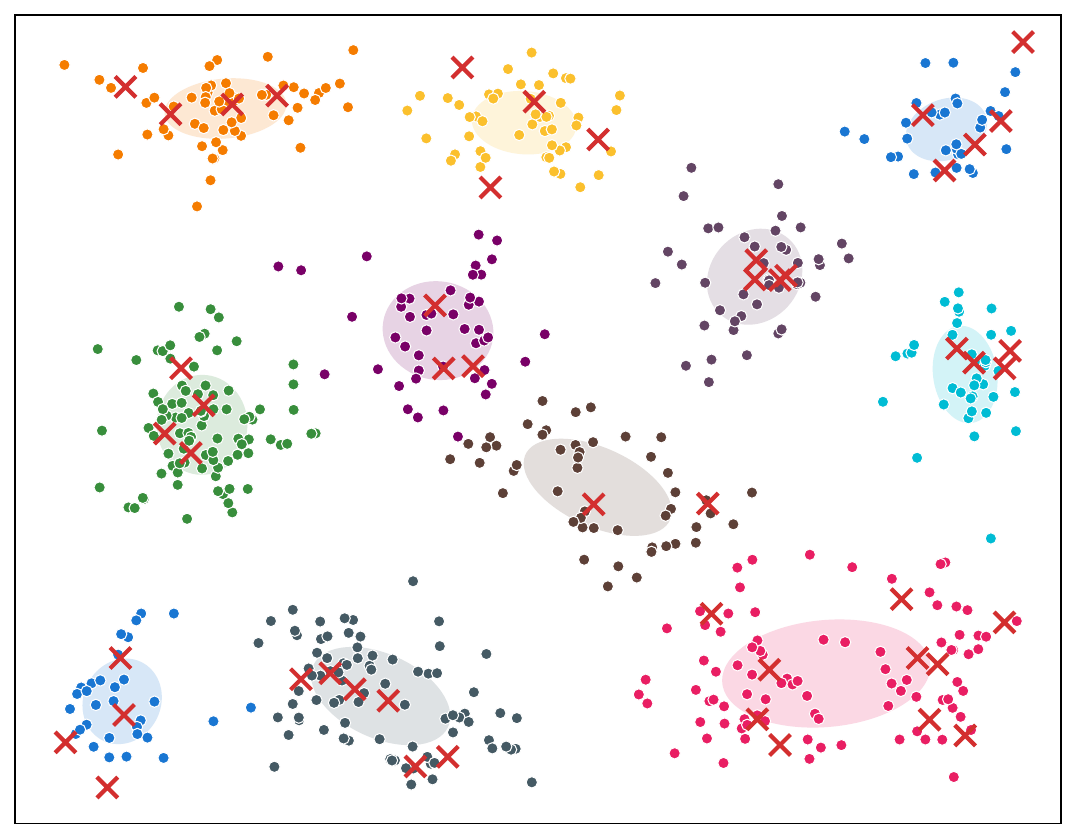} &
        \includegraphics[width=0.23\textwidth]{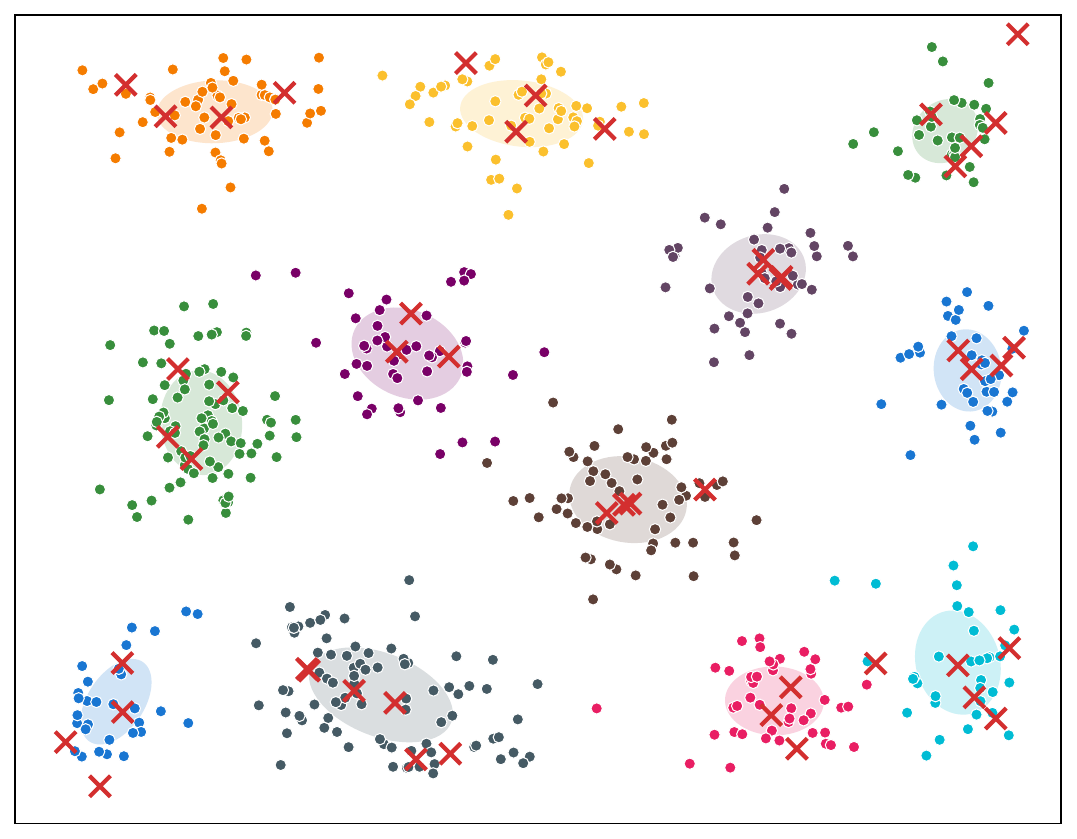} \\[1mm]
    \end{tabular}
    \begin{tabular}{c@{\hspace{0.1cm}}c@{\hspace{0.1cm}}c@{\hspace{0.1cm}}c}
        % Header row
        
        % First row with images
        \includegraphics[width=0.23\textwidth]{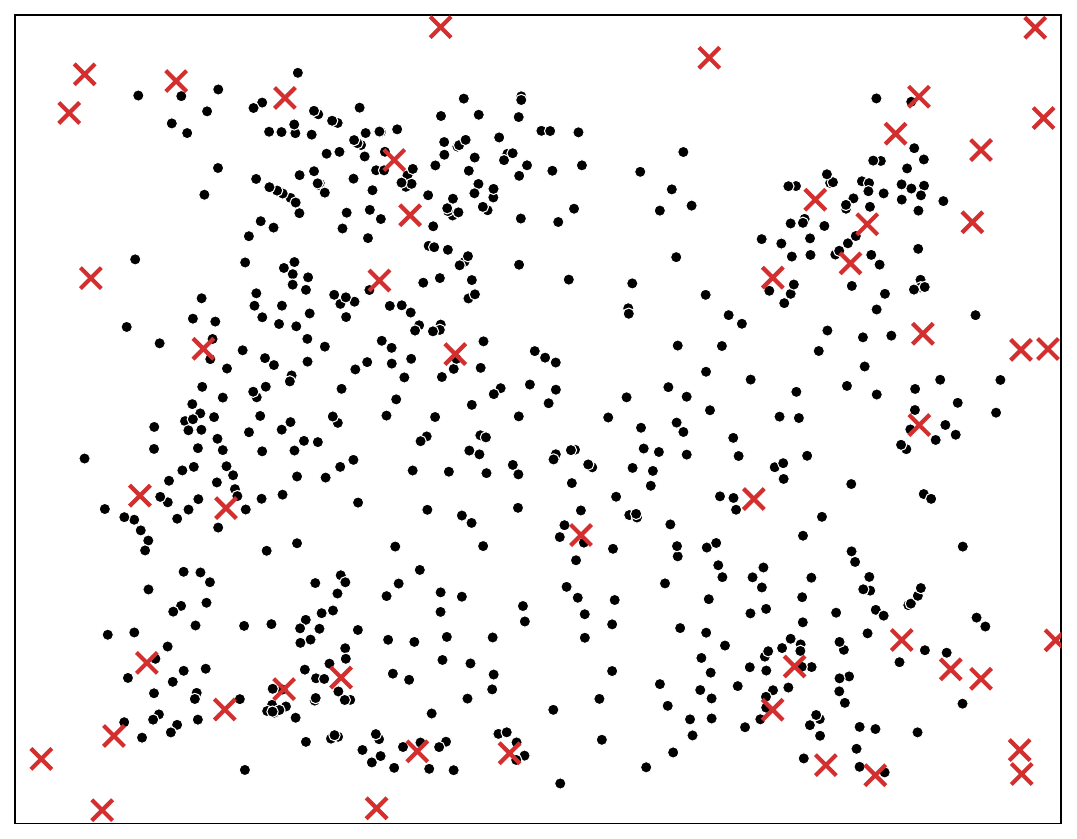} &
        \includegraphics[width=0.23\textwidth]{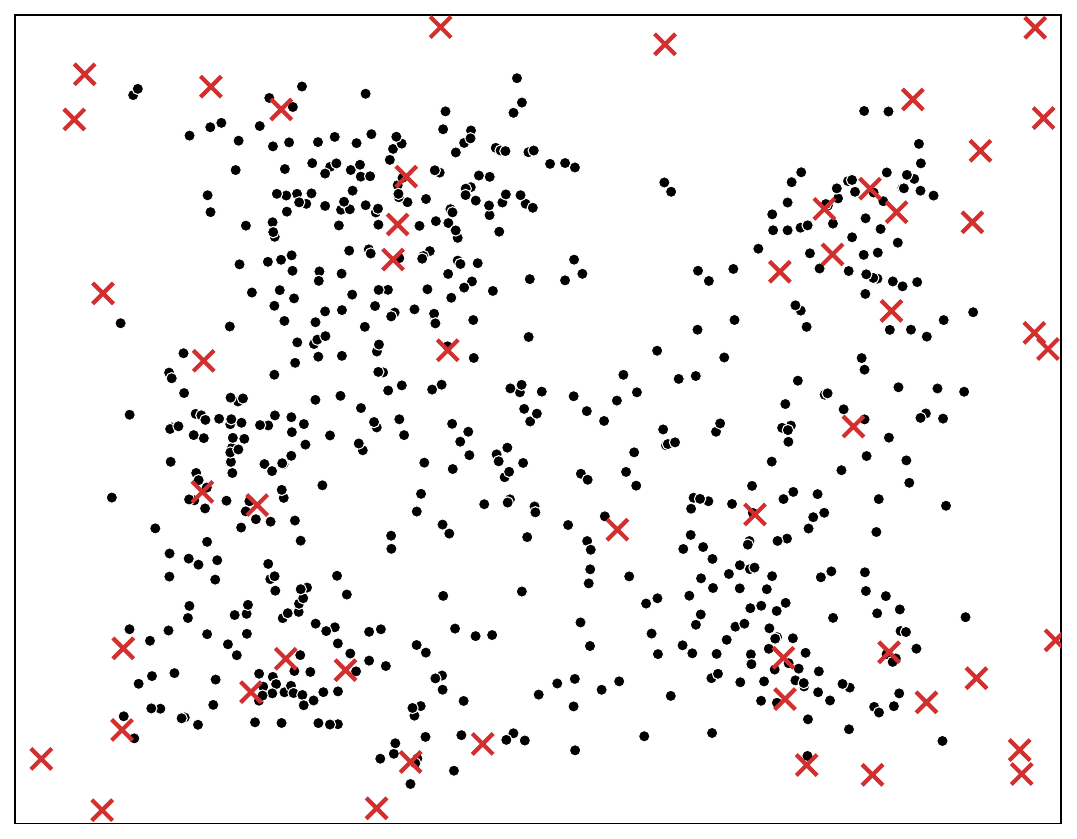} &
        \includegraphics[width=0.23\textwidth]{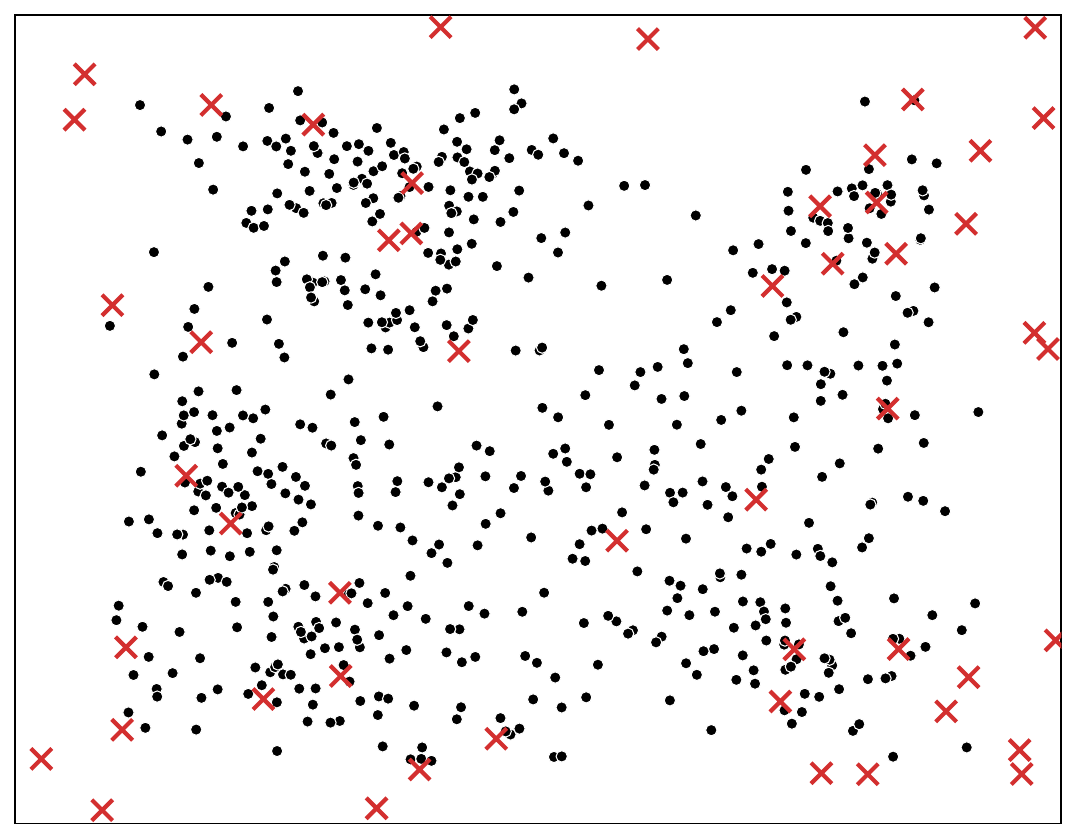} &
        \includegraphics[width=0.23\textwidth]{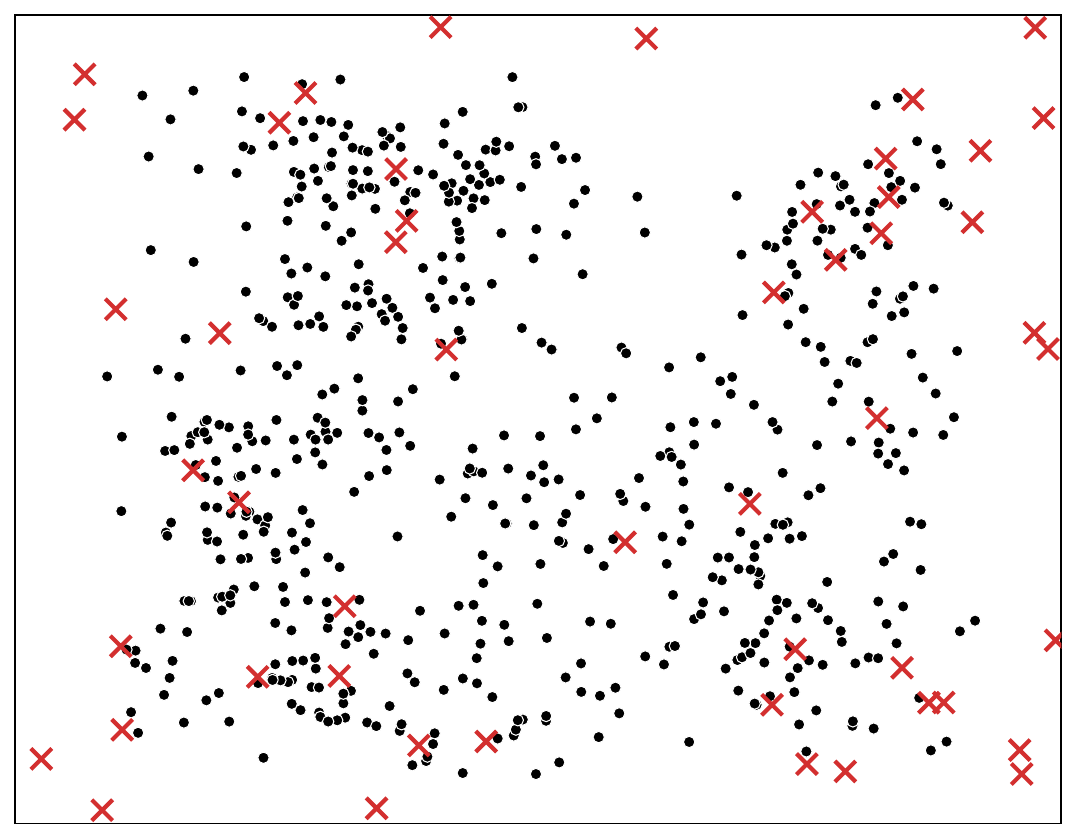} \\[1mm]
    \end{tabular}
    
    \caption{Snapshots of the opinion environment with $k=4$ simulated with a localized region $d=0$, first row,  $d=3$, second row. \textcolor{red}{$\times$} denotes the creator, $\bullet$ the users respectively.}
    \label{fig:dynamics_k4}
\end{figure*}

\section{Environment Variables for Real Dataset}
\subsection{Parameters for User-Creator Dynamics}\label{app:real_data_params}
The parameters governing the dynamics in \eqref{eqn:dynamics} are sampled independently from uniform distributions with bounds given in Table~\ref{tab:simulation_parameters2}.
\begin{table}[htbp]
\centering
\caption{Simulation Parameters for Uniform Distribution Sampling}
\label{tab:simulation_parameters2}
\begin{tabular}{lcc}
\toprule
\textbf{Parameter} & \textbf{Lower Bound} & \textbf{Upper Bound} \\
\midrule
\multicolumn{3}{l}{\textit{User Parameters}} \\
User Stubbornness $\Lambda_i$ & 0.0 & 0.5 \\
User Self-Influence $A_{ii}$& 0.5 & 0.8 \\
Recommender Influence $B_{ij}$ & 0.2 & 0.8 \\
Neighbor Influence $\sum _{j=0}^{N-1}A_{ij}$ & 0.25 & 0.5 \\
\midrule
\multicolumn{3}{l}{\textit{Creator Parameters}} \\
Creator Stubbornness  $\Gamma_j$ & 0.0 & 0.5 \\
Creator Self-Influence $E_j$& 0.5 & 0.8 \\
User-Creator Influence $C_{j}$& 0.2 & 0.8 \\
\bottomrule
\end{tabular}
\end{table}

The user-creator influence is evenly distributed among the audience set of creator $j$; specifically, for creator $j$ with audience set $\mathcal{F}_j$, each user $i \in \mathcal{F}_j$ exerts influence $C_{ji} = C_j / |\mathcal{F}_j|$. The overall social influence on user $i$ is determined by summing the influences from all neighbors. Each user is influenced by exactly one creator. Thus, referring to the FJ model in \cref{app:multi_topic}, we obtain the stochastic constraints: $A_{ii} + B_{ij} + \sum_{j =0}^{N-1} A_{ij} = 1$ for users and $C_j + E_j = 1$ for creators. 

\subsection{Ego-Facebook Dataset}\label{app:Ego_facebook}
\textcolor{black}{
The network comprises 4039 anonymous users and their social connections. The resulting graph has an average degree of $45$ with a number of neighbors reaching from $2$ to $1046$. We identify 34 community centers $\{C_1, C_2,...,C_{34}\}$ that are randomly dispersed in $[0,1]^3$ and apply spectral clustering to assign each user to one of the specified communities. After assignment, any user $i$, assigned to community center $j$, is initialized with $u_i^0 = C_j + \epsilon_i$, with $\epsilon_i \sim \mathcal{N}(0, 0.15)$.}

% \section{Results for variational environment on Real Dataset}\label{app:variational_k_real_dataset}

\newpage
\section{Compute Resources}
All simulations and experiments  were conducted on a MacBook Air equipped with an Apple M2 chip and 8 GB of unified memory, running macOS 15.6.1. 

\section{Theorem and Lemma Proofs}
 
\subsection{Proof of Theorem 1}\label{app:proof_th1}
\begin{proof}
    The proof follows by the extended Friedkin-Johnsen dynamics with
    \begin{equation}
        \begin{bmatrix}
        u^{t+1}\\
        c^{t+1}
        \end{bmatrix}
        = \begin{bmatrix}
        I-\Lambda & 0 \\
        0 & I-\Gamma
        \end{bmatrix}\underbrace{
        \begin{bmatrix}
        A & B \\
        C & E
        \end{bmatrix}}_{\Pi}\begin{bmatrix}
            u^t\\
            c^t
        \end{bmatrix} + \begin{bmatrix}
            \Lambda & 0 \\
            0 & \Gamma
        \end{bmatrix} \begin{bmatrix}
            u^0\\
            c^0
        \end{bmatrix}
    \end{equation}
    with $\Pi\in \mathbb{R}^{(N+M)\times(N+M)}$ row stochastic. 
        From \cite[Theorem 21]{opiniondyn_tutorial} we get that in the limit for $t\to \infty$ we have
        \begin{equation}
             \begin{bmatrix}
        u^{\infty}\\
        c^{\infty}
        \end{bmatrix}
         = \Bigg(\begin{bmatrix}
        I & 0 \\
        0 & I
        \end{bmatrix}-\underbrace{
        \begin{bmatrix}
        I-\Lambda & 0 \\
        0 & I-\Gamma
        \end{bmatrix}
        \begin{bmatrix}
        A & B \\
        C & E
        \end{bmatrix}}_{J}\Bigg)^{-1}
        \begin{bmatrix}
        \Lambda u^{0}\\
        \Gamma c^{0}
        \end{bmatrix}
        \end{equation}
        and from the series expansion of $(I-J)^{-1}\approx \sum_{k=0}^\infty J^k$, one gets that $u^*=\Lambda u(0)+(I-\Lambda)(Au(0)+ Bc(0))+ ((I-\Lambda)A)^2u(0)+ (I-\Lambda)B(I-\Gamma)Cu(0) + (I-\Lambda)A(I-\Lambda)Bc(0)+(I-\Lambda)B(I-\Gamma)Ec(0) + h.o.t.$. We put the focus on the first order terms $(I-\Lambda)(Au(0)+ Bc(0))$ and notice that $[A, B] \vec{1}_{N+M}=\vec{1}_N$, which in particular entails \textcolor{black}{the equality constraint: $B_i + \sum_{j}A_{i,j} = 1$. Thus, by strengthening the influence of of the social network for user $i$, characterized by $\sum_{j}A_{i,j}$, the influence of the recommended content, characterized by $B_i$ needs to go lower.}
        % with $A,B$ taking entries in $[0,1]$,
        % by strengthening the influence of the social network ($A$) the influence of the recommended content ($B$) needs to go lower. 
\end{proof}
 \subsection{Proof of Lemma 1}\label{app:proof_lemma1}
 \textcolor{black}{
 \begin{proof}
 We note that under a greedy RS for each user $i$, we have $r(u_i^t) = \text{min}_j(||c_j^t-u_i^t||_2)$, and thus $u_i^t \in \mathcal{F}_j^t$. Furthermore, given $\Gamma = I_M$, we have $c_j^0=c_j^t$, $\forall t$ and we simply write $c_j$.\newline
 \textbf{Induction Hypothesis} Assume for any user $i$, with opinion $u_i^t$, governed by the system dynamics in \eqref{eqn:dynamics}, that $u_i^t,u_i^{t-1} \in \mathcal{F}_j^t \times \mathcal{F}_j ^{t-1}$ with $(u_i^t - c_j) = \alpha^{t} (u_i^{t-1} - c_j)$, $\alpha^{t} \in [\eta,1]$, where $\eta = (1-\Lambda_i)A_{ii} \in [0,1]$. That is, $(u_i^t - c_j) $ and $ (u_i^{t-1} - c_j)$ are parallel, point in the same direction, and $||u_i^t - c_j|| \leq ||u_i^{t-1} - c_j||$. \newline
\textbf{Induction Step} Using \eqref{eqn:dynamics} we can write
 \begin{equation} 
 \begin{aligned}\label{eq:proof_1_1}
     &u_i^{t}  = (1-\Lambda_{i})(A_{ii} u_i^{t-1} + B_i c_j) +\Lambda_i u_i^0 \\
     &u_i^{t+1}  = (1-\Lambda_{i})(A_{ii} u_i^t + B_i c_j) +\Lambda_i u_i^0. \\
 \end{aligned}
 \end{equation}
 Where we used the fact that matrix $A$ is diagonal and that $u_i^t,u_i^{t-1} \in \mathcal{F}_j^t \times \mathcal{F}_j ^{t-1}$. For notational simplicity, we omit the Kronecker product in \eqref{eq:proof_1_1}. Using $\Lambda_i u_i^0= u_i^t - (1-\Lambda_{i})(A_{ii} u_i^{t-1} + B_i c_j)$, we get
 \begin{equation*}
 \begin{aligned}
     &u_i^{t+1} = -(1-\Lambda_i)A_{ii}(u_i^{t-1}-u_i^t) + u_i^t \Rightarrow\\
     & (u_i^{t+1} - c_j) = (u_i^t -c_j)-(1-\Lambda_i)A_{ii}(u_i^{t-1}- u_i^t).
\end{aligned}
 \end{equation*}
Intuitively, this states that the user opinions $u_i^{t-1}, u_i^{t}, u_i^{t+1}$, and the creators opinion $c_j$ lie on straight line. Now let $\eta = (1-\Lambda_i)A_{ii} \in [0,1]$, we can write
\begin{equation*}
\begin{aligned}
    u_i^{t+1} - c_j &= (u_i^t -c_j)-\eta(u_i^{t-1}-u_i^t) \\
    &= (u_i^t -c_j)-\eta((u_i^{t-1} -c_j)-(u_i^t-c_j)) \\
    &=(1+\eta)(u_i^t -c_j)-\eta(u_i^{t-1}-c_j) \\
    &= (1+\eta)(u_i^t -c_j)-\eta/\alpha^{t}(u_i^t - c_j)\\
    &= (1+\eta-\eta/\alpha^{t}) (u_i^{t}-c_j). 
\end{aligned}
\end{equation*}
The fourth equality follows from the induction hypothesis, namely: $(u_i^t - c_j) = \alpha^{t} (u_i^{t-1} - c_j)$. Now we observe that
\begin{equation*}
    \alpha^t \in [\eta,1] \Rightarrow (1+\eta-\eta/\alpha^{t}) \in [\eta, 1].
\end{equation*}
Now let $\alpha^{t+1} = (1+\eta-\eta/\alpha^{t})$, which leads to the desired property: $u_i^{t+1} - c_j = \alpha^{t+1} (u_i^t-c_j)$, with $\alpha^{t+1} \in [\eta,1]$, where $\eta = (1-\Lambda_i)A_{ii}$. We can use this property to further deduce that $||u_i^{t+1} - c_j ||_2 \leq ||u_i^t-c_j||_2$. Because all other creators are stubborn as well, this directly implies that under the greedy RS: $u_i^{t+1} \in F_j^{t+1}$. \newline
\textbf{Base Case} Induction now follows by:
\begin{equation*} 
 \begin{aligned}
    &u_i^0 \in  \mathcal{F}_j^0 \Rightarrow\\
     &u_i^{1}  = (1-\Lambda_{i})(A_{ii} u_i^{0} + B_i c_j) +\Lambda_i u_i^0 \\
     &u_i^{1}  = (1-\Lambda_{i})(A_{ii} u_i^{0} + (1-A_{ii}) c_j) +\Lambda_i u_i^0 \\
     &u_i^{1}  = ((1-\Lambda_{i})A_{ii} + \Lambda_i) u_i^{0} + (1-\Lambda_{i})(1-A_{ii}) c_j \Rightarrow\\
     &u_i^{1} - c_j  = ((1-\Lambda_{i})A_{ii} + \Lambda_i) (u_i^{0}-c_j).\\
 \end{aligned}
 \end{equation*}
The proof now follows by noting that: $\eta = (1-\Lambda_{i})A_{ii} \leq ((1-\Lambda_{i})A_{ii} + \Lambda_i) = \alpha^1 \leq 1$. We conclude $u_i^{1} - c_j  = \alpha^1(u_i^{0}-c_j)$, with $\alpha^1 \in [\eta,1]$, which given stubborn users, implies $u_i^0,u_i^{1} \in \mathcal{F}_j^0 \times \mathcal{F}_j ^{1}$. Thus the user partitions $\mathcal{F}_0,..,\mathcal{F}_M$ are in fact static under the greedy RS and for any user $i \in \mathcal{F}_j$, the distance to the creator $||u_i^t-c_j||_2$, decreases monotonically with $\alpha^t$.
\end{proof}}

\textcolor{black}{
\section{Related Work}\label{sec:related_work}
The better position our paper, the following table provides a schematic summary of the related work, by classifying user and creators as Static (S) or Dynamic (D) and wheather they are seen as embedded in a Network (N) or seen as Isolated (I). For the Recommender System we distinguish if it is Fixed (F), namely taken from the literature, or Explicitly Designed (ED). 
\begin{center}
\begin{tabular}{| c  | c | c | c |}
\hline
  & \textbf{Users} & \textbf{Creators} & \textbf{Recommender System} \\ 
 \hline
Us & D,N & D, I & ED\\ 
 \hline
\cite{dual_influence} & D,I & D,I&F\\   
 \hline
\cite{closed_loop_opinion} & D,I & N/A & ED \\
 \hline
\cite{micro_macro_op_effects} & D,I & N/A & ED \\
 \hline
\cite{dean2024usercreators} & D,I & D,I & ED\\
 \hline
\cite{srs_firstpaper} & S,N & N/A & F\\
 \hline
\cite{network_aware_rec_sys_via_feedback} & D,N & N/A & ED\\
 \hline
\cite{topic_diversification} & S,I & N/A & ED \\
 \hline
\cite{learning_to_recommend} & S, I& N/A & ED\\
 \hline
\cite{diversified_recommendations} &S,I &N/A & ED \\
 \hline
\cite{avoiding_monotony} & S,I & N/A & ED\\
\hline
\end{tabular}\label{tab:related}
\end{center}
}

\end{document}